\documentclass[12pt, onecolumn]{IEEEtran} 
\usepackage{url}
\usepackage{makecell}
\usepackage{amsfonts}
\usepackage{amssymb}
\usepackage{amsmath}
\usepackage{graphicx, colordvi, psfrag}
\usepackage{calc,pstricks, pgf, xcolor}
\usepackage{epsfig, cite}
\usepackage{bm} 
\usepackage{bbm} 
\usepackage{enumerate}

\newcommand{\beq}[1]{\begin{equation}\label{#1}}
\newcommand{\eeq}{\end{equation}}

\newcommand{\beqn}[1]{\begin{eqnarray}\label{#1}}
\newcommand{\eeqn}{\end{eqnarray}}

\newtheorem{thmbody}{Theorem}
\newenvironment{thm}{
\begin{thmbody}
	}{
	\end{thmbody} 
	}
\newtheorem{dfnbody}{Definition}

\newtheorem{corbody}{Corollary}

\newtheorem{lemmabody}{Lemma}
\newenvironment{lemma}{
\begin{lemmabody}
	}{
	\end{lemmabody} 
	}
\newtheorem{propbody}{Proposition}
\newenvironment{prop}{
\begin{propbody}
	}{
	\end{propbody} 
	}
\newenvironment{proof}{
	{\it Proof:}
	}{
 $\Box$
	}

\usepackage{dsfont}
\begin{document}
\title{The Method of Types for the AWGN Channel}

\author{
\IEEEauthorblockN{Sergey Tridenski and Anelia Somekh-Baruch}\\
\IEEEauthorblockA{Faculty of Engineering\\Bar-Ilan University\\ Ramat-Gan, Israel\\
Email: tridens@biu.ac.il, somekha@biu.ac.il}
}



\maketitle
\begin{abstract}
For the discrete-time AWGN channel with a power constraint,
we 
give an alternative derivation of
Shannon's
sphere-packing upper bound on the optimal block error exponent
and prove for the first time 
an analogous lower bound on the optimal correct-decoding exponent.
The derivations
use the method of types with finite alphabets of sizes
depending on the block length $n$
and with the
number of types 
sub-exponential in $n$.
\end{abstract}



%
%
%
%



\section{Introduction}\label{Int}


We study reliability 
of the discrete-time additive white Gaussian noise (AWGN) channel
with a power constraint imposed on blocks of its inputs. 
Consider the capacity of this channel, found by Shannon:
\begin{equation} \label{eqShannonCapacity}
C \; = \; \tfrac{1}{2}\,{\log\mathstrut}_{\!2} (1 + s^{2}/\sigma^{2}),
\end{equation}
where $\sigma^{2}$ is the channel noise variance and $s^{2}$ is the power constraint.
This capacity corresponds to the maximum of the mutual information $I(\,{p\mathstrut}_{X}, \, w)$ over ${p\mathstrut}_{X}$,
under the power constraint on
${p\mathstrut}_{X}$, where $w$ stands for the channel transition probability density function (PDF)
and ${p\mathstrut}_{X}$ is the channel input PDF.
Let us briefly recall the technicalities \cite{CoverThomas} of how the expression (\ref{eqShannonCapacity}) is obtained 
from the mutual information:
\begin{align}
\max_{\substack{\\{p\mathstrut}_{X}:\;\mathbb{E}[X^{2}] \, \leq \, s^{2}}}
I(\,{p\mathstrut}_{X}, \, w)
\;\; = \;\; &
\max_{\substack{\\{p\mathstrut}_{X}:\;\mathbb{E}[X^{2}] \, \leq \, s^{2}}}
\;
\Big\{
D\big(\, w \, \| \, \,{\widehat{p}\mathstrut}_{Y} \, | \, \, {p\mathstrut}_{X}\big)
\, - \,
D\big(\,{p\mathstrut}_{Y} \, \| \, \,{\widehat{p}\mathstrut}_{Y}\big)
\Big\}
\nonumber \\
= \;\; &
\max_{\substack{\\{p\mathstrut}_{X}:\;\mathbb{E}[X^{2}] \, \leq \, s^{2}}}
\;
\bigg\{
\frac{1}{2}\,{\log\mathstrut}_{\!2}\left(1 +\frac{s^{2}}{\sigma^{2}}\right)
\, + \,
\underbrace{\frac{\mathbb{E}\big[X^{2}\big] - s^{2}}{2\ln(2)(s^{2} + \sigma^{2})}
\, - \,
D\big(\,{p\mathstrut}_{Y} \, \| \, \,{\widehat{p}\mathstrut}_{Y}\big)}_{\leq \; 0}
\bigg\}.
\label{eqGaussOpt}
\end{align}
Here
${\widehat{p}\mathstrut}_{Y}(y) \triangleq \frac{1}{\sqrt{2\pi(s^{2} \,+\, \sigma^{2})}}e^{-\frac{y^{2}}{2(s^{2} \,+\, \sigma^{2})}}$
and ${p\mathstrut}_{Y}(y) \equiv \int_{\mathbb{R}}{p\mathstrut}_{X}(x)w(y\,|\,x)dx$,
the operator $\mathbb{E}[\,\cdot\,]$ denotes the expectation,
and $D$ is the Kullback–Leibler divergence between two probability densities.
In this paper we consider the optimal exponents in the block error/correct-decoding probability
of the AWGN channel.
We propose explanations, similar to (\ref{eqGaussOpt}), both
for Shannon's sphere-packing converse bound on the optimal error exponent \cite[Eq.~3,~4,~11]{Shannon59}
and for a similar converse bound on the optimal correct-decoding exponent,
an expression for which has been given by Oohama \cite[Eq.~22]{Oohama17} without proof.

In the case of discrete memoryless channels,
the mutual information enters into the expressions for correct-decoding and error exponents through the method of types \cite{CsiszarKorner},
\cite{DueckKorner79}.
For 
the moment
without any 
interpretation,
let us rewrite the sphere-packing
constant-composition 
exponent 
\cite[Eq.~5.19]{CsiszarKorner}
with PDF's:
\begin{equation} \label{eqLagrangeMultipliers}
\min_{\substack{\\{p\mathstrut}_{Y|X}: \;\;
 I(\,{\widehat{p}\mathstrut}_{X}, \,\, {p\mathstrut}_{Y|X}) \; \leq \; R
}}
D\big(\,{p\mathstrut}_{Y|X}\,\|\, \, w \, \,|\, \,{\widehat{p}\mathstrut}_{X}\big),
\end{equation}
where ${\widehat{p}\mathstrut}_{X}$ 
denotes
the Gaussian density with variance $s^{2}$ which maximizes (\ref{eqGaussOpt}),
and $R>0$ is
the information rate.
When ${\widehat{p}\mathstrut}_{X}$ is Gaussian,
the minimum (\ref{eqLagrangeMultipliers}) allows an explicit solution by the method of Lagrange multipliers.
The minimizing solution ${p\mathstrut}_{Y|X}^{*}$ 
of (\ref{eqLagrangeMultipliers}) 
is Gaussian,
and we obtain that (\ref{eqLagrangeMultipliers}) is the same as
Shannon's 
converse bound on the error exponent \cite[Eq.~3,~4,~11]{Shannon59}
in the limit of a large block length:
\begin{align}
E_{sp}(R) \; & = \; 
\big[\,
\tfrac{1}{2}A^{2}
\, - \,
\tfrac{1}{2} AG\cos \theta \, - \, \ln(G \sin \theta)
\,\big] \,{\log\mathstrut}_{\!2} \,e , 
\label{eqShannonExponent} \\
G \; & = \; \tfrac{1}{2}
(
A \cos \theta \, + \, \sqrt{A^{2}\cos^{2}\theta + 4}
),
\;\;\;
\sin \theta \, = \, 2^{\,-\min\,\{R, \,C\}},
\;\;\;
A \, = \, s/\sigma.
\nonumber
\end{align}
Then it turns out that ${p\mathstrut}_{Y|X}^{*}$ of (\ref{eqLagrangeMultipliers}) and the $y$-marginal PDF
of the product ${\widehat{p}\mathstrut}_{X}{p\mathstrut}_{Y|X}^{*}$
play the same roles in the derivation of the converse bound,
as $w$ and ${\widehat{p}\mathstrut}_{Y}$, respectively, in the maximization (\ref{eqGaussOpt}).

In this paper, in order to 
derive expressions similar to (\ref{eqLagrangeMultipliers}), we extend the method of types \cite[Ch.~11.1]{CoverThomas},
\cite{Csiszar98}
to include
countable alphabets consisting of real numbers, with the help of power constraints on types.
The countable alphabets depend on the block length $n$ and the number of types
satisfying the power constraints
is kept sub-exponential in $n$.
The latter idea is inspired by a different subject --- of ``runs'' in a binary sequence.
If we treat every ``run'' of ones or zeros in a binary sequence as a separate symbol
from the countable alphabet of run-lengths,
then the number of different empirical distributions of such symbols in a binary sequence of length $n$
is equivalent to the number of different partitions of the integer $n$
into sum of positive integers, which is $\sim e^{c\sqrt{n}}$
\cite[Eq.~5.1.2]{Andrews}.
Thus it is sub-exponential, and the method of types can be 
extended to
that case.
In our present case, however, the types are empirical distributions
of
uniformly quantized real numbers
in quantized versions of real channel input and output vectors
of length $n$. The quantized versions serve 
only for classification of channel input and output vectors and not for 
the communication itself.
The uniform quantization step is 
different for the quantized versions of channel inputs and outputs,
and 
in both cases it
is chosen to be a decreasing function of $n$.

Similarly as (\ref{eqGaussOpt}),
the proposed derivations demonstrate, that, in order to achieve the converse bounds
on the correct-decoding and error exponents,
it is necessary for
the types of the quantized versions of codewords to converge to the Gaussian distribution
in characteristic function (CF), or, equivalently, in
cumulative distribution function (CDF).

The contributions of the current paper are twofold.
Firstly, we successfully apply the method of types 
to derive
converse bounds on the exponents of the AWGN channel.
Secondly, we 
prove the converse bound on the correct-decoding exponent \cite[Eq.~22]{Oohama17}
for the first time.
This underscores the advantage of the method of types.

In Sections~\ref{ComSys} and~\ref{Defs}
we describe the communication system and introduce other definitions.
In Section~\ref{Main} we present the main results of the paper,
which consist of two theorems and a proposition.
Section~\ref{MOT} provides an extension to the method of types.
In Section~\ref{ConvLemma} we prove a converse lemma,
which is then used 
for derivation of both the correct-decoding and error exponents
in Sections~\ref{CorDecExp} and~\ref{ErrExp}, respectively.
Section~\ref{PDFtypePDF} 
connects between PDF's and types.

\section*{Notation}\label{Not}

Countable alphabets consisting of real numbers
are denoted by ${\cal X}_{n}$, ${\cal Y}_{n}$.
The set of types with denominator $n$ over 
${\cal X}_{n}$
is denoted by ${\cal P}_{n}({\cal X}_{n})$.
Capital `$P\,$' denotes probability mass functions, which are types:
${P\mathstrut}_{\!X}$, ${P\mathstrut}_{\!Y}$, ${P\mathstrut}_{\!XY}$, ${P\mathstrut}_{\!Y|X}$.
The type class and the support of a type ${P\mathstrut}_{\!X}$ are denoted by
$T({P\mathstrut}_{\!X})$ and ${\cal S}({P\mathstrut}_{\!X})$, respectively.
The expectation w.r.t. a probability distribution ${P\mathstrut}_{\!X}$ is denoted by $\mathbb{E}_{{P\mathstrut}_{\!X}}[\,\cdot\,]$.
Small `$p$' denotes probability density functions: ${p\mathstrut}_{X}$, ${p\mathstrut}_{Y}$, ${p\mathstrut}_{XY}$,
${p\mathstrut}_{Y|X}$.
Thin letters $x$, $y$ represent real values, while thick letters ${\bf x}$, ${\bf y}$ represent real vectors. 
Capital letters $X$, $Y$ 
represent random variables, 
boldface ${\bf Y}$ 
represents a random vector of length $n$.
The conditional type class of ${P\mathstrut}_{\!X|\,Y}$ given ${\bf y}$ is denoted by $T({P\mathstrut}_{\!X|\,Y}\,|\, {\bf y})$.
The 
quantized versions of variables are denoted by a superscript `$q$': $x^{q}_{k}$, ${\bf x}^{q}$, ${\bf Y}^{q}$.
Small $w$ stands for a conditional PDF, and
${W\mathstrut}_{\!n}$ stands for a discrete positive measure, which does not necessarily add up to $1$.
All information-theoretic quantities 
such as joint and conditional entropies
$H({P\mathstrut}_{\!XY})$, $H(Y\,|\,X)$, the mutual information
$I({P\mathstrut}_{\!XY})$,
$I\big({P\mathstrut}_{\!X}, {P\mathstrut}_{\!Y|X}\big)$,
$I\big({P\mathstrut}_{\!X}, \, {p\mathstrut}_{Y|X}\big)$, the Kullback-Leibler divergence
$D\big({P\mathstrut}_{\!Y|X}\,\|\, {W\mathstrut}_{\!n} \,|\,  {P\mathstrut}_{\!X}\big)$,
$D\big(\,{p\mathstrut}_{Y|X}\,\|\, \, w \, \,|\, {P\mathstrut}_{\!X}\big)$,
and the information rate $R$
are defined with respect to the
logarithm to a base $b>1$, denoted as
$\,{\log\mathstrut}_{\!b}(\cdot)$.
It is assumed that $0 \,{\log\mathstrut}_{\!b} (0) = 0$.
The natural logarithm is denoted as $\ln$.
The cardinality of a discrete set is denoted by $|\,\cdot\,|$, while the volume of a continuous region is denoted by $\text{vol}\,(\cdot)$.
The complementary set of a set $A$ is denoted by $A{\mathstrut}^{c}$.
Logical ``or'' and ``and'' are represented by the symbols $\lor$ and $\land$, respectively.
In Appendix B, ${p\mathstrut}_{XY}^{q}$ represents the rounded down version of the PDF ${p\mathstrut}_{XY}$.



\newpage

\section{Communication system}\label{ComSys}


We consider communication over the time-discrete 
additive white Gaussian noise
channel
with real channel inputs $x \in \mathbb{R}$ and channel outputs $y \in \mathbb{R}$
and a 
transition probability density
\begin{displaymath}
w(y \, | \, x) \;\; \triangleq \;\; \frac{1}{\sigma\sqrt{2\pi}}e^{-\frac{(y\, - \,x)^{2}}{2\sigma^{2}}}.
\end{displaymath}

Communication is performed by blocks of $n$ channel inputs. Let $R > 0$ denote a nominal information rate.
Each block is used for transmission of one out of $M$ messages,
where $M = M(n, R) \triangleq \lfloor b^{\,n R}\rfloor$, for some logarithm base $b > 1$.
The encoder is a deterministic function $f: \{1, 2, .\,.\,.\, , \, M\} \rightarrow \mathbb{R}{\mathstrut}^{n}$,
which converts a message into a transmitted block,
such that
\begin{displaymath}
f(m) \; = \; {\bf x}(m)
\; = \; \big(
x_{1}(m), \,x_{2}(m), .\,.\,.\, , \, x_{n}(m)
\big),
\;\;\;\;\;\;\;\;\; m = 1, 2, .\,.\,.\, , \, M,
\end{displaymath}
where $x_{k}(m) \in \mathbb{R}$, for all $k = 1, 2, .\,.\,.\, , \, n$.
The set of all the 
codewords ${\bf x}(m)$, $m = 1, 2, .\,.\,.\, , \, M$, constitutes a codebook ${\cal C}$.
Each codeword ${\bf x}(m)$ in ${\cal C}$ satisfies the power constraint:
\begin{equation} \label{eqPowerConstraint}
\frac{1}{n}\sum_{k \, = \, 1}^{n}x_{k}^{2}(m) \; \leq \; s^{2},
\;\;\;\;\;\;\;\;\; m = 1, 2, .\,.\,.\, , \, M.
\end{equation}
The decoder is another deterministic function
$g: \mathbb{R}{\mathstrut}^{n} \rightarrow \{0, 1, 2, .\,.\,.\, , \, M\}$,
which converts the received block of $n$ channel outputs ${\bf y} \in \mathbb{R}{\mathstrut}^{n}$
into an estimated message, or, possibly, to 
a special error symbol `$0$':
\begin{equation} \label{eqDec}
g({\bf y})
\;\; = \;\;
\Bigg\{
\begin{array}{r l}
0, & \;\;\; {\bf y} \in \bigcap_{\,m \, = \, 1}^{\,M} {\cal D}{\mathstrut}_{m}^{c}, \\
m, & \;\;\; {\bf y} \in {\cal D}{\mathstrut}_{m}, \;\;\; 
m \in \{1, 2, .\,.\,.\, , \, M\},
\end{array}
\end{equation}
where each set ${\cal D}{\mathstrut}_{m} \subseteq \mathbb{R}{\mathstrut}^{n}$ is either an open region or the empty set, 
and the
regions are disjoint:
${\cal D}{\mathstrut}_{m} \cap \,{\cal D}{\mathstrut}_{m'} = \varnothing$
for $m \neq m'$. 
Observe that the maximum-likelihood decoder with open decision regions ${\cal D}{\mathstrut}_{m}^{*}\,$, defined for $m = 1, 2, .\,.\,.\, , \, M$
as
\begin{displaymath}
{\cal D}{\mathstrut}_{m}^{*}
\;\; = \;\;
\mathbb{R}{\mathstrut}^{n} \setminus
\bigcup_{
m':\;\;
(m' \, < \; m)
\; \lor \;
\big(\,m' \, > \, m \;\, \land \;\, {\bf x}(m') \, \neq \, {\bf x}(m)\,\big)
}
\Big\{
{\bf y}: \;
\|
{\bf y} - {\bf x}(m')
\|
\, \leq \,
\|
{\bf y} - {\bf x}(m)
\|
\Big\},
\end{displaymath}
is a special case of (\ref{eqDec}).
Note that the formal
definition of ${\cal D}{\mathstrut}_{m}^{*}$ includes the undesirable possibility of ${\bf x}(m') = {\bf x}(m)$ for $m' \neq m$.


\section{Definitions}\label{Defs}


For each $n$, we define two discrete countable alphabets ${\cal X}_{n}$ and ${\cal Y}_{n}$ 
as
one-dimensional lattices:
\begin{align}
&
\alpha, \beta, \gamma \in (0, 1), \;\;\;
\alpha + \beta + \gamma = 1,
\nonumber \\
&
\Delta_{\alpha,\,n} \; \triangleq \; 1/n^{\alpha},
\;\;\;
\Delta_{\beta,\,n} \; \triangleq \; 1/n^{\beta},
\;\;\;
\Delta_{\gamma,\,n} \; \triangleq \; 1/n^{\gamma},
\label{eqDelta} \\
& \Delta_{\alpha,\,n}\cdot \Delta_{\beta,\,n}\cdot \Delta_{\gamma,\,n} \;\; = \;\; 1/n,
\label{eqCube}
\end{align}
\begin{align}
&
{\cal X}_{n} \;\; \triangleq \;\;
\bigcup_{i \, \in \, \mathbb{Z}}\big\{i\Delta_{\alpha,\,n}\big\},
\;\;\;\;\;\;
{\cal Y}_{n} \;\; \triangleq \;\;
\bigcup_{i \, \in \, \mathbb{Z}}\big\{i\Delta_{\beta,\,n}\big\}.
\label{eqAlphabets}
\end{align}
For each $n$, we define also a discrete positive measure (not necessarily a distribution), which will approximate the channel $w$:
\begin{align}
{W\mathstrut}_{\!n}(y \, | \, x) \;\; & \triangleq \;\;
w(y \, | \, x)\cdot\Delta_{\beta,\,n},
\;\;\; \forall x \in {\cal X}_{n}, \; \forall y \in {\cal Y}_{n}.
\label{eqChanApprox}
\end{align}
Denoting by $C^{0}(A)$ a class of functions $f\! : \mathbb{R} \rightarrow \mathbb{R}_{\,\geq\, 0}$
continuous on an open subset $A \subseteq \mathbb{R}$,
we define 
\begin{align}
{\cal F}_{n} \; & \triangleq \;
\bigg\{
f\! : \mathbb{R} \rightarrow \mathbb{R}_{\,\geq\, 0}
\;\;\; \Big| \;\;\;
\sup_{y\,\in\,\mathbb{R}}f(y)< +\infty; \;\;\;
f \in
C^{0}\big(\mathbb{R} \setminus \{ {\cal Y}_{n} +  
\Delta_{\beta,\,n}/2 \}\big);
\;\;
\int_{\mathbb{R}}f(y)dy = 1
\bigg\}.
\label{eqBoundedContinuous}
\end{align}
The 
set (\ref{eqBoundedContinuous}), defined for a given $n$, will be used only in the derivation of the correct-decoding exponent,
while the following set of functions will be used only in the derivation of the error exponent:
\begin{align}
{\cal L} \; & \triangleq \;
\bigg\{
f\! : \mathbb{R} \rightarrow \mathbb{R}_{\,\geq\, 0}
\;\; \Big| \;\;
|f(y_{1}) - f(y_{2})| \leq  
K|y_{1} - y_{2}|,
\; \forall y_{1}, y_{2}\,;
\;
\int_{\mathbb{R}}f(y)dy = 1
\bigg\}, \;\; K  \triangleq \frac{1}{\sigma^{2}\sqrt{2\pi e}}.
\label{eqLipschitz} 
\end{align}
Note that ${\cal L}$ is a convex set and also 
each function $f \in {\cal L}$ is 
bounded and cannot exceed $\sqrt{K}$.

With a parameter $\rho \in (-1, \, +\infty)$, we define the following Gaussian probability density functions: 
\begin{align}
{p\mathstrut}_{Y|X}^{(\rho)}(y \, | \, x) \;\; & \triangleq \;\; \tfrac{1}{\sigma{\mathstrut}_{Y|X}(\rho)\sqrt{2\pi}}
\exp\Big\{\!-\tfrac{(y \, - \, k_{\rho} 
\, \cdot \, x)^{2}}{2\sigma_{Y|X}^{2}(\rho)}
\Big\},
\label{eqSolution} \\
k_{\rho} 
\;\; & \triangleq \;\;
\tfrac{
\text{SNR} \, - \, \rho \, - \, 1
\; + \;
\sqrt{(\text{SNR} \, - \, \rho \, - \, 1)^{2} \; + \; 4\,\cdot\,\text{SNR}}
}
{2\,\cdot\,\text{SNR}},
\;\;\;\;\;\; \text{SNR} \,\triangleq \, s^{2}/\sigma^{2},
\label{eqSquareroot} \\
\sigma_{Y|X}^{2}(\rho) \;\; & \triangleq \;\;
(1 + \rho)k_{\rho}\,\sigma^{2},
\label{eqSigmaYXDef} \\
{\widehat{p}\mathstrut}_{Y}^{\,(\rho)}(y) \;\; & \triangleq \;\; \tfrac{1}{\sigma{\mathstrut}_{Y}(\rho)\sqrt{2\pi}}
\exp\Big\{\!-\tfrac{y^{2}}{2\sigma_{Y}^{2}(\rho)}
\Big\},
\label{eqIdealYPDF} \\
\sigma_{Y}^{2}(\rho)
\;\; & \triangleq \;\;
\sigma^{2} \, + \, k_{\rho}\, 
s^{2},
\label{eqSigmaY} \\
{\widehat{p}\mathstrut}_{X}(x) \;\; & \triangleq \;\;
\tfrac{1}{s\sqrt{2\pi}}\exp\Big\{\!-\tfrac{x^{2}}{2s^{2}}
\Big\}.
\label{eqIdealXPDF}
\end{align}
The first property of the following lemma shows that ${\widehat{p}\mathstrut}_{Y}^{\,(\rho)}$
is the $y$-marginal PDF of the product $\,{\widehat{p}\mathstrut}_{X}{p\mathstrut}_{Y|X}^{(\rho)}$.

\bigskip

\begin{lemma}[Properties of (\ref{eqSolution})-(\ref{eqIdealXPDF})] \label{LemSigmaZ}
{\em The following properties hold:}
\begin{align}
\sigma_{Y}^{2}(\rho)
\;\; & = \;\;
\sigma_{Y|X}^{2}(\rho) \, + \, k_{\rho}^{2}\, 
s^{2},
\label{eqIdealYMarginal} \\
\frac{1 + \rho}{\sigma_{Y|X}^{2}(\rho)} \;\; & = \;\; \frac{\rho}{\sigma_{Y}^{2}(\rho)} \, + \, \frac{1}{\sigma^{2}},
\label{eqProperty} \\
\sigma_{Y|X}^{2}(\rho)
\;\; & = \;\;
\sigma^{2} \, + \, k_{\rho}(1 - k_{\rho})
s^{2},
\label{eqSigmaYX} \\
1 \, & \geq \,k_{\rho} \, > \, 0, \;\;\;\;\;\;\;\;\;\;\;\;\;\,\,\,\,
\rho \geq 0,
\label{eqPropk} \\
\tfrac{1}{2}
\big[ 1
\; + \;
\sqrt{1 \; + \; 4 \sigma^{2}s^{-2}}\,
\big]
\, & \geq \,k_{\rho} \, \geq \, 1, \;\;\;\;\;\; -1 \leq \rho \leq 0,
\label{eqPropk2} \\
{p\mathstrut}_{Y|X}^{(\rho)}(\,\cdot\,\,|\,x) \, & \in \, {\cal L}, \;\;\;\;\;\; \forall \rho \geq 0, \;
\forall x \in \mathbb{R},
\label{eqInL}
\end{align}
{\em 
and for any two jointly distributed random variables $(X, Y)$, such that
$\mathbb{E}
\big[X^{2}\big] = \sigma_{X}^{2} \, \leq \, s^{2} + \epsilon$, $\;\epsilon > 0$,
and $Y \,|\, X = x \; \sim \; {\cal N}\big(k_{\rho}\,x, \, \sigma_{Y|X}^{2}(\rho)\big)$, 
it holds that
}
\begin{align} 
\mathbb{E}
\big[(Y-X)^{2}\big]
\;\; & = \;\;
\sigma^{2} \, + \, (1 - k_{\rho})s^{2} \, + \, (1 - k_{\rho})^{2}
(\sigma_{X}^{2}
- s^{2})
\label{eqSigmaZ} \\
& \leq \;\;
\Bigg\{
\begin{array}{l r}
\sigma^{2} \, + \, s^{2} \, + \, \epsilon, &
\;\;\;\;\;\;  \rho \geq 0, \\
\sigma^{2} \, + \,
\epsilon\sigma^{2}s^{-2}, &
\;\;\;\;\;\;
 -1 < \rho \leq 0.
\end{array}
\label{eqSigmazBound}
\end{align}
\end{lemma}
\begin{proof}
The first property (\ref{eqIdealYMarginal}) can be verified using (\ref{eqSigmaY}), (\ref{eqSigmaYXDef}), (\ref{eqSquareroot}).
Then (\ref{eqProperty}) can be obtained from (\ref{eqIdealYMarginal}), (\ref{eqSigmaY}), (\ref{eqSigmaYXDef}).
Property (\ref{eqSigmaYX}) follows by (\ref{eqIdealYMarginal}) and (\ref{eqSigmaY}).
It can be verified from (\ref{eqSquareroot}) that $k_{\rho}$ is a positive monotonically decreasing function of $\rho$,
such that $k_{0} = 1$. Then we get (\ref{eqPropk}) and (\ref{eqPropk2}).
From (\ref{eqSigmaYX}) and (\ref{eqPropk}) we see that $\sigma_{Y|X}^{2}(\rho) \geq \sigma^{2}$ for all $\rho \geq 0$,
which gives (\ref{eqInL}).
Equality (\ref{eqSigmaZ}) can be obtained using (\ref{eqSigmaYX}).
Then, using (\ref{eqPropk}) and (\ref{eqPropk2}), we obtain (\ref{eqSigmazBound}).
\end{proof}

The following expressions will describe our results for the error and correct-decoding exponents:
\begin{align}
E_{e}(R) \;\; & \triangleq \;\;
\;\;
\sup_{\rho\,\geq\,0}
\;\;
\Big\{
D\big(\,{p\mathstrut}_{Y|X}^{(\rho)}\,\|\, \, w \, \,|\, \,{\widehat{p}\mathstrut}_{X}\big)
\, + \,
\rho
\big[
I\big(\,{\widehat{p}\mathstrut}_{X}, \,{p\mathstrut}_{Y|X}^{(\rho)}\big)
\, - \,
R
\big]
\Big\},
\label{eqEDef} \\
E_{c}(R) \;\; & \triangleq \;\;
\sup_{\!\!\!\!\! -1 \, < \, \rho\,\leq\,0}
\Big\{
D\big(\,{p\mathstrut}_{Y|X}^{(\rho)}\,\|\, \, w \, \,|\, \,{\widehat{p}\mathstrut}_{X}\big)
\, + \,
\rho
\big[
I\big(\,{\widehat{p}\mathstrut}_{X}, \,{p\mathstrut}_{Y|X}^{(\rho)}\big)
\, - \,
R
\big]
\Big\}.
\label{eqECDef}
\end{align}
The following identity can be obtained
using 
(\ref{eqSolution}),  
(\ref{eqSigmaYXDef}), (\ref{eqIdealYPDF}), (\ref{eqIdealXPDF}), (\ref{eqProperty}):
\begin{align}
& D\big(\,{p\mathstrut}_{Y|X}^{(\rho)}\,\|\, \, w \, \,|\, \,{\widehat{p}\mathstrut}_{X}\big)
\, + \,
\rho
I\big(\,{\widehat{p}\mathstrut}_{X}, \,{p\mathstrut}_{Y|X}^{(\rho)}\big)
\;\; \equiv \;\;
c_{0}(\rho) \, + \, c_{1}(\rho)s^{2},
\nonumber \\
& c_{0}(\rho) \;\; \triangleq \;\; \frac{1}{\ln b} 
\ln\bigg(\frac{\sigma \cdot \sigma_{Y}^{\rho}(\rho)}{\sigma_{Y|X}^{1\,+\,\rho}(\rho)}\bigg),
\;\;\;\;\;\;\;\;\;
c_{1}(\rho) \;\;
\triangleq \;\; \frac{1 - k_{\rho}}{2\sigma^{2}\ln b}.
\label{eqC0C1}
\end{align}
We note also that $c_{0}(\rho) \rightarrow 0$, as $\rho \rightarrow +\infty$, which can be verified using the properties (\ref{eqSigmaYXDef}), (\ref{eqSigmaY}), and (\ref{eqSigmaYX}).
It can be verified that the expression inside the supremum of (\ref{eqEDef})
is equivalent to the expression for the Gaussian random-coding
error exponent of Gallager before the maximization over $\rho$ \cite[Eq.~7.4.24 with Eq.~7.4.28]{Gallager}.
Therefore, with the supremum over $\rho \geq 0$,
the expression (\ref{eqEDef}) coincides with the converse sphere-packing bound of Shannon
(\ref{eqShannonExponent}).




\section{Main results}\label{Main}


In this section we present two theorems and a proposition.
The proof of the first theorem relies on Lemmas~\ref{LemAllCodebooks} and~\ref{LemPDFtoT},
which appear in Sections~\ref{ErrExp} and~\ref{PDFtypePDF}, respectively.

\bigskip

\begin{thm}[Error exponent] \label{thmErrorExp}
{\em
Let $J \sim \text{Unif}\,\big(\{1, 2, .\,.\,.\, , \, M\}\big)$ be a random variable,
independent of the channel, and let ${\bf x}(J) \rightarrow {\bf Y}$ be the random channel-input and channel-output vectors, respectively.
Then
}
\begin{align}
\limsup_{\substack{n\,\rightarrow\,\infty}}
\;
\sup_{{\cal C}}
\;
\sup_{g}
\;
\bigg\{\!
- \frac{1}{n}\,{\log\mathstrut}_{\!b}\Pr \big\{
g({\bf Y}) \neq J \big\}
\bigg\}
\;\; & \leq \;\;
E_{e}(R),
\nonumber
\end{align}
{\em where $E_{e}(R)$ is defined by (\ref{eqEDef}), decoder functions $g$ are defined by (\ref{eqDec}), and codebooks ${\cal C}$ satisfy (\ref{eqPowerConstraint}).}
\end{thm}

\bigskip

\begin{proof}
Starting from Lemma~\ref{LemAllCodebooks}, we can write the following sequence of inequalities:
\begin{align}
&
\limsup_{\substack{n\,\rightarrow\,\infty}}
\;\;\;
\sup_{{\cal C}}
\;\;
\sup_{g}
\;\;
\left\{
- \frac{1}{n}\,{\log\mathstrut}_{\!b} \Pr \big\{
g({\bf Y}) \neq J \big\}
\right\}
\label{eqReliabilityF} \\
\overset{a}{\leq} \;\; &
\limsup_{\substack{n\,\rightarrow\,\infty}}
\;
\max_{\substack{\\{P\mathstrut}_{\!X{\color{white}|}}\!\!:\\
{P\mathstrut}_{\!X}\,\in \, {\cal P}_{n}({\cal X}_{n}),
\\
\mathbb{E}[X^{2}] \; \leq \; s^{2}\, + \, \epsilon}}
\;\;\;\;\;\;
\min_{\substack{\\{P\mathstrut}_{\!Y|X}:\\
{P\mathstrut}_{\!XY}\,\in \, {\cal P}_{n}({\cal X}_{n}\,\times\, {\cal Y}_{n}),
\\
\mathbb{E}[(Y-X)^{2}] \; \leq \; \sigma^2 \, + \, s^{2} \, + \, 2\epsilon, 
\\ I({P\mathstrut}_{\!X}, \, {P\mathstrut}_{\!Y|X}) \; \leq \; R \, - \, \epsilon
}}
\!\!\!\!\!
\!
\;\,\,
D\big({P\mathstrut}_{\!Y|X}\,\|\, {W\mathstrut}_{\!n} \,|\,  {P\mathstrut}_{\!X}\big)
\label{eqCCExponent} \\
\overset{b}{\leq} \;\; &
\limsup_{\substack{n\,\rightarrow\,\infty}}
\;
\max_{\substack{\\{P\mathstrut}_{\!X{\color{white}|}}\!\!:\\
{P\mathstrut}_{\!X}\,\in \, {\cal P}_{n}({\cal X}_{n}),
\\
\mathbb{E}[X^{2}] \; \leq \; s^{2}\, + \, \epsilon}}
\;\;\;\;\;\;\,
\inf_{\substack{{p\mathstrut}_{Y|X}:\\
{p\mathstrut}_{Y|X}(\, \cdot \, \,|\, x) \, \in \, {\cal L}, \; \forall x,
\\
\mathbb{E}[(Y-X)^{2}] \; \leq \; \sigma^2 \, + \, s^{2} \, + \, \epsilon,
\\ I({P\mathstrut}_{\!X}, \,\, {p\mathstrut}_{Y|X}) \; \leq \; R \, - \, 2\epsilon
}}
\!\!\!\!\!
\;\,\,
D\big(\,{p\mathstrut}_{Y|X}\,\|\, \, w \, \,|\, {P\mathstrut}_{\!X}\big)
\label{eqInf1} \\
\overset{c}{=}
\;\; &
\limsup_{\substack{n\,\rightarrow\,\infty}}
\;
\max_{\substack{\\{P\mathstrut}_{\!X{\color{white}|}}\!\!:\\
{P\mathstrut}_{\!X}\,\in \, {\cal P}_{n}({\cal X}_{n}),
\\
\mathbb{E}[X^{2}] \; \leq \; s^{2}\, + \, \epsilon}}
\;
\sup_{\substack{\rho\,\geq\,0}}
\inf_{\substack{{p\mathstrut}_{Y|X}:\\
{p\mathstrut}_{Y|X}(\, \cdot \, \,|\, x) \, \in \, {\cal L}, \; \forall x,
\\
\mathbb{E}[(Y-X)^{2}] \; \leq \; \sigma^2 \, + \, s^{2} \, + \, \epsilon
}}
\!\!\!\!\!
\Big\{
D\big(\,{p\mathstrut}_{Y|X}\,\|\, \, w \, \,|\, {P\mathstrut}_{\!X}\big)
\, + \,
\rho
\Big[
I\big({P\mathstrut}_{\!X}, \, {p\mathstrut}_{Y|X}\big)
 -
R
+
2\epsilon
\Big]
\Big\}
\label{eqInf2} \\
\overset{d}{\leq}  \;\; &
\limsup_{\substack{n\,\rightarrow\,\infty}}
\;
\max_{\substack{\\{P\mathstrut}_{\!X{\color{white}|}}\!\!:\\
{P\mathstrut}_{\!X}\,\in \, {\cal P}_{n}({\cal X}_{n}),
\\
\mathbb{E}[X^{2}] \; \leq \; s^{2}\, + \, \epsilon}}
\;
\sup_{\substack{\rho\,\geq\,0}}
\!\!\!\!\!
\;\;\;\;\;\;\;\;\;\;\;\;\;\;\;\;\;\;\;\;\;\;\;\;\;\;\;\;\,
\Big\{
D\big(\,{p\mathstrut}_{Y|X}^{(\rho)}\,\|\, \, w \, \,|\, {P\mathstrut}_{\!X}\big)
\, + \,
\rho
\Big[
I\big({P\mathstrut}_{\!X}, \, {p\mathstrut}_{Y|X}^{(\rho)}\big)
 -
R
+
2\epsilon
\Big]
\Big\}
\nonumber \\
\overset{e}{\equiv}  \;\; &
\limsup_{\substack{n\,\rightarrow\,\infty}}
\;
\max_{\substack{\\{P\mathstrut}_{\!X{\color{white}|}}\!\!:\\
{P\mathstrut}_{\!X}\,\in \, {\cal P}_{n}({\cal X}_{n}),
\\
\mathbb{E}[X^{2}] \; \leq \; s^{2}\, + \, \epsilon}}
\;
\sup_{\substack{\rho\,\geq\,0}}
\!\!\!\!\!
\;\;\;\;\;\;\;\;\;\;\;\;\;\;\;\;\;\;\;
\Big\{
c_{0}(\rho) \, + \, c_{1}(\rho)\,
\mathbb{E}\big[X^{2}\big]
\, + \,
\rho
\Big[
\! - \!
D\big(\,{p\mathstrut}_{Y}^{\,(\rho)}\,\|\, \, {\widehat{p}\mathstrut}_{Y}^{\,(\rho)} \big)
-
R
+
2\epsilon
\Big]
\Big\}
\label{eqMax} \\
\overset{f}{\leq}  \;\; &
\;\;\;\;\;\;\;\;\;\;\;\;\;\;\;\;\;\;\;\;\;\;\;\;\;\;\;\;\;\,\,
\sup_{\substack{\rho\,\geq\,0}}
\!\!\!\!\!
\;\;\;\;\;\;\;\;\;\;\;\;\;\;\;\;\;\;\;
\Big\{
c_{0}(\rho) \, + \, c_{1}(\rho)
(s^{2} + \epsilon)
\, - \,
\rho (R - 2\epsilon)
\Big\},
\label{eqWithEpsilon}
\end{align}
where:

($a$) holds for any $\epsilon > 0$ by Lemma~\ref{LemAllCodebooks} with $c_{XY} = \sigma^2 + s^{2} + 2\epsilon$.
Note also
that $D\big({P\mathstrut}_{\!Y|X}\,\|\, {W\mathstrut}_{\!n} \,|\,  {P\mathstrut}_{\!X}\big)$ in (\ref{eqCCExponent})
denotes the Kullback–Leibler divergence 
between
the probability distribution ${P\mathstrut}_{\!Y|X}$ 
and
the positive measure ${W\mathstrut}_{\!n}$ defined in (\ref{eqChanApprox}),
which is not a probability distribution but only approximates the channel $w$.

($b$) follows by Lemma~\ref{LemPDFtoT} for the alphabet parameters $\alpha \in \big(0, \tfrac{1}{4}\big)$
and $\tfrac{1}{3} + \tfrac{2}{3}\alpha < \beta < \tfrac{2}{3} - \tfrac{2}{3}\alpha$.

($c$) holds 
for all $R > 0$ with the possible exception of the {\em single point} on $R$-axis
where (\ref{eqInf1}) may transition between a finite value and $+\infty$.
For the equality, let us compare the infimum of (\ref{eqInf1}) and the supremum over $\rho \geq 0$ in (\ref{eqInf2})
as functions of $R \in \mathbb{R}$.
First, it can be observed that the supremum 
of (\ref{eqInf2})
is the lower convex envelope of the infimum of (\ref{eqInf1}). 
Second, the infimum of (\ref{eqInf1}) itself
is a convex ($\cup$) 
function of $R$. Then they coincide for all values of $R$, except possibly for the single point where
they both jump to $+\infty$.
This
property carries over to the external 
`$\limsup \; \max$'
as well.

($d$) follows because
by (\ref{eqSigmazBound}) and (\ref{eqInL}) function ${p\mathstrut}_{Y|X}^{(\rho)}$
satisfies the conditions under the infimum of (\ref{eqInf2}).

($e$) holds as equality inside the supremum over $\rho \geq 0$, separately for each $\rho$.
In (\ref{eqMax})
by ${p\mathstrut}_{Y}^{\,(\rho)}$ we denote the corresponding marginal PDF of the product
${P\mathstrut}_{\!X}{p\mathstrut}_{Y|X}^{(\rho)}$ and use the 
definitions (\ref{eqC0C1}).
Then ($e$) follows by the definitions of ${p\mathstrut}_{Y|X}^{(\rho)}$ and ${\widehat{p}\mathstrut}_{Y}^{\,(\rho)}$
in (\ref{eqSolution}), (\ref{eqIdealYPDF}), and by their properties (\ref{eqSigmaYXDef}), (\ref{eqProperty}).

($f$) follows by the non-negativity of the divergence, and by the condition under the maximum of (\ref{eqMax}),
since $c_{1}(\rho)\geq 0$ for $\rho \geq 0$.

In conclusion, according to ($c$) we obtain that the inequality between (\ref{eqReliabilityF}) and (\ref{eqWithEpsilon}),
as functions of $R$, holds for all $R > 0$,
except possibly for
the single point $R = 2\epsilon$, where the jump to $+\infty$ in (\ref{eqWithEpsilon}) occurs.
Therefore, taking the limit as $\epsilon \rightarrow 0$, we obtain that (\ref{eqReliabilityF}) is upper-bounded for all $R > 0$ by
\begin{displaymath}
\sup_{\substack{\rho\,\geq\,0}}
\;
\big\{
c_{0}(\rho) \, + \, c_{1}(\rho)s^{2}
\, - \,
\rho R
\big\},
\end{displaymath}
which is the same as (\ref{eqEDef}).
\end{proof}


The second theorem relies on Lemmas~\ref{LemCorDec} and~\ref{LemTtoPDF},
which appear in Sections~\ref{CorDecExp} and~\ref{PDFtypePDF}, respectively.

\bigskip

\begin{thm}[Correct-decoding exponent] \label{thmCorDecExp}
{\em
Let $J \sim \text{Unif}\,\big(\{1, 2, .\,.\,.\, , \, M\}\big)$ be a random variable,
independent of the channel, and let ${\bf x}(J) \rightarrow {\bf Y}$ be the random channel-input and channel-output vectors, respectively.
Then}
\begin{align}
\liminf_{\substack{n\,\rightarrow\,\infty}}
\;
\inf_{{\cal C}}
\;
\inf_{g}
\;
\bigg\{\!
- \frac{1}{n}\,{\log\mathstrut}_{\!b}\Pr \big\{
g({\bf Y}) = J \big\}
\bigg\}
\;\; & \geq \;\;
E_{c}(R),
\nonumber
\end{align}
{\em where $E_{c}(R)$ is defined by (\ref{eqECDef}), decoder functions $g$ are defined by (\ref{eqDec}), and codebooks ${\cal C}$ satisfy (\ref{eqPowerConstraint}).}
\end{thm}

\bigskip

\begin{proof}
Starting from Lemma~\ref{LemCorDec}, for each $R > 0$ we can choose a different parameter $\widetilde{\sigma} = \widetilde{\sigma}(R) \geq \sigma$,
such that there is equality $E(\widetilde{\sigma}(R)) = E_{c}(R)$ between (\ref{eqOutlierExp}) and (\ref{eqECDef}).
Then by (\ref{eqMinofTwoExponents2}) we obtain
\begin{align}
\liminf_{\substack{n\,\rightarrow\,\infty}}
\;
\inf_{{\cal C}}
\;
\inf_{g}
\;
\bigg\{\!
- \frac{1}{n}\,{\log\mathstrut}_{\!b}\Pr \big\{
g({\bf Y}) = J \big\}
\bigg\}
\;\; & \geq \;\;
\min\Big\{
\liminf_{\substack{n\,\rightarrow\,\infty}}
\,
E_{n}(R, \,\widetilde{\sigma}(R), \, \widetilde{\epsilon}, \, \epsilon), \;\;
E_{c}(R)
\Big\}.
\nonumber
\end{align}
With the choice $2\widetilde{\epsilon} = \epsilon\sigma^{2}s^{-2}$,
the first term
in the minimum
can be lower-bounded as follows:
\begin{align}
&
\liminf_{\substack{n\,\rightarrow\,\infty}}
\,
\min_{\substack{\\{P\mathstrut}_{\!X{\color{white}|}}\!\!:\\
{P\mathstrut}_{\!X}\,\in \, {\cal P}_{n}({\cal X}_{n}),
\\
\mathbb{E}[X^{2}] \; \leq \; s^{2}\, + \, \epsilon}}
\;\;
\min_{\substack{\\{P\mathstrut}_{\!Y|X}:\\
{P\mathstrut}_{\!XY}\,\in \, {\cal P}_{n}({\cal X}_{n}\,\times\, {\cal Y}_{n}),
\\
\mathbb{E}[(Y-X)^{2}] \; \leq \; \widetilde{\sigma}^{2}(R) \, + \, \widetilde{\epsilon}
}}
\!\!
\,
\Big\{
D\big({P\mathstrut}_{\!Y|X}\,\|\, {W\mathstrut}_{\!n} \,|\,  {P\mathstrut}_{\!X}\big)
\, + \,
\big|\,
R \, - \, I\big({P\mathstrut}_{\!X}, {P\mathstrut}_{\!Y|X}\big)
\,\big|^{+}
\Big\}
\nonumber \\
\overset{a}{\geq} \;
&
\liminf_{\substack{n\,\rightarrow\,\infty}}
\min_{\substack{\\{P\mathstrut}_{\!X{\color{white}|}}\!\!:\\
{P\mathstrut}_{\!X}\,\in \, {\cal P}_{n}({\cal X}_{n}),
\\
\mathbb{E}[X^{2}] \; \leq \; s^{2}\, + \, \epsilon}}
\;\;
\inf_{\substack{{p\mathstrut}_{Y|X}:\\
{p\mathstrut}_{Y|X}(\, \cdot \, \,|\, x) \, \in \, {\cal F}_{n}, \; \forall x,
\\
\mathbb{E}[(Y-X)^{2}] \; \leq \; \widetilde{\sigma}^{2}(R) \, + \, 2\widetilde{\epsilon}
}}
\!\!
\Big\{
D\big(\,{p\mathstrut}_{Y|X}\,\|\, \, w \, \,|\, {P\mathstrut}_{\!X}\big)
\, + \,
\big|\,
R \, - \, I\big({P\mathstrut}_{\!X}, \, {p\mathstrut}_{Y|X}\big)
\,\big|^{+}
\Big\}
\nonumber \\
\overset{b}{\geq} \;
&
\liminf_{\substack{n\,\rightarrow\,\infty}}
\min_{\substack{\\{P\mathstrut}_{\!X{\color{white}|}}\!\!:\\
{P\mathstrut}_{\!X}\,\in \, {\cal P}_{n}({\cal X}_{n}),
\\
\mathbb{E}[X^{2}] \; \leq \; s^{2}\, + \, \epsilon}}
\;\;
\inf_{\substack{{p\mathstrut}_{Y|X}:\\
{p\mathstrut}_{Y|X}(\, \cdot \, \,|\, x) \, \in \, {\cal F}_{n}, \; \forall x,
\\
\mathbb{E}[(Y-X)^{2}] \; \leq \; \widetilde{\sigma}^{2}(R) \, + \, 2\widetilde{\epsilon}
}}
\!\!
\Big\{
D\big(\,{p\mathstrut}_{Y|X}\,\|\, \, w \, \,|\, {P\mathstrut}_{\!X}\big)
\, - \,
\rho \Big[
R
\, - \,
D\big(\,{p\mathstrut}_{Y|X}\,\|\, \, {\widehat{p}\mathstrut}_{Y}^{\,(\rho)} \,|\, {P\mathstrut}_{\!X}\big)
\Big]
\Big\}
\nonumber
\end{align}
\begin{align}
\overset{c}{\equiv} \;
&
\liminf_{\substack{n\,\rightarrow\,\infty}}
\min_{\substack{\\{P\mathstrut}_{\!X{\color{white}|}}\!\!:\\
{P\mathstrut}_{\!X}\,\in \, {\cal P}_{n}({\cal X}_{n}),
\\
\mathbb{E}[X^{2}] \; \leq \; s^{2}\, + \, \epsilon}}
\;\;
\inf_{\substack{{p\mathstrut}_{Y|X}:\\
{p\mathstrut}_{Y|X}(\, \cdot \, \,|\, x) \, \in \, {\cal F}_{n}, \; \forall x,
\\
\mathbb{E}[(Y-X)^{2}] \; \leq \; \widetilde{\sigma}^{2}(R) \, + \, 2\widetilde{\epsilon}
}}
\!\!
\Big\{
c_{0}(\rho) \, + \, c_{1}(\rho)\,
\mathbb{E}\big[X^{2}\big]
\, - \, \rho R
\, + \,
(1 + \rho)
D\big(\,{p\mathstrut}_{Y|X}\,\|\, \, {p\mathstrut}_{Y|X}^{(\rho)} \,|\, {P\mathstrut}_{\!X}\big)
\Big\}
\nonumber \\
\overset{d}{=} \;
&
\liminf_{\substack{n\,\rightarrow\,\infty}}
\min_{\substack{\\{P\mathstrut}_{\!X{\color{white}|}}\!\!:\\
{P\mathstrut}_{\!X}\,\in \, {\cal P}_{n}({\cal X}_{n}),
\\
\mathbb{E}[X^{2}] \; \leq \; s^{2}\, + \, \epsilon}}
\;\;
\;\;\;\;\;\;\;\;\;\;\;\;\;\;\;\;\;\;\;\;\;\;\;\;\;\;\,\,\,\,
\!\!
\Big\{
c_{0}(\rho) \, + \, c_{1}(\rho)\,
\mathbb{E}\big[X^{2}\big]
\, - \, \rho R
\Big\}
\label{eqMin2} \\
\overset{e}{\geq} \;
&
\!\!
\;\;\;\;\;\;\;\;\;\;\;\;\;
\;\;\;\;\;\;\;\;\;\;\;\;\;\;\;\;\;
\;\;
\;\;\;\;\;\;\;\;\;\;\;\;\;\;\;\;\;\;\;\;\;\;\;\;\;\;\;\;\,
c_{0}(\rho) \, + \, c_{1}(\rho)
(s^{2} + \epsilon)
\, - \, \rho R,
\label{eqAnyRho}
\end{align}
where:

($a$) follows by Lemma~\ref{LemTtoPDF} with $c_{XY} = \widetilde{\sigma}^{2}(R) + \widetilde{\epsilon}$.

($b$) holds for $\rho \in (-1, 0]$, because
$\big|\,
R \, - \, I\big({P\mathstrut}_{\!X}, \, {p\mathstrut}_{Y|X}\big)
\,\big|^{+} \, \geq \, -\rho \big[\,
R \, - \, I\big({P\mathstrut}_{\!X}, \, {p\mathstrut}_{Y|X}\big)
\,\big]$ for any such $\rho$, and because
$I\big({P\mathstrut}_{\!X}, \, {p\mathstrut}_{Y|X}\big) \leq D\big(\,{p\mathstrut}_{Y|X}\,\|\, \, {\widehat{p}\mathstrut}_{Y}^{\,(\rho)} \,|\, {P\mathstrut}_{\!X}\big)$,
where ${\widehat{p}\mathstrut}_{Y}^{\,(\rho)}$ is the Gaussian PDF defined in (\ref{eqIdealYPDF}).

($c$) holds as an identity inside the infimum by the definitions (\ref{eqSolution}), (\ref{eqIdealYPDF}), (\ref{eqC0C1}),
and properties (\ref{eqSigmaYXDef}), (\ref{eqProperty}).

($d$) holds if $2\widetilde{\epsilon} \geq \epsilon\sigma^{2}s^{-2}$ and $\rho \in (-1, 0]$, because then
by (\ref{eqSigmazBound}) and (\ref{eqBoundedContinuous})
the function ${p\mathstrut}_{Y|X}^{(\rho)}$
satisfies the conditions under the infimum
and achieves
the infimum.

($e$) follows by the condition under the minimum of (\ref{eqMin2}) since $c_{1}(\rho) \leq 0$
for $\rho \in (-1, 0]$.

In conclusion,
since (\ref{eqAnyRho}) is the lower bound for any $\rho \in (-1, 0]$ and $2\widetilde{\epsilon} \geq \epsilon\sigma^{2}s^{-2}$,
we obtain
\begin{displaymath}
\liminf_{\substack{n\,\rightarrow\,\infty}}
\,
E_{n}\big(R, \,\widetilde{\sigma}(R), \, \epsilon\sigma^{2}s^{-2}/2, \, \epsilon\big)
\;\; \geq \;\;
\sup_{\!\!\!\!\! -1 \, < \, \rho\,\leq\,0}
\big\{
c_{0}(\rho) \, + \, c_{1}(\rho)
(s^{2} + \epsilon)
\, - \, \rho R
\big\}
\;\; \overset{\epsilon\, \rightarrow \,0}{\longrightarrow} \;\; E_{c}(R).
\end{displaymath}
\end{proof}

{\em Remark:} Observe that 
neither
the inequality ($f$) in the proof of Theorem~\ref{thmErrorExp}
nor the inequality ($b$) in the proof of Theorem~\ref{thmCorDecExp}
can be met with equality unless
$D\big(\,{p\mathstrut}_{Y}^{\,(\rho)}\,\|\, \, {\widehat{p}\mathstrut}_{Y}^{\,(\rho)} \big) \rightarrow 0$,
where ${p\mathstrut}_{Y}^{\,(\rho)}$ is the $y$-marginal PDF of
${P\mathstrut}_{\!X}{p\mathstrut}_{Y|X}^{(\rho)}$.
Accordingly, since ${\widehat{p}\mathstrut}_{Y}^{\,(\rho)}$ is Gaussian, while ${p\mathstrut}_{Y}^{\,(\rho)}$
is a convolution of ${P\mathstrut}_{\!X}$ with the Gaussian PDF ${p\mathstrut}_{Y|X}^{(\rho)}$,
the type ${P\mathstrut}_{\!X}$ must converge to the Gaussian
distribution
in CF\footnote{It follows because the expression for the characteristic function of the zero-mean Gaussian distribution also has a Gaussian form.} and CDF
in order to achieve the 
exponents of Theorems~\ref{thmErrorExp} and~\ref{thmCorDecExp}.
In both proofs the type ${P\mathstrut}_{\!X}$ represents the histograms of codewords,
i.e., the empirical distributions of their 
quantized versions.

\bigskip

\begin{prop}[Parametric representations of $E_{c}$ and $E_{e}\,$]
\label{prpCorErrExp}

{\em For every 
$R \, \geq \, I(\,{\widehat{p}\mathstrut}_{X}, \,w)$ there exists a unique $\rho \in (-1, 0]$,
such that}
\begin{align}
R \, & = \, I\big(\,{\widehat{p}\mathstrut}_{X}, \,{p\mathstrut}_{Y|X}^{(\rho)}\big),
\;\;\;\;\;\;
E_{c}(R) \, = \, D\big(\,{p\mathstrut}_{Y|X}^{(\rho)}\,\|\, \, w \, \,|\, \,{\widehat{p}\mathstrut}_{X}\big).
\label{eqParC}
\end{align}

{\em For every 
$R \, \leq \, I(\,{\widehat{p}\mathstrut}_{X}, \,w)$
there exists a unique $\rho \geq 0$, such that}
\begin{align}
R \, & = \, I\big(\,{\widehat{p}\mathstrut}_{X}, \,{p\mathstrut}_{Y|X}^{(\rho)}\big),
\;\;\;\;\;\;
E_{e}(R) \, = \, D\big(\,{p\mathstrut}_{Y|X}^{(\rho)}\,\|\, \, w \, \,|\, \,{\widehat{p}\mathstrut}_{X}\big).
\label{eqParE}
\end{align}
\end{prop}


\begin{proof}
Let us denote $R_{\beta} \triangleq I\big(\,{\widehat{p}\mathstrut}_{X}, \,{p\mathstrut}_{Y|X}^{(\beta)}\big)$. Then
for $\beta \in (-1, 0]$ 
we can write a sandwich proof:
\begin{align}
&
\inf_{\substack{{p\mathstrut}_{Y|X}:\\
{\widehat{p}\mathstrut}_{X}{p\mathstrut}_{Y|X} \, \in \; {\cal N}
}}
\!\!
\Big\{
D\big(\,{p\mathstrut}_{Y|X}\,\|\, \, w \, \,|\, \,{\widehat{p}\mathstrut}_{X}\big)
\, + \,
\big|\,
R_{\beta}  -  I\big(\,{\widehat{p}\mathstrut}_{X}, \, {p\mathstrut}_{Y|X}\big)
\,\big|^{+}
\Big\}
\label{eqUseInf} \\
\overset{a}{\geq} \;
\sup_{\!\!\!\!\! -1 \, < \, \rho\,\leq\,0}
&
\inf_{\substack{{p\mathstrut}_{Y|X}:\\
{\widehat{p}\mathstrut}_{X}{p\mathstrut}_{Y|X} \, \in \; {\cal N}
}}
\!\!
\Big\{
D\big(\,{p\mathstrut}_{Y|X}\,\|\, \, w \, \,|\, \,{\widehat{p}\mathstrut}_{X}\big)
\, - \,
\rho \big[
R_{\beta}
 -
D\big(\,{p\mathstrut}_{Y|X}\,\|\, \, {\widehat{p}\mathstrut}_{Y}^{\,(\rho)} \,|\, \,{\widehat{p}\mathstrut}_{X}\big)
\big]
\Big\}
\nonumber \\
\overset{b}{\equiv} \;
\sup_{\!\!\!\!\! -1 \, < \, \rho\,\leq\,0}
&
\inf_{\substack{{p\mathstrut}_{Y|X}:\\
{\widehat{p}\mathstrut}_{X}{p\mathstrut}_{Y|X} \, \in \; {\cal N}
}}
\!\!
\Big\{
D\big(\,{p\mathstrut}_{Y|X}^{(\rho)}\,\|\, \, w \, \,|\, \,{\widehat{p}\mathstrut}_{X}\big)
\, - \,
\rho
\big[
R_{\beta}
 -
R_{\rho}
\big]
\, + \,
(1 + \rho)
D\big(\,{p\mathstrut}_{Y|X}\,\|\, \, {p\mathstrut}_{Y|X}^{(\rho)} \,|\, \,{\widehat{p}\mathstrut}_{X}\big)
\Big\}
\nonumber \\
\overset{c}{=} \;
\sup_{\!\!\!\!\! -1 \, < \, \rho\,\leq\,0}
&
\;\;\;\;\;\;\;\;\;\;\;\;\,\,\,\,
\Big\{
D\big(\,{p\mathstrut}_{Y|X}^{(\rho)}\,\|\, \, w \, \,|\, \,{\widehat{p}\mathstrut}_{X}\big)
\, - \,
\rho
\big[
R_{\beta}
 -
R_{\rho}
\big]
\Big\}
\; \equiv \; E_{c}(R_{\beta})
\; \overset{d}{\geq} \;
D\big(\,{p\mathstrut}_{Y|X}^{(\beta)}\,\|\, \, w \, \,|\, \,{\widehat{p}\mathstrut}_{X}\big),
\label{eqSandwich}
\end{align}
where ${\cal N}$ denotes the set of all bivariate non-degenerate Gaussian PDF's.
Here ($a$) follows similarly to the inequality ($b$) in Theorem~\ref{thmCorDecExp};
($b$) is an identity;
($c$) follows because ${\widehat{p}\mathstrut}_{X}{p\mathstrut}_{Y|X}^{(\rho)}$ 
is 
Gaussian 
and ${p\mathstrut}_{Y|X}^{(\rho)}$ achieves the infimum;
($d$) is a lower bound on the supremum 
at $\rho = \beta$.
Finally, since 
the RHS of (\ref{eqSandwich}) is further lower-bounded by the infimum (\ref{eqUseInf}),
we conclude that $E_{c}(R_{\beta}) = D\big(\,{p\mathstrut}_{Y|X}^{(\beta)}\,\|\, \, w \, \,|\, \,{\widehat{p}\mathstrut}_{X}\big)$.

For $\beta \geq 0$, besides $R_{\beta}$
let us define
$R_{\beta}^{(\rho)} \triangleq D\big(\,{p\mathstrut}_{Y|X}^{(\rho)}\,\|\, \, {\widehat{p}\mathstrut}_{Y}^{\,(\beta)} \,|\, \,{\widehat{p}\mathstrut}_{X}\big)$.
Then 
\begin{align}
&
\inf_{\substack{{p\mathstrut}_{Y|X}:\\
{\widehat{p}\mathstrut}_{X}{p\mathstrut}_{Y|X} \, \in \; {\cal N}, \\
D(\,{p\mathstrut}_{Y|X}\,\|\, \, {\widehat{p}\mathstrut}_{Y}^{\,(\beta)} \,|\, \,{\widehat{p}\mathstrut}_{X})
\; \leq \; R_{\beta}
}}
\!\!\!\!\!\!\!\!\,\,\,
D\big(\,{p\mathstrut}_{Y|X}\,\|\, \, w \, \,|\, \,{\widehat{p}\mathstrut}_{X}\big)
\label{eqUseInf2} \\
\overset{a}{\geq} \;
&
\sup_{\rho \, \geq \, 0}
\,\,\,
\inf_{\substack{{p\mathstrut}_{Y|X}:\\
{\widehat{p}\mathstrut}_{X}{p\mathstrut}_{Y|X} \, \in \; {\cal N}
}}
\,\,\,\,
\Big\{
D\big(\,{p\mathstrut}_{Y|X}\,\|\, \, w \, \,|\, \,{\widehat{p}\mathstrut}_{X}\big)
\, + \,
\rho \big[
D\big(\,{p\mathstrut}_{Y|X}\,\|\, \, {\widehat{p}\mathstrut}_{Y}^{\,(\beta)} \,|\, \,{\widehat{p}\mathstrut}_{X}\big)
 -
R_{\beta}
\big]
\Big\}
\nonumber \\
\overset{b}{\equiv} \;
&
\sup_{\rho\,\geq\,0}
\,\,\,
\inf_{\substack{{p\mathstrut}_{Y|X}:\\
{\widehat{p}\mathstrut}_{X}{p\mathstrut}_{Y|X} \, \in \; {\cal N}
}}
\,\,\,\,
\Big\{
D\big(\,{p\mathstrut}_{Y|X}^{(\rho)}\,\|\, \, w \, \,|\, \,{\widehat{p}\mathstrut}_{X}\big)
\, + \,
\rho
\big[
R_{\beta}^{(\rho)}
 -
R_{\beta}
\big]
\, + \,
(1 + \rho)
D\big(\,{p\mathstrut}_{Y|X}\,\|\, \, {p\mathstrut}_{Y|X}^{(\rho)} \,|\, \,{\widehat{p}\mathstrut}_{X}\big)
\Big\}
\nonumber \\
\overset{c}{=} \;
&
\sup_{\rho\,\geq\,0}
\;\;\;\;\;\;\;\;\;\;\;\;\;\;\;\;\;\;\,\,
\Big\{
D\big(\,{p\mathstrut}_{Y|X}^{(\rho)}\,\|\, \, w \, \,|\, \,{\widehat{p}\mathstrut}_{X}\big)
\, + \,
\rho
\big[
R_{\beta}^{(\rho)}
 -
R_{\beta}
\big]
\Big\}
\nonumber \\
\overset{d}{\geq} \;
&
\sup_{\rho\,\geq\,0}
\;\;\;\;\;\;\;\;\;\;\;\;\;\;\;\;\;\;\,\,
\Big\{
D\big(\,{p\mathstrut}_{Y|X}^{(\rho)}\,\|\, \, w \, \,|\, \,{\widehat{p}\mathstrut}_{X}\big)
\, + \,
\rho
\big[
R_{\rho}
 -
R_{\beta}
\big]
\Big\}
\; \equiv \; E_{e}(R_{\beta})
\; \overset{e}{\geq} \;
D\big(\,{p\mathstrut}_{Y|X}^{(\beta)}\,\|\, \, w \, \,|\, \,{\widehat{p}\mathstrut}_{X}\big).
\label{eqSandwich2}
\end{align}
Here ($a$) follows due to the inequality under the first infimum;
($b$) is an identity;
($c$) follows because 
${\widehat{p}\mathstrut}_{X}{p\mathstrut}_{Y|X}^{(\rho)}$
is 
Gaussian 
and ${p\mathstrut}_{Y|X}^{(\rho)}$ achieves the infimum;
($d$) follows because $R_{\beta}^{(\rho)} \geq R_{\rho}^{(\rho)} \equiv R_{\rho}$;
($e$) is a lower bound on the supremum at $\rho = \beta$.
Since the RHS of (\ref{eqSandwich2}) is 
lower-bounded by the infimum (\ref{eqUseInf2}),
we obtain
$E_{e}(R_{\beta}) = D\big(\,{p\mathstrut}_{Y|X}^{(\beta)}\,\|\, \, w \, \,|\, \,{\widehat{p}\mathstrut}_{X}\big)$.
From $I\big(\,{\widehat{p}\mathstrut}_{X}, \,{p\mathstrut}_{Y|X}^{(\rho)}\big) =
\frac{1}{2}\,{\log\mathstrut}_{\!b}\big(\sigma_{Y}^{2}(\rho)/\sigma_{Y|X}^{2}(\rho)\big)$
using (\ref{eqSigmaY}) and (\ref{eqSigmaYX}) we obtain
$\frac{d R_{\rho}}{d\rho} =  \frac{dR_{\rho}}{d k_{\rho}}\cdot\frac{d k_{\rho}}{d\rho}< 0$.
Hence for every 
$R > 0$ the parameter $\rho(R) \in (-1, +\infty)$ is unique.
\end{proof}

\bigskip

The parametric representation (\ref{eqParC}) is equivalent to \cite[Eq.~22]{Oohama17}.

\section{Method of types}\label{MOT}


In this section we extend the method of types \cite{CoverThomas} to include the countable alphabets of reals
(\ref{eqAlphabets})
by using 
power constraints on the types. The results of this section are then used
in the rest of the paper.


\subsection*{\underline{Alphabet size}}



Consider all the types ${P\mathstrut}_{\!X} \in {\cal P}_{n}({\cal X}_{n})$
satisfying the power constraint $\mathbb{E}_{{P\mathstrut}_{\!X}}\!\big[X^{2}\big] \leq c_{X}$.
Let ${\cal X}_{n}(c_{X}) \subseteq {\cal X}_{n}$ denote the subset of the 
alphabet
used by these types. Every letter $x = i\Delta_{\alpha,\,n} \in {\cal X}_{n}(c_{X})$ must satisfy $|\, i\Delta_{\alpha,\,n}\,| \,\leq\, \sqrt{n c_{X}}$.
Then ${\cal X}_{n}(c_{X})$ is finite and we obtain


\begin{lemma}[Alphabet size] \label{LemAlphabetSize}
{\em $\;\;|\, {\cal X}_{n}(c_{X}) \,| \;\; \leq \;\; 2\sqrt{c_{X}} n^{1/2 \, + \, \alpha} + 1 \;\; \leq \;\; (2\sqrt{c_{X}} + 1) n^{1/2 \, + \, \alpha}$.}
\end{lemma}


\subsection*{\underline{Size of a type class}}


For 
${P\mathstrut}_{\!XY} \in {\cal P}_{n}({\cal X}_{n}\times {\cal Y}_{n})$ let us define
\begin{align}
{\cal S}({P\mathstrut}_{\!XY}) \;\; & \triangleq \;\;
\big\{
(x, y) \in {\cal X}_{n}\times {\cal Y}_{n} : \;\; {P\mathstrut}_{\!XY}(x, y) > 0
\big\},
\nonumber \\
{\cal S}({P\mathstrut}_{\!X}) \;\; & \triangleq \;\;
\big\{
x \in {\cal X}_{n}: \;\; {P\mathstrut}_{\!X}(x) > 0
\big\},
\;\;\;\;\;\;
{\cal S}({P\mathstrut}_{\!Y}) \;\; \triangleq \;\;
\big\{
y \in {\cal Y}_{n}: \;\; {P\mathstrut}_{\!Y}(y) > 0
\big\}.
\nonumber
\end{align}


\begin{lemma}[Support of a joint type] \label{LemSupportXY}
{\em 
Let ${P\mathstrut}_{\!XY} \in {\cal P}_{n}({\cal X}_{n}\times {\cal Y}_{n})$ be a joint type,
such that $\mathbb{E}_{{P\mathstrut}_{\!X}}\!\big[X^{2}\big] \leq c_{X}$
and $\mathbb{E}_{{P\mathstrut}_{\!Y}}\!\big[Y^{2}\big] \leq c_{Y}$.
Then}
\begin{align}
|\, {\cal S}({P\mathstrut}_{\!XY}) \,| \;\; & \leq \;\; \sqrt{2\pi (c_{X} + c_{Y} + 1/6)}\cdot n^{(1 \, + \, \alpha \, + \, \beta)/2}.
\nonumber
\end{align}
\end{lemma}
The proof is given in the Appendix A.


\begin{lemma}[Support of a type] \label{LemSupportX}
{\em 
Let ${P\mathstrut}_{\!X} \in {\cal P}_{n}({\cal X}_{n})$
and
${P\mathstrut}_{\!Y} \in {\cal P}_{n}({\cal Y}_{n})$
be types,
such that $\mathbb{E}_{{P\mathstrut}_{\!X}}\!\big[X^{2}\big] \leq c_{X}$
and $\mathbb{E}_{{P\mathstrut}_{\!Y}}\!\big[Y^{2}\big] \leq c_{Y}$.
Then}
\begin{align}
|\, {\cal S}({P\mathstrut}_{\!X}) \,| \;\; & \leq \;\; (12c_{X} + 1)^{1/3}\cdot n^{(1 \, + \, 2\alpha)/3},
\nonumber \\
|\, {\cal S}({P\mathstrut}_{\!Y}) \,| \;\; & \leq \;\; (12c_{Y} + 1)^{1/3}\cdot n^{(1 \, + \, 2\beta)/3}.
\nonumber
\end{align}
\end{lemma}
The proof for ${P\mathstrut}_{\!X}$ is given in the Appendix A.
For ${P\mathstrut}_{\!Y}$, the parameters $c_{Y}$, $\beta$ replace, respectively, $c_{X}$, $\alpha$.

\bigskip

\begin{lemma}[Size of a type class] \label{LemTypeSize}
{\em Let ${P\mathstrut}_{\!XY} \in {\cal P}_{n}({\cal X}_{n}\times {\cal Y}_{n})$ be a joint type,
such that $\mathbb{E}_{{P\mathstrut}_{\!X}}\!\big[X^{2}\big] \leq c_{X}$
and $\mathbb{E}_{{P\mathstrut}_{\!Y}}\!\big[Y^{2}\big] \leq c_{Y}$.
Then}
\begin{align}
H({P\mathstrut}_{\!XY})
\; - \;
c_{1}\,
\frac{{\log\mathstrut}_{\!b}\,(n + 1)}{n^{\gamma/2}}
\;\; & \leq \;\;
\frac{1}{n}\,{\log\mathstrut}_{\!b}\,|\, T({P\mathstrut}_{\!XY}) \,|
\;\; \leq \;\;
H({P\mathstrut}_{\!XY}),
\label{eqJointExp} \\
H({P\mathstrut}_{\!X})
\; - \;
c_{2}\,
\frac{{\log\mathstrut}_{\!b}\,(n + 1)}{n^{2(1\,-\,\alpha)/3}}
\;\; & \leq \;\;
\frac{1}{n}\,{\log\mathstrut}_{\!b}\,|\, T({P\mathstrut}_{\!X}) \,|
\;\;
\;\;
\leq \;\;
H({P\mathstrut}_{\!X}),
\label{eqMargExpX} \\
H({P\mathstrut}_{\!Y})
\; - \;
c_{3}\,
\frac{{\log\mathstrut}_{\!b}\,(n + 1)}{n^{2(1\,-\,\beta)/3}}
\;\; & \leq \;\;
\frac{1}{n}\,{\log\mathstrut}_{\!b}\,|\, T({P\mathstrut}_{\!Y}) \,|
\;\;
\;\,\,
\leq \;\;
H({P\mathstrut}_{\!Y}),
\label{eqMargExpY}
\end{align}
{\em where $c_{1} \triangleq \sqrt{2\pi (c_{X} + c_{Y} + 1/6)}$,
$\;c_{2} \triangleq (12c_{X} + 1)^{1/3}$, and
$c_{3} \triangleq (12c_{Y} + 1)^{1/3}$}.
\end{lemma}

\bigskip

\begin{proof}
Observe that the standard type-size bounds (see, e.g.,\cite[Eq.~11.16]{CoverThomas}) can be rewritten as
\begin{equation} \label{eqMOTSup}
\frac{1}{(n + 1)^{|\, {\cal S}({P\mathstrut}_{\!XY}) \,|}}\,b^{\,n H({P\mathstrut}_{\!XY})}
\;\; \leq \;\;
|\, T({P\mathstrut}_{\!XY}) \,|
\;\; \leq \;\;
b^{\,n H({P\mathstrut}_{\!XY})}.
\end{equation}
Here
$|\, {\cal S}({P\mathstrut}_{\!XY}) \,|$ 
can be replaced
with its upper bound of Lemma~\ref{LemSupportXY}.
This gives (\ref{eqJointExp}).
The remaining bounds of (\ref{eqMargExpX}) and (\ref{eqMargExpY})
are obtained similarly using Lemma~\ref{LemSupportX}.
\end{proof}

\bigskip

Since it holds for any ${\bf y} \in T({P\mathstrut}_{\!Y})$ that
$|\, T({P\mathstrut}_{\!X|\,Y}\,|\, {\bf y})\,| \, = \, |\, T({P\mathstrut}_{\!XY}) \,| \, / \, |\, T({P\mathstrut}_{\!Y}) \,|$,
and similarly for ${\bf x} \in T({P\mathstrut}_{\!X})$,
as a corollary of the previous lemma we also obtain

\bigskip

\begin{lemma}[Size of a conditional type class] \label{LemCondTypeSize}
{\em Let ${P\mathstrut}_{\!XY} \in {\cal P}_{n}({\cal X}_{n}\times {\cal Y}_{n})$ be a joint type,
such that $\mathbb{E}_{{P\mathstrut}_{\!X}}\!\big[X^{2}\big] \leq c_{X}$
and $\mathbb{E}_{{P\mathstrut}_{\!Y}}\!\big[Y^{2}\big] \leq c_{Y}$.
Then for ${\bf y} \in T({P\mathstrut}_{\!Y})$ and ${\bf x} \in T({P\mathstrut}_{\!X})$ respectively}
\begin{align}
H(X\,|\,Y)
\; - \;
c_{1}\,
\frac{{\log\mathstrut}_{\!b}\,(n + 1)}{n^{\gamma/2}}
\;\; & \leq \;\;
\frac{1}{n}\,{\log\mathstrut}_{\!b}\,|\, T({P\mathstrut}_{\!X|\,Y}\,|\, {\bf y}) \,|
\;\; \leq \;\;
H(X\,|\,Y)
\; + \;
c_{3}\,
\frac{{\log\mathstrut}_{\!b}\,(n + 1)}{n^{2(1\,-\,\beta)/3}},
\label{eqXgivenY} \\
H(Y\,|\,X)
\; - \;
c_{1}\,
\frac{{\log\mathstrut}_{\!b}\,(n + 1)}{n^{\gamma/2}}
\;\; & \leq \;\;
\frac{1}{n}\,{\log\mathstrut}_{\!b}\,|\, T({P\mathstrut}_{\!Y|X}\,|\, {\bf x}) \,|
\,
\;\; \leq \;\;
H(Y\,|\,X)
\; + \;
c_{2}\,
\frac{{\log\mathstrut}_{\!b}\,(n + 1)}{n^{2(1\,-\,\alpha)/3}},
\label{eqYgivenX}
\end{align}
{\em where $c_{1}$, $c_{2}$, and $c_{3}$ are defined as in Lemma~\ref{LemTypeSize}.}
\end{lemma}



\subsection*{\underline{Number of types}}


Let ${\cal P}_{n}\big({\cal X}_{n}, \,c_{X}\big)$
be the set of all the types ${P\mathstrut}_{\!X} \in {\cal P}_{n}({\cal X}_{n})$
satisfying the power constraint $\mathbb{E}_{{P\mathstrut}_{\!X}}\!\big[X^{2}\big] \leq c_{X}$.
Then its cardinality can be upper-bounded as follows:
\begin{equation} \label{eqNumberofTypes}
\big|\,  {\cal P}_{n}\big({\cal X}_{n}, \,c_{X}\big) \,\big|
\;\; \overset{(a)}{\leq} \;\;
\big|\,  {\cal P}_{n}\big({\cal X}_{n}(c_{X})\big) \,\big|
\;\; \overset{(b)}{\leq} \;\;
(n + 1)^{|\, {\cal X}_{n}(c_{X})\,|}
\;\; \overset{(c)}{\leq} \;\;
(n + 1)^{(2\sqrt{c_{X}} \, + \, 1) n^{1/2 \, + \, \alpha}},
\end{equation}
where ($a$) follows by the definition of ${\cal X}_{n}(c_{X})$
preceding Lemma~\ref{LemAlphabetSize}, ($b$) follows by \cite[Eq.~11.6]{CoverThomas},
and ($c$) follows by Lemma~\ref{LemAlphabetSize}.
This bound is sub-exponential in $n$ for $\alpha < 1/2$.
This can be also further improved and made sub-exponential in $n$ for all $\alpha \in (0, 1)$
using Lemma~\ref{LemSupportX}, as follows.

\bigskip

\begin{lemma}[Number of types] \label{LemNumofTypes}
\begin{equation} \label{eqImprovement}
\big|\,  {\cal P}_{n}\big({\cal X}_{n}, \,c_{X}\big) \,\big| \;\; \leq \;\;
\big((n + 1)c\big)^{\tilde{c}\, n^{(1\,+\, 2\alpha)/3}},
\end{equation}
{\em where $c \,\triangleq\, (2\sqrt{c_{X}} + 1)^{1/(3/2 \, + \, \alpha)}$
and $\tilde{c} \,\triangleq\, (3/2 + \alpha)(12 c_{X} + 1)^{1/3}$.}
\end{lemma}

\bigskip

\begin{proof}
Denoting $k \, \triangleq \, |\, {\cal X}_{n}(c_{X})\,|$ and
$\ell \, \triangleq \, \max_{\,{P\mathstrut}_{\!X} \, \in \;{\cal P}_{n}({\cal X}_{n}, \; c_{X})} \,|\, {\cal S}({P\mathstrut}_{\!X}) \,|$,
we can upper-bound as follows
\begin{equation} \label{eqKchooseL}
\big|\,  {\cal P}_{n}\big({\cal X}_{n}, \,c_{X}\big) \,\big| \;\; \leq \;\;
\tbinom{k}{\ell} (n + 1)^{\ell} \;\; \leq \;\;
k^{\ell}(n + 1)^{\ell}.
\end{equation}
Substituting 
for
$k$ and $\ell$ their upper bounds of Lemma~\ref{LemAlphabetSize} (with $n + 1$) and Lemma~\ref{LemSupportX},
we obtain (\ref{eqImprovement}).
\end{proof}

\bigskip

Similarly, let ${\cal P}_{n}\big({\cal X}_{n}\times {\cal Y}_{n}, \,c_{X}, \,c_{Y}\big)$
denote the set of all the joint types ${P\mathstrut}_{\!XY} \in {\cal P}_{n}({\cal X}_{n}\times {\cal Y}_{n})$,
such that $\mathbb{E}_{{P\mathstrut}_{\!X}}\!\big[X^{2}\big] \leq c_{X}$
and $\mathbb{E}_{{P\mathstrut}_{\!Y}}\!\big[Y^{2}\big] \leq c_{Y}$. Then its cardinality can be bounded as follows.

\bigskip

\begin{lemma}[Number of joint types] \label{LemNumofJTypes}
\begin{equation} \label{eqNJRypes}
\big|\,  {\cal P}_{n}\big({\cal X}_{n}\times {\cal Y}_{n}, \,c_{X}, \,c_{Y}\big) \,\big| \;\; \leq \;\;
\big((n + 1)c\big)^{\tilde{c}\, n^{(1\,+\, \alpha \, + \, \beta)/2}},
\end{equation}
{\em where $\,c \,\triangleq\, \big[(2\sqrt{c_{X}} + 1)(2\sqrt{c_{Y}} + 1)\big]^{1/(2 \, + \, \alpha \, + \, \beta)}$
and $\,\tilde{c} \,\triangleq\, (2 + \alpha + \beta)\sqrt{2\pi (c_{X} + c_{Y} + 1/6)}$.}
\end{lemma}

\bigskip

\begin{proof}
Denoting $k \, \triangleq \, |\, {\cal X}_{n}(c_{X})\,|\cdot |\, {\cal Y}_{n}(c_{Y})\,|$ and
$\ell \, \triangleq \,
\max_{\,{P\mathstrut}_{\!XY} \, \in \;{\cal P}_{n}({\cal X}_{n}\,\times\, {\cal Y}_{n}, \; c_{X}, \; c_{Y})}
\,|\, {\cal S}({P\mathstrut}_{\!XY}) \,|$,
we repeat the steps of (\ref{eqKchooseL}) and use the bounds of Lemma~\ref{LemAlphabetSize} and Lemma~\ref{LemSupportXY} to
obtain (\ref{eqNJRypes}).
\end{proof}


\section{Converse lemma}\label{ConvLemma}


In this section 
we prove a converse Lemma~\ref{LemConvLem}, which is then used both for the error exponent in Section~\ref{ErrExp}
and for the correct-decoding exponent in Section~\ref{CorDecExp}.

In order to determine exponents in channel probabilities, 
it is convenient to 
take hold of
the exponent in the channel probability {\em density}.
Let ${\bf x} = (x_{1}, x_{2}, .\,.\,. \,, x_{n})\in \mathbb{R}{\mathstrut}^{n}$
be a vector of $n$ channel inputs
and let ${\bf x}^{q} = (x_{1}^{q}, x_{2}^{q}, .\,.\,. \,, x_{n}^{q}) \in {\cal X}_{n}^{n}$
be its quantized version, with components
\begin{equation} \label{eqQuantizer}
x_{k}^{q} \; = \; Q_{\alpha}(x_{k}) \; \triangleq \; \Delta_{\alpha,\,n}\cdot \lfloor x_{k}/\Delta_{\alpha,\,n} + 1/2\rfloor,
\;\;\;\;\;\; k = 1, .\,.\,.\,, n.
\end{equation}
Similarly,
let ${\bf y} = (y_{1}, y_{2}, .\,.\,. \,, y_{n}) \in \mathbb{R}{\mathstrut}^{n}$
be a vector of $n$ channel outputs
and let ${\bf y}^{q} = (y_{1}^{q}, y_{2}^{q}, .\,.\,. \,, y_{n}^{q}) \in {\cal Y}_{n}^{n}$
be its quantized version, with $y_{k}^{q} = Q_{\beta}(y_{k})$ for all $k = 1, .\,.\,.\,, n$.
Then we have the following

\bigskip

\begin{lemma}[PDF exponent] \label{LemPDFExponent}
{\em Let ${\bf x}\in \mathbb{R}{\mathstrut}^{n}$ and ${\bf y}\in \mathbb{R}{\mathstrut}^{n}$
be two channel input and output vectors,
with their respective quantized versions $({\bf x}^{q}, {\bf y}^{q}) \in T({P\mathstrut}_{\!XY})$,
such that $\mathbb{E}_{{P\mathstrut}_{\!XY}}\!\big[(Y-X)^{2}\big]\leq c_{XY}$. Then
}
\begin{align}
&
-\frac{1}{n}\,{\log\mathstrut}_{\!b}\, w({\bf y}^{q} \, | \, {\bf x}^{q})
\, + \, \frac{
(\Delta_{\alpha,\,n} + \Delta_{\beta,\,n})\sqrt{c_{XY}} +
(\Delta_{\alpha,\,n} + \Delta_{\beta,\,n})^{2}/4
}{2\sigma^{2}\ln b}
\;\; \geq \;\;
-\frac{1}{n}\,{\log\mathstrut}_{\!b}\, w({\bf y} \, | \, {\bf x})
\nonumber \\
\geq \;\;
&
-\frac{1}{n}\,{\log\mathstrut}_{\!b}\, w({\bf y}^{q} \, | \, {\bf x}^{q})
\, - \, \frac{(\Delta_{\alpha,\,n} + \Delta_{\beta,\,n})\sqrt{c_{XY}}}{2\sigma^{2}\ln b}.
\nonumber
\end{align}
\end{lemma}


\begin{proof}
The exponent can be equivalently rewritten as
\begin{equation} \label{eqRealW}
-\frac{1}{n}\,{\log\mathstrut}_{\!b}\, w({\bf y} \, | \, {\bf x})
\;\; \equiv \;\;
{\log\mathstrut}_{\!b}\,(\sigma\sqrt{2\pi}) \; + \;
\frac{1}{2\sigma^{2}\ln b}\cdot
\frac{1}{n}\sum_{k \, = \, 1}^{n}(y_{k} - x_{k})^{2}.
\end{equation}
Defining $\delta_{k} \, \triangleq \, (y_{k} - x_{k}) - (y_{k}^{q} - x_{k}^{q})$,
we observe that
\begin{equation} \label{eqRealAverage}
\frac{1}{n}\sum_{k \, = \, 1}^{n}(y_{k} - x_{k})^{2} \;\; = \;\;
\frac{1}{n}\sum_{k \, = \, 1}^{n}(y_{k}^{q} - x_{k}^{q})^{2}
\, + \,
\frac{2}{n}\sum_{k \, = \, 1}^{n}(y_{k}^{q} - x_{k}^{q})\delta_{k}
\, + \,
\frac{1}{n}\sum_{k \, = \, 1}^{n}\delta_{k}^{2}.
\end{equation}
The second term on the RHS is bounded as:
\begin{align}
\bigg|\,
\frac{2}{n}\sum_{k \, = \, 1}^{n}(y_{k}^{q} - x_{k}^{q})\delta_{k}
\,\bigg|
\;\; & \leq \;\;
\frac{2}{n}\sum_{k \, = \, 1}^{n}|\,y_{k}^{q} - x_{k}^{q}\,|\cdot|\,\delta_{k}\,|
\;\; \overset{a}{\leq} \;\;
(\Delta_{\alpha,\,n} + \Delta_{\beta,\,n})\cdot \frac{1}{n}\sum_{k \, = \, 1}^{n}|\,y_{k}^{q} - x_{k}^{q}\,|
\nonumber \\
& \overset{b}{\leq} \;\;
(\Delta_{\alpha,\,n} + \Delta_{\beta,\,n})\bigg[
\frac{1}{n}\sum_{k \, = \, 1}^{n}(y_{k}^{q} - x_{k}^{q})^{2}
\bigg]^{1/2}
\; \overset{c}{\leq} \;\; (\Delta_{\alpha,\,n} + \Delta_{\beta,\,n})\sqrt{c_{XY}},
\label{eqRealBounded}
\end{align}
where ($a$) follows because $|\, \delta_{k} \,| \,\leq\, (\Delta_{\alpha,\,n} + \Delta_{\beta,\,n})/2$,
($b$) follows by Jensen's inequality for the concave ($\cap$) function $f(t) = \sqrt{t}$,
and ($c$) follows by the condition of the lemma.
The third term is bounded as
\begin{equation} \label{eqRealBounded2}
\frac{1}{n}\sum_{k \, = \, 1}^{n}\delta_{k}^{2} \;\; \leq \;\; (\Delta_{\alpha,\,n} + \Delta_{\beta,\,n})^{2}/4.
\end{equation}
Since the 
exponent with the quantized versions
$-\frac{1}{n}\,{\log\mathstrut}_{\!b}\, w({\bf y}^{q} \, | \, {\bf x}^{q})$, in turn,
can also be rewritten similarly to (\ref{eqRealW}), the result of the lemma follows by
(\ref{eqRealW})-(\ref{eqRealBounded2}).
\end{proof}

\bigskip

The following lemma will be used both for the upper bound on the error exponent and for the lower bound on the correct-decoding exponent.

\bigskip

\begin{lemma}[Conditional probability of correct decoding] 
\label{LemConvLem}
{\em Let ${P\mathstrut}_{\!XY} \in {\cal P}_{n}({\cal X}_{n}\times {\cal Y}_{n})$ be a joint type,
such that $\mathbb{E}_{{P\mathstrut}_{\!X}}\!\big[X^{2}\big] \leq c_{X}$, $\mathbb{E}_{{P\mathstrut}_{\!Y}}\!\big[Y^{2}\big] \leq c_{Y}$,
and $\mathbb{E}_{{P\mathstrut}_{\!XY}}\!\big[(Y-X)^{2}\big]\leq c_{XY}$, and
let ${\cal C}$ be a codebook, 
such that the quantized versions (\ref{eqQuantizer}) of 
its codewords ${\bf x}(m)$, $m = 1, 2, .\,.\,.\, , \, M(n, R)$,
are all of the 
type ${P\mathstrut}_{\!X}$, that is:
\begin{displaymath} 
{\bf x}^{q}(m) \; = \; Q_{\alpha}({\bf x}(m))
\; = \;
\big(Q_{\alpha}(x_{1}(m)), Q_{\alpha}(x_{2}(m)), .\,.\,.\, , \,Q_{\alpha}(x_{n}(m))\big)
\; \in \;
T({P\mathstrut}_{\!X}), \;\;\; \forall m.
\end{displaymath}
Let $J \sim \text{Unif}\,\big(\{1, 2, .\,.\,.\, , \, M\}\big)$ be a random variable,
independent of the channel, and let ${\bf x}(J) \rightarrow {\bf Y}$ be the random channel-input and channel-output vectors, respectively.
Let ${\bf Y}^{q} = Q_{\beta}({\bf Y})\in {\cal Y}_{n}^{n}$. Then}
\begin{displaymath}
\Pr \Big\{
g({\bf Y}) = J \; \big| \;
\big({\bf x}^{q}(J), \,{\bf Y}^{q}\big) \, \in \, T({P\mathstrut}_{\!XY})
\Big\}
\;\; \leq \;\;
b^{\,-n\big(\widetilde{R} \, - \, I({P\mathstrut}_{\!XY})  
\, + \, o(1)\big)},
\end{displaymath}
{\em where $\widetilde{R} = \frac{1}{n}\,{\log\mathstrut}_{\!b}\,M(n, R)$,
and $o(1)\rightarrow 0$, as $n\rightarrow \infty$, 
depending only on 
$\alpha$, $\beta$, $c_{X}$, $c_{Y}$, $c_{XY}$, and $\sigma^{2}$.}
\end{lemma}

\bigskip

\begin{proof}
First,
from the single code $({\cal C}, g)$ we create an ensemble of codes, where each member code has the same probability
of error/correct-decoding as the original code $({\cal C}, g)$.
Then we 
upper bound the ensemble average probability of correct decoding.

Considering the codebook ${\cal C}$ as an $M\times n$ matrix, we permute its $n$ columns.
This produces a set of 
codebooks: ${\cal C}_{\ell}$, $\ell = 1, .\,.\,.\, , \, n!$.
The quantized versions of all the codewords of each codebook ${\cal C}_{\ell}$ 
belong to the same type class $T({P\mathstrut}_{\!X})$.
In accordance with ${\cal C}_{\ell}$, we permute also
the $n$ coordinates of each ${\bf y} \in {\cal D}{\mathstrut}_{m}$
of the decision regions ${\cal D}{\mathstrut}_{m}$ in the definition (\ref{eqDec})
of the decoder $g$, 
obtaining
open sets ${\cal D}{\mathstrut}_{m}^{(\ell)}$ and
creating in this way an ensemble of codes $({\cal C}_{\ell}, g_{\ell})$, $\ell = 1, .\,.\,.\, , \, n!$.

Let ${\bf x}_{\ell}(J) \rightarrow {\bf Y}_{\!\ell}$ denote the random channel-input and channel-output vectors, respectively,
when using the code with an index $\ell \in \{1, .\,.\,.\, , \, n!\}$.
Let ${\bf x}_{\ell}^{q}(J)$ and ${\bf Y}_{\!\ell}^{q}$ denote their respective quantized versions.
Since the additive channel noise is i.i.d., permutation of 
components does not change the distribution of the noise vector ${\bf Y}_{\!\ell} - {\bf x}_{\ell}(J)$,
and
we obtain 
\begin{align}
\Pr \Big\{
g({\bf Y}) = J, \;
\big({\bf x}^{q}(J), \,{\bf Y}^{q}\big) \, \in \, T({P\mathstrut}_{\!XY})
\Big\}
\; & = \;
\Pr \Big\{
g_{\ell}({\bf Y}_{\!\ell}) = J, \;
\big({\bf x}_{\ell}^{q}(J), \,{\bf Y}_{\!\ell}^{q}\big) \, \in \, T({P\mathstrut}_{\!XY})
\Big\},
\;\;\; \forall \ell,
\label{eqJointEvent} \\
\Pr \Big\{
\big({\bf x}^{q}(J), \,{\bf Y}^{q}\big) \, \in \, T({P\mathstrut}_{\!XY})
\Big\}
\; & = \;
\Pr \Big\{
\big({\bf x}_{\ell}^{q}(J), \,{\bf Y}_{\!\ell}^{q}\big) \, \in \, T({P\mathstrut}_{\!XY})
\Big\},
\;\;\; \forall \ell.
\label{eqTypeEvent}
\end{align}
Suppose that one of the codes $({\cal C}_{\ell}, g_{\ell})$, $\ell = 1, .\,.\,.\, , \, n!$,
is 
used for communication with probability $1/n!$,
chosen independently of the sent message $J$ and of the channel. Let $L \sim \text{Unif}\,\big(\{1, 2, .\,.\,.\, , \, n!\}\big)$
be the random variable denoting the index of this code.
Then, using (\ref{eqJointEvent}) and (\ref{eqTypeEvent}) we obtain
\begin{align}
\Pr \Big\{
g({\bf Y}) = J \; \big| \;
\big({\bf x}^{q}(J), \,{\bf Y}^{q}\big) \, \in \, T({P\mathstrut}_{\!XY})
\Big\}
\; & = \;
\Pr \Big\{
g_{L}({\bf Y}_{\!L}) = J \; \big| \;
\big({\bf x}_{L}^{q}(J), \,{\bf Y}_{\!L}^{q}\big) \, \in \, T({P\mathstrut}_{\!XY})
\Big\}.
\label{eqEnsembleAverage}
\end{align}
In what follows, we upper bound the RHS of (\ref{eqEnsembleAverage}) with an added condition that ${\bf Y}_{\!L}^{q} = {\bf y} \in T({P\mathstrut}_{\!Y})$:
\begin{displaymath}
\Pr \Big\{\;
g_{L}({\bf Y}_{\!L}) = J \;\; \big| \;\;
{\bf x}_{L}^{q}(J) \in T({P\mathstrut}_{\!X|\,Y}\,|\, {\bf y}), \; {\bf Y}_{\!L}^{q} = {\bf y}
\Big\}.
\end{displaymath}

The total number of codes in the ensemble can be rewritten as
\begin{equation} \label{eqTotalCodes}
n! \;\; = \;\; |\, T({P\mathstrut}_{\!X}) \,| \prod_{x \, \in \, {\cal S}({P\mathstrut}_{\!X})}\big(n{P\mathstrut}_{\!X}(x)\big)!
\;\; \triangleq \;\; |\, T({P\mathstrut}_{\!X}) \,|\cdot \Pi({P\mathstrut}_{\!X}).
\end{equation}
Given ${\bf y} \in T({P\mathstrut}_{\!Y})$, the total number of all the codewords
in the ensemble such that their quantized versions belong to the same conditional type class
$T({P\mathstrut}_{\!X|\,Y}\,|\, {\bf y})$
(counted as distinct if the codewords 
belong to different ensemble member codes or represent different messages) is
given by
\begin{equation} \label{eqTotalCodewords}
S \; = \; \underbrace{|\, T({P\mathstrut}_{\!X|\,Y}\,|\, {\bf y}) \,| \,
\cdot \, \Pi({P\mathstrut}_{\!X})}_{\text{for a message $m$}} \,\, \cdot \,\, M. 
\end{equation}
Let $N(\ell)$ denote the number of the codewords in a codebook ${\cal C}_{\ell}$
such that their quantized versions belong to $T({P\mathstrut}_{\!X|\,Y}\,|\, {\bf y})$.
Given that ${\bf Y}_{\!L}^{q} = {\bf y}$, the channel output vector ${\bf Y}_{\!L}$ falls
into a hypercube region of $\mathbb{R}{\mathstrut}^{n}$:
\begin{displaymath}
{\cal B} \;\; \triangleq \;\; \{\widetilde{\bf y} \in \mathbb{R}{\mathstrut}^{n}: \; Q_{\beta}(\widetilde{\bf y}) = {\bf y}\}.
\end{displaymath}
For any ${\bf x} \in \mathbb{R}{\mathstrut}^{n}$ such that $Q_{\alpha}({\bf x}) \in T({P\mathstrut}_{\!X|\,Y}\,|\, {\bf y})$
and any open region ${\cal D}\subseteq \mathbb{R}{\mathstrut}^{n}$,
by Lemma~\ref{LemPDFExponent} we obtain
\begin{align}
&
b^{\,n\big(\mathbb{E}_{{P\mathstrut}_{\!XY}}\![\,{\log\mathstrut}_{\!b}\, w(Y \, | \, X)\,] \; + \; o_{1}(1)\big)}\cdot\text{vol}({\cal B}\cap {\cal D})
\;\; \geq \;\;
\Pr \big\{
{\bf Y}_{\!L} \in {\cal B}\cap {\cal D} \; | \;
{\bf x}_{L}(J) = {\bf x}
\big\}
\nonumber \\
\geq \;\; &
b^{\,n\big(\mathbb{E}_{{P\mathstrut}_{\!XY}}\![\,{\log\mathstrut}_{\!b}\, w(Y \, | \, X)\,] \; + \; o_{2}(1)\big)}\cdot\text{vol}({\cal B}\cap {\cal D}).
\label{eqBoxProb}
\end{align}
Then, 
since all the 
codes and messages
are equiprobable,
the conditional probability of the code with the index $\ell$ is 
upper-bounded as
\begin{equation} \label{eqCondCode}
\Pr \Big\{\;
L = \ell \;\; \big| \;\;
{\bf x}_{L}^{q}(J) \in T({P\mathstrut}_{\!X|\,Y}\,|\, {\bf y}), \; {\bf Y}_{\!L}^{q} = {\bf y}
\Big\}
\;\; \leq \;\;
b^{\,n 
\big(o_{1}(1)\; - \; o_{2}(1)\big)}N(\ell)/S.
\end{equation}
For $N(\ell) > 0$, let $m_{1} < m_{2} <  .\,.\,. < m_{N(\ell)}$ be the indices of all the codewords in the codebook ${\cal C}_{\ell}$
with their quantized versions in $T({P\mathstrut}_{\!X|\,Y}\,|\, {\bf y})$.
Given that indeed the codebook 
${\cal C}_{\ell}$ has been used for communication,
similarly to (\ref{eqCondCode}),
by (\ref{eqBoxProb})
the conditional probability of correct decoding can be upper-bounded as
\begin{align}
\Pr \Big\{\,
g_{\ell}({\bf Y}_{\!\ell}) = J \; \big| \;
{\bf x}_{\ell}^{q}(J) \in T({P\mathstrut}_{\!X|\,Y}\,|\, {\bf y}), \, {\bf Y}_{\!\ell}^{q} = {\bf y},
\, L = \ell
\Big\}
\,
& \leq
\,
\sum_{j \, = \, 1}^{N(\ell)}
\frac{\text{vol}\big({\cal B}\cap {\cal D}{\mathstrut}_{m_{j}}^{(\ell)}\big)b^{\,n \,\cdot \, o(1)}}
{N(\ell)\text{vol}({\cal B})}
\,
\leq
\,
\frac{b^{\,n \,\cdot \, o(1)}}{N(\ell)},
\label{eqCondCorDec}
\end{align}
where the second inequality follows because the decision regions ${\cal D}{\mathstrut}_{m_{j}}^{(\ell)}$ are disjoint.
Summing up over all the codes, we finally obtain:
\begin{align}
&
\Pr \Big\{\;
g_{L}({\bf Y}_{\!L}) = J \;\; \big| \;\;
{\bf x}_{L}^{q}(J) \in T({P\mathstrut}_{\!X|\,Y}\,|\, {\bf y}), \; {\bf Y}_{\!L}^{q} = {\bf y}
\Big\}
\;\;
\overset{a}{\leq} \;\;
\sum_{\substack{1 \, \leq \, \ell \, \leq \, n! \, : \;
N(\ell)\, > \, 0}}
\frac{N(\ell)}{S}
\cdot
\frac{1}{N(\ell)}\cdot b^{\,n \,\cdot \, o(1)}
\nonumber \\
= \;\; &
\sum_{\substack{1 \, \leq \, \ell \, \leq \, n! \, : \; N(\ell)\, > \, 0}}
\frac{1}{S}\cdot b^{\,n \,\cdot \, o(1)}
\;\; \leq \;\;
\frac{n!}{S}\cdot b^{\,n \,\cdot \, o(1)}
\;\; \overset{b}{=} \;\;
\frac{|\, T({P\mathstrut}_{\!X}) \,|}{|\, T({P\mathstrut}_{\!X|\,Y}\,|\, {\bf y}) \,| \,  \, M}\cdot b^{\,n \,\cdot \, o(1)}
\;\; \overset{c}{\leq} \;\;
b^{\,-n\big(\widetilde{R} \, - \, I({P\mathstrut}_{\!XY})  
\, + \, o(1)\big)},
\nonumber
\end{align}
where ($a$) follows by (\ref{eqCondCode}) and (\ref{eqCondCorDec}),
($b$) follows by (\ref{eqTotalCodes}) and (\ref{eqTotalCodewords}),
and ($c$) follows by (\ref{eqMargExpX}) of Lemma~\ref{LemTypeSize} and (\ref{eqXgivenY}) of Lemma~\ref{LemCondTypeSize}.
\end{proof}

In the next two sections we derive converse bounds on the error and correct-decoding exponents in terms of types.


\section{Error exponent}\label{ErrExp}

The end result of this section is given by Lemma~\ref{LemAllCodebooks} and represents a converse bound on the error exponent by the method of types.

\begin{lemma}[Error exponent of 
mono-composition codebooks]
\label{LemConstComp}
{\em 
Let ${P\mathstrut}_{\!X} \in {\cal P}_{n}({\cal X}_{n})$ be a type,
such that $\mathbb{E}_{{P\mathstrut}_{\!X}}\!\big[X^{2}\big] \leq c_{X}$,
and
let ${\cal C}$ be a codebook, 
such that the quantized versions (\ref{eqQuantizer}) of 
its codewords ${\bf x}(m)$, $m = 1, 2, .\,.\,.\, , \, M(n, R)$,
are all of the 
type ${P\mathstrut}_{\!X}$, that is:
\begin{displaymath} 
{\bf x}^{q}(m) \; = \; Q_{\alpha}({\bf x}(m))
\; = \;
\big(Q_{\alpha}(x_{1}(m)), Q_{\alpha}(x_{2}(m)), .\,.\,.\, , \,Q_{\alpha}(x_{n}(m))\big)
\; \in \;
T({P\mathstrut}_{\!X}), \;\;\; \forall m.
\end{displaymath}
Let $J \sim \text{Unif}\,\big(\{1, 2, .\,.\,.\, , \, M\}\big)$ be a random variable,
independent of the channel, and let ${\bf x}(J) \rightarrow {\bf Y}$ be the random channel-input and channel-output vectors, respectively.
Then for any parameter $c_{XY}$
}
\begin{equation} \label{eqCCBound}
- \frac{1}{n}\,{\log\mathstrut}_{\!b}\Pr \big\{
g({\bf Y}) \neq J \big\}
\;\; \leq \;\;
\min_{\substack{\\{P\mathstrut}_{\!Y|X}:\\
{P\mathstrut}_{\!XY}\,\in \, {\cal P}_{n}({\cal X}_{n}\,\times\, {\cal Y}_{n}),
\\
\mathbb{E}[(Y-X)^{2}] \; \leq \; c_{XY},
\\ I({P\mathstrut}_{\!X}, \, {P\mathstrut}_{\!Y|X}) \; \leq \; \widetilde{R} \, - \, o(1)
}}
\Big\{
D\big({P\mathstrut}_{\!Y|X}\,\|\, {W\mathstrut}_{\!n} \,|\,  {P\mathstrut}_{\!X}\big)
\Big\} \; + \; o(1),
\end{equation}
{\em where
$\widetilde{R} = \frac{1}{n}\,{\log\mathstrut}_{\!b}\,M(n, R)$,
and
$o(1)\rightarrow 0$, as $n\rightarrow \infty$, 
depending only on 
$\alpha$, $\beta$, $c_{X}$, $c_{XY}$, and $\sigma^{2}$.}
\end{lemma}

\bigskip

\begin{proof}
For a joint type ${P\mathstrut}_{\!XY} \in {\cal P}_{n}({\cal X}_{n}\times {\cal Y}_{n})$ with the marginal type ${P\mathstrut}_{\!X}$,
such that
$\mathbb{E}_{{P\mathstrut}_{\!XY}}\!\big[(Y-X)^{2}\big]\leq c_{XY}$, we have also
\begin{displaymath}
\mathbb{E}_{{P\mathstrut}_{\!Y}}\!\big[Y^{2}\big]
\;\; 
\leq
\;\;
\Big(
\mathbb{E}_{{P\mathstrut}_{\!X}}^{1/2}\big[X^{2}\big]\,+\,
\mathbb{E}_{{P\mathstrut}_{\!XY}}^{1/2}\!\big[(Y-X)^{2}\big]
\Big)^{2}
\;\; \leq \;\;
\big(
\sqrt{c_{X}} \, + \, \sqrt{c_{XY}}
\big)^{2}
\;\; \triangleq \;\;
c_{Y}.
\end{displaymath}
Then with $\,{\bf Y}^{q} = Q_{\beta}({\bf Y})\in {\cal Y}_{n}^{n}\,$ for any $1 \leq j \leq M$ we obtain
\begin{align}
\Pr \Big\{
\big({\bf x}^{q}(J), \,{\bf Y}^{q}\big) \, \in \, T({P\mathstrut}_{\!XY}) \; \big| \;
J = j
\Big\}
\;\; & \overset{a}{\geq} \;\;
\big|\, T\big({P\mathstrut}_{\!Y|X}\,|\, {\bf x}^{q}(j)\big) \,\big| \cdot \Delta_{\beta,\,n}^{n}\cdot
b^{\,n\big(\mathbb{E}_{{P\mathstrut}_{\!XY}}\![\,{\log\mathstrut}_{\!b}\, w(Y \, | \, X)\,] \; + \; o(1)\big)}
\nonumber \\
& \overset{b}{\geq} \;\;
b^{\,-n\big(D({P\mathstrut}_{\!Y|X}\,\|\, {W\mathstrut}_{\!n} \,|\,  {P\mathstrut}_{\!X}) \; + \; o(1)\big)},
\;\;\;\;\;\; \forall j,
\label{eqConTypeExp}
\end{align}
where ($a$) follows by Lemma~\ref{LemPDFExponent}, and ($b$) follows by (\ref{eqYgivenX}) of Lemma~\ref{LemCondTypeSize} and (\ref{eqChanApprox}).
This gives
\begin{align}
\Pr \Big\{
\big({\bf x}^{q}(J), \,{\bf Y}^{q}\big) \, \in \, T({P\mathstrut}_{\!XY})
\Big\}
\;\; & \geq \;\;
b^{\,-n\big(D({P\mathstrut}_{\!Y|X}\,\|\, {W\mathstrut}_{\!n} \,|\,  {P\mathstrut}_{\!X}) \; + \; o(1)\big)}.
\label{eqPrior}
\end{align}
Now we are ready to apply Lemma~\ref{LemConvLem}:
\begin{align}
\Pr \big\{
g({\bf Y}) \neq J \big\}
\;\; & \geq \;\;
\Pr \Big\{
\big({\bf x}^{q}(J), \,{\bf Y}^{q}\big) \, \in \, T({P\mathstrut}_{\!XY})
\Big\}\cdot
\Pr \Big\{
g({\bf Y}) \neq J \; \big| \;
\big({\bf x}^{q}(J), \,{\bf Y}^{q}\big) \, \in \, T({P\mathstrut}_{\!XY})
\Big\}
\nonumber \\
& \overset{a}{\geq} \;\;
b^{\,-n\big(D({P\mathstrut}_{\!Y|X}\,\|\, {W\mathstrut}_{\!n} \,|\,  {P\mathstrut}_{\!X}) \; + \; o(1)\big)}
\cdot
\Big[
1 \, - \,
b^{\,-n\big(\widetilde{R} \; - \; I({P\mathstrut}_{\!X}, \, {P\mathstrut}_{\!Y|X})  
\; + \; o(1)\big)}
\Big]
\nonumber \\
& \overset{b}{\geq} \;\;
b^{\,-n\big(D({P\mathstrut}_{\!Y|X}\,\|\, {W\mathstrut}_{\!n} \,|\,  {P\mathstrut}_{\!X}) \; + \; o(1)\big)}
\cdot 1/2,
\nonumber
\end{align}
where ($a$) follows by (\ref{eqPrior}) and Lemma~\ref{LemConvLem},
and ($b$) holds for $I({P\mathstrut}_{\!X}, {P\mathstrut}_{\!Y|X}) \,\leq \, \widetilde{R} - {\log\mathstrut}_{\!b}(2)/n + o(1)$.
\end{proof}

\bigskip

\begin{lemma}[Type constraint] \label{LemTypeConstraint}
{\em
For any $\epsilon > 0$ there exists $n_{0} = n_{0}(\alpha, s^{2}, \epsilon) \in \mathbb{N}$,
such that for any $n > n_{0}$ and any codeword ${\bf x} \in \mathbb{R}{\mathstrut}^{n}$,
satisfying the power constraint (\ref{eqPowerConstraint}), the quantized version of that codeword, 
defined by (\ref{eqQuantizer}),
satisfies the power constraint (\ref{eqPowerConstraint}) within $\epsilon$, that is
with $s^{2}$ replaced by $s^{2} + \epsilon$.
}
\end{lemma}
The proof is the same as (\ref{eqRealAverage})-(\ref{eqRealBounded2}).

\bigskip

\begin{lemma}[Error exponent] \label{LemAllCodebooks}
{\em
Let $J \sim \text{Unif}\,\big(\{1, 2, .\,.\,.\, , \, M\}\big)$ be a random variable,
independent of the channel, and let ${\bf x}(J) \rightarrow {\bf Y}$ be the random channel-input and channel-output vectors, respectively.
Then for any 
$c_{XY}$ and $\epsilon > 0$ there exists
$n_{0} = n_{0}(\alpha, \,\beta, \,s^{2}, \, \sigma^{2}, \,c_{XY}, 
\,\epsilon) \in \mathbb{N}$,
such that for any $n > n_{0}$
}
\begin{equation} \label{eqBoundTypes}
- \frac{1}{n}\,{\log\mathstrut}_{\!b}\Pr \big\{
g({\bf Y}) \neq J \big\}
\;\; \leq \;\;
\max_{\substack{\\{P\mathstrut}_{\!X{\color{white}|}}\!\!:\\
{P\mathstrut}_{\!X}\,\in \, {\cal P}_{n}({\cal X}_{n}),
\\
\mathbb{E}[X^{2}] \; \leq \; s^{2}\, + \, \epsilon}}
\;\;
\min_{\substack{\\{P\mathstrut}_{\!Y|X}:\\
{P\mathstrut}_{\!XY}\,\in \, {\cal P}_{n}({\cal X}_{n}\,\times\, {\cal Y}_{n}),
\\
\mathbb{E}[(Y-X)^{2}] \; \leq \; c_{XY},
\\ I({P\mathstrut}_{\!X}, \, {P\mathstrut}_{\!Y|X}) \; \leq \; R \, - \, \epsilon
}}
\Big\{
D\big({P\mathstrut}_{\!Y|X}\,\|\, {W\mathstrut}_{\!n} \,|\,  {P\mathstrut}_{\!X}\big)
\Big\} \; + \; o(1),
\end{equation}
{\em where $o(1)\rightarrow 0$, as $n\rightarrow \infty$, 
depending only on the parameters
$\alpha$, $\beta$, $s^{2} + \epsilon$, $c_{XY}$, and $\sigma^{2}$.}
\end{lemma}

\bigskip

\begin{proof}
For a type ${P\mathstrut}_{\!X} \in {\cal P}_{n}({\cal X}_{n})$ let us define $\,M({P\mathstrut}_{\!X}) \, \triangleq \,
\big|\,
\big\{
j : \,
{\bf x}^{q}(j) \in T({P\mathstrut}_{\!X})
\big\}\,\big|$.
Then for any $n$ greater than $n_{0}$ of Lemma~\ref{LemTypeConstraint}
there exists at least one type ${P\mathstrut}_{\!X}$ such that
\begin{equation} \label{eqFrequentType}
M({P\mathstrut}_{\!X})
\;\; \geq \;\;
\frac{M}{\big|\,  {\cal P}_{n}\big({\cal X}_{n}, \,s^{2} + \epsilon\big) \,\big|}
\;\; \geq \;\;
\frac{M}{\big((n + 1)c\big)^{\tilde{c}\, n^{(1\,+\, 2\alpha)/3}}},
\end{equation}
where the second inequality follows by Lemma~\ref{LemNumofTypes} applied with $c_{X} = s^{2} + \epsilon$.
Then we can use such a type for a bound:
\begin{align}
& \Pr \big\{
g({\bf Y}) \neq J \big\} \;\; = \;\; \frac{1}{M}\sum_{j \, = \, 1}^{M} \Pr \big\{{\bf Y} \notin {\cal D}{\mathstrut}_{J} \; | \; J = j\big\}
\;\; \geq \;\;
\frac{1}{M}\sum_{\substack{1\,\leq\,j\,\leq\,M \, : \\
{\bf x}^{q}(j) \, \in \, T({P\mathstrut}_{\!X})
}} \Pr \big\{{\bf Y} \notin {\cal D}{\mathstrut}_{J} \; | \; J = j\big\}
\nonumber \\
&
\;\;\;\;\;\;\;\;\;\;\;\;\;\;\;\;\;\;\;\;\;\;\;\;\;\;\;\;\;\;\;\;\;\;\;\;\;\;\;\;\;\;\;\;\;\;\;
\;\;\;\;\;\;\;\;\;\;\;\,\,\,\,
\overset{a}{\geq} \;\;
b^{\,n \,\cdot \, o(1)}\,
\frac{1}{M({P\mathstrut}_{\!X})}\sum_{\substack{1\,\leq\,j\,\leq\,M \, : \\
{\bf x}^{q}(j) \, \in \, T({P\mathstrut}_{\!X})
}} \Pr \big\{{\bf Y} \notin {\cal D}{\mathstrut}_{J} \; | \; J = j\big\}
\nonumber \\
&
\;\;\;\;\;\;\;\;\;\;\;\;\;\;\;\;\;\;\;\;\;\;\;\;\;\;\;\;\;\;\;\;\;\;\;\;\;\;\;\;\;\;\;\;\;\;\;
\;\;\;\;\;\;\;\;\;\;\;\,\,\,\,
\overset{b}{=} \;\;
b^{\,n \,\cdot \, o(1)}\,
\Pr \big\{\,
\widetilde{g}(\widetilde{\bf Y}) \,\neq\, \widetilde{J} \,\big\},
\label{eqSmallerCodebook}
\end{align}
where ($a$) follows by (\ref{eqFrequentType}),
and ($b$) holds for the random variable $\widetilde{J} \sim \text{Unif}\,\big(\{1, 2, .\,.\,.\, , \, M({P\mathstrut}_{\!X})\}\big)$,
independent of the channel,
and the 
channel input/output
random vectors
${\bf x}\big(m_{\widetilde{J}}\big) \rightarrow \widetilde{\bf Y}$
with the decoder
\begin{displaymath} 
\widetilde{g}({\bf y})
\;\; \triangleq \;\;
\left\{
\begin{array}{r l}
0, & \;\;\; {\bf y} \in \bigcap_{\,j \, = \, 1}^{\,M({P\mathstrut}_{\!X})} {\cal D}{\mathstrut}_{m_{j}}^{c}, \\
j, & \;\;\; {\bf y} \in {\cal D}{\mathstrut}_{m_{j}}, \;\;\; 
j \in \{1, 2, .\,.\,.\, , \, M({P\mathstrut}_{\!X})\},
\end{array}
\right.
\end{displaymath}
where $m_{1} < .\,.\,. < m_{M({P\mathstrut}_{\!X})}$ are the indices of the codewords in ${\cal C}$ with
their quantized versions in $T({P\mathstrut}_{\!X})$.

It follows now from (\ref{eqSmallerCodebook}) that the LHS of (\ref{eqBoundTypes}) can be upper-bounded by
(\ref{eqCCBound}) of Lemma~\ref{LemConstComp}
with $\widetilde{R} = \frac{1}{n}\,{\log\mathstrut}_{\!b}\,M({P\mathstrut}_{\!X})$.
Substituting then (\ref{eqFrequentType}) in place of $M({P\mathstrut}_{\!X})$
we obtain a stricter condition under the minimum of (\ref{eqCCBound}), leading to 
an upper bound with
a condition $I({P\mathstrut}_{\!X}, {P\mathstrut}_{\!Y|X}) \,\leq \, R - o(1)$
and to (\ref{eqBoundTypes}).
\end{proof}


\section{Correct-decoding exponent}\label{CorDecExp}

The end result of this section is Lemma~\ref{LemCorDec}, which is a converse bound on the correct-decoding exponent by the method of types.

\begin{lemma}[Joint type constraint] \label{LemTypeNoise}
{\em
For any $\epsilon > 0$ and $\widetilde{\sigma}^{2}$ there exists $n_{0} = n_{0}(\alpha, \,\beta, \,\widetilde{\sigma}^{2}, \,\epsilon) \in \mathbb{N}$,
such that for any $n > n_{0}$ and any pair of vectors
${\bf x}, {\bf y} \in \mathbb{R}{\mathstrut}^{n}$ 
satisfying}
\begin{displaymath}
\frac{1}{n}\sum_{k \, = \, 1}^{n}(y_{k} - x_{k})^{2} \;\; \leq \;\; \widetilde{\sigma}^{2},
\end{displaymath}
{\em the respective quantized versions ${\bf x}^{q} = Q_{\alpha}({\bf x})$ and ${\bf y}^{q} = Q_{\beta}({\bf y})$,
defined as in (\ref{eqQuantizer}),
satisfy}
\begin{displaymath}
\frac{1}{n}\sum_{k \, = \, 1}^{n}(y_{k}^{q} - x_{k}^{q})^{2} \;\; \leq \;\; \widetilde{\sigma}^{2} \, + \, \epsilon.
\end{displaymath}
\end{lemma}
The proof is the same as (\ref{eqRealAverage})-(\ref{eqRealBounded2}).


We use a Chernoff bound for the probability of an event when the method of types cannot be applied:

\begin{lemma}[Chernoff bound] \label{LemChernoff}
{\em Let $Z_{k} \sim {\cal N}(0, \sigma^{2})$, $k = 1, 2, .\,.\,.\, , \, n$, be $n$ independent random variables.
Then for $\,\widetilde{\sigma}^{2} \geq \sigma^{2}\,$
and $\,f(x) \,\triangleq\, \tfrac{1}{2}\big[x - 1 - \ln(x)\big]$:
}
\begin{displaymath}
\Pr\bigg\{
\frac{1}{n}\sum_{k \, = \, 1}^{n}Z_{k}^{2} \; \geq \; \widetilde{\sigma}^{2}
\bigg\}
\;\; \leq \;\; \exp\big\{-nf(\widetilde{\sigma}^{2}/\sigma^{2})\big\}.
\end{displaymath}
\end{lemma}


\begin{lemma}[Correct-decoding exponent of 
mono-composition codebooks]
\label{LemConstComp2}
{\em 
Let ${P\mathstrut}_{\!X} \in {\cal P}_{n}({\cal X}_{n})$ be a type,
such that $\mathbb{E}_{{P\mathstrut}_{\!X}}\!\big[X^{2}\big] \leq c_{X}$,
and
let ${\cal C}$ be a codebook, 
such that the quantized versions (\ref{eqQuantizer}) of 
its codewords ${\bf x}(m)$, $m = 1, 2, .\,.\,.\, , \, M(n, R)$,
are all of the 
type ${P\mathstrut}_{\!X}$, that is:
\begin{displaymath} 
{\bf x}^{q}(m) \; = \; Q_{\alpha}({\bf x}(m))
\; = \;
\big(Q_{\alpha}(x_{1}(m)), Q_{\alpha}(x_{2}(m)), .\,.\,.\, , \,Q_{\alpha}(x_{n}(m))\big)
\; \in \;
T({P\mathstrut}_{\!X}), \;\;\; \forall m.
\end{displaymath}
Let $J \sim \text{Unif}\,\big(\{1, 2, .\,.\,.\, , \, M\}\big)$ be a random variable,
independent of the channel, and let ${\bf x}(J) \rightarrow {\bf Y}$ be the random channel-input and channel-output vectors, respectively.
Then for any $\epsilon > 0$ and $\,\widetilde{\sigma}^{2} \geq \sigma^{2}$
there exists $n_{0} = n_{0}(\alpha, \,\beta, \,\widetilde{\sigma}^{2}, \,\epsilon) \in \mathbb{N}$,
such that for any $n > n_{0}$
}
\begin{align}
- \frac{1}{n}\,{\log\mathstrut}_{\!b}\Pr \big\{
g({\bf Y}) = J \big\}
\;\; \geq \;\; &
\min \big\{E_{n}({P\mathstrut}_{\!X}, \,R, \,\widetilde{\sigma}, \, \epsilon), \; E(\widetilde{\sigma})\big\}
\, + \, o(1),
\label{eqMinofTwoExponents} \\
E(\widetilde{\sigma}) \;\; \triangleq \;\; & \frac{1}{2\ln b}
\big[\,\widetilde{\sigma}^{2}/\sigma^{2} \, - \, 1 \, - \, \ln\,(\widetilde{\sigma}^{2}/\sigma^{2})
\,\big],
\label{eqOutlierExp} \\
E_{n}({P\mathstrut}_{\!X}, \,R, \,\widetilde{\sigma}, \, \epsilon)
\;\; \triangleq \;\; &
\min_{\substack{\\{P\mathstrut}_{\!Y|X}:\\
{P\mathstrut}_{\!XY}\,\in \, {\cal P}_{n}({\cal X}_{n}\,\times\, {\cal Y}_{n}),
\\
\mathbb{E}[(Y-X)^{2}] \; \leq \; \widetilde{\sigma}^{2} \, + \, \epsilon
}}
\Big\{
D\big({P\mathstrut}_{\!Y|X}\,\|\, {W\mathstrut}_{\!n} \,|\,  {P\mathstrut}_{\!X}\big)
\, + \,
\big|\,
R \, - \, I\big({P\mathstrut}_{\!X}, {P\mathstrut}_{\!Y|X}\big)
\,\big|^{+}
\Big\}, 
\label{eqPartialExp}
\end{align}
{\em where $|\,t\,|^{+}\triangleq \max\,\{0, t\}$,
and
$o(1)\rightarrow 0$, as $n\rightarrow \infty$, 
depending only on 
$\alpha$, $\beta$, $c_{X}$, $\widetilde{\sigma}^{2} + \epsilon$, and $\sigma^{2}$.}
\end{lemma}

\bigskip

\begin{proof}
We consider probabilities of two disjoint events:
\begin{align}
\Pr \big\{
g({\bf Y}) = J \big\}
\;\, & \leq \;\,
\Pr \big\{
g({\bf Y}) = J, \;
\|{\bf Y} - {\bf x}(J)\|^{2}
\leq n\widetilde{\sigma}^{2}
\big\} \; + \;
\Pr \big\{
\|{\bf Y} - {\bf x}(J)\|^{2}
>
n\widetilde{\sigma}^{2}
\big\}
\nonumber \\
& \leq \;\,
2 \max \Big\{
\Pr \big\{
g({\bf Y}) = J, \;
\|{\bf Y} - {\bf x}(J)\|^{2}
\leq n\widetilde{\sigma}^{2}
\big\}, \;\,
\Pr \big\{
\|{\bf Y} - {\bf x}(J)\|^{2}
\geq n\widetilde{\sigma}^{2}
\big\}
\Big\}.
\label{eqTwoEvents}
\end{align}
For the first term in the maximum we obtain:
\begin{align}
&
\;\,\,
\Pr \big\{
g({\bf Y}) = J, \;
\|{\bf Y} - {\bf x}(J)\|^{2}
\leq n\widetilde{\sigma}^{2}
\big\}
\;\; \overset{a}{\leq} \;\;
\Pr \big\{
g({\bf Y}) = J, \;
\|{\bf Y}^{q} - {\bf x}^{q}(J)\|^{2}
\leq n(\widetilde{\sigma}^{2} + \epsilon)
\big\}
\nonumber \\
&
\;\;\;\;\;\;\;\;\;\;\;\;\;\;\;\;\;\;\;\;\;\;\;\;\;\;\;\;\;\;\;\;\;\;\;\;\;\;\,
\overset{b}{=}
\sum_{\substack{{P\mathstrut}_{\!Y|X}:\\
{P\mathstrut}_{\!XY}\,\in \, {\cal P}_{n}({\cal X}_{n}\,\times\, {\cal Y}_{n}),
\\
\mathbb{E}[(Y-X)^{2}] \; \leq \; \widetilde{\sigma}^{2} \, + \, \epsilon
}}
\Pr \big\{
g({\bf Y}) = J, \;
\big({\bf x}^{q}(J), \,{\bf Y}^{q}\big) \, \in \, T({P\mathstrut}_{\!XY})
\big\}
\nonumber \\
& \overset{c}{\leq} \;\;
\big|\,  {\cal P}_{n}\big({\cal X}_{n}\times {\cal Y}_{n}, \,c_{X}, \,c_{Y}\big) \,\big|
\max_{\substack{\\{P\mathstrut}_{\!Y|X}:\\
{P\mathstrut}_{\!XY}\,\in \, {\cal P}_{n}({\cal X}_{n}\,\times\, {\cal Y}_{n}),
\\
\mathbb{E}[(Y-X)^{2}] \; \leq \; \widetilde{\sigma}^{2} \, + \, \epsilon
}}
\Pr \big\{
g({\bf Y}) = J, \;
\big({\bf x}^{q}(J), \,{\bf Y}^{q}\big) \, \in \, T({P\mathstrut}_{\!XY})
\big\},
\label{eqMaxoverTypes}
\end{align}
where ($a$) holds
with ${\bf Y}^{q} = Q_{\beta}({\bf Y})\in {\cal Y}_{n}^{n}\,$
for all $n > n_{0}(\alpha, \,\beta, \,\widetilde{\sigma}^{2}, \,\epsilon)$ of Lemma~\ref{LemTypeNoise},
($b$) holds because ${\bf x}^{q}(J) \in T({P\mathstrut}_{\!X})$,
in ($c$) we use the notation 
${\cal P}_{n}\big({\cal X}_{n}\times {\cal Y}_{n}, \,c_{X}, \,c_{Y}\big)$ for the set of all
the joint types ${P\mathstrut}_{\!XY} \in {\cal P}_{n}({\cal X}_{n}\times {\cal Y}_{n})$
satisfying both $\mathbb{E}_{{P\mathstrut}_{\!X}}\!\big[X^{2}\big] \leq c_{X}$
and
$\mathbb{E}_{{P\mathstrut}_{\!Y}}\!\big[Y^{2}\big] \leq \big(\sqrt{c_{X}} + \sqrt{\widetilde{\sigma}^{2} + \epsilon}\big)^{2} \triangleq c_{Y}$.
By the same steps as in (\ref{eqConTypeExp}) we further obtain
\begin{align}
\Pr \Big\{
\big({\bf x}^{q}(J), \,{\bf Y}^{q}\big) \, \in \, T({P\mathstrut}_{\!XY})
\Big\}
\;\; & \leq \;\;
b^{\,-n\big(D({P\mathstrut}_{\!Y|X}\,\|\, {W\mathstrut}_{\!n} \,|\,  {P\mathstrut}_{\!X}) \; + \; o(1)\big)},
\label{eqPrior2}
\end{align}
while Lemma~\ref{LemConvLem} gives
\begin{equation} \label{eqPositivePart}
\Pr \Big\{
g({\bf Y}) = J \; \big| \;
\big({\bf x}^{q}(J), \,{\bf Y}^{q}\big) \, \in \, T({P\mathstrut}_{\!XY})
\Big\}
\;\; \leq \;\;
b^{\,-n\big(|\, 
R \; - \; I({P\mathstrut}_{\!XY})\,
|{\mathstrut}^{+}
\, + \; o(1)\big)}.
\end{equation}
Now by Lemma~\ref{LemNumofJTypes} for the number of joint types and (\ref{eqMaxoverTypes})-(\ref{eqPositivePart})
we obtain
\begin{displaymath}
\Pr \big\{
g({\bf Y}) = J, \;
\|{\bf Y} - {\bf x}(J)\|^{2}
\leq n\widetilde{\sigma}^{2}
\big\}
\;\; 
\leq
\;\;
b^{\,-n\big(
E_{n}({P\mathstrut}_{\!X}, \,R, \,\widetilde{\sigma}, \, \epsilon) \, + \, o(1) \big)
},
\end{displaymath}
where $E_{n}({P\mathstrut}_{\!X}, \,R, \,\widetilde{\sigma}, \, \epsilon)$ denotes the expression (\ref{eqPartialExp}).
Applying Lemma~\ref{LemChernoff} to the second term in the maximum of (\ref{eqTwoEvents})
we obtain 
(\ref{eqMinofTwoExponents})-(\ref{eqPartialExp}).
\end{proof}


\begin{lemma}[Correct-decoding exponent] \label{LemCorDec}
{\em
Let $J \sim \text{Unif}\,\big(\{1, 2, .\,.\,.\, , \, M\}\big)$ be a random variable,
independent of the channel, and let ${\bf x}(J) \rightarrow {\bf Y}$ be the random channel-input and channel-output vectors, respectively.
Then for any $\widetilde{\epsilon}, \,\epsilon > 0$ and $\,\widetilde{\sigma}^{2} \geq \sigma^{2}$
there exists $n_{0} = n_{0}(\alpha, \,\beta, \,s^{2}, \,\widetilde{\sigma}^{2}, \, \widetilde{\epsilon}, \,\epsilon) \in \mathbb{N}$,
such that for any $n > n_{0}$}
\begin{align}
- \frac{1}{n}\,{\log\mathstrut}_{\!b}\Pr \big\{
g({\bf Y}) = J \big\}
\;\; \geq \;\; &
\min \big\{E_{n}(R, \,\widetilde{\sigma}, \, \widetilde{\epsilon}, \, \epsilon), \; E(\widetilde{\sigma})\big\}
\, + \, o(1),
\label{eqMinofTwoExponents2} \\
E_{n}(R, \,\widetilde{\sigma}, \, \widetilde{\epsilon}, \, \epsilon)
\;\; \triangleq \;\; &
\min_{\substack{\\{P\mathstrut}_{\!X{\color{white}|}}\!\!:\\
{P\mathstrut}_{\!X}\,\in \, {\cal P}_{n}({\cal X}_{n}),
\\
\mathbb{E}[X^{2}] \; \leq \; s^{2}\, + \, \epsilon}}
E_{n}({P\mathstrut}_{\!X}, \,R, \,\widetilde{\sigma}, \,\widetilde{\epsilon}\,)
\label{eqPartialExp2}
\end{align}
{\em where $E(\widetilde{\sigma})$ and $E_{n}({P\mathstrut}_{\!X}, \,R, \,\widetilde{\sigma}, \, \widetilde{\epsilon}\,)$
are as defined in (\ref{eqOutlierExp}) and (\ref{eqPartialExp}), respectively,
and
$o(1)\rightarrow 0$, as $n\rightarrow \infty$, 
depending only on the parameters
$\alpha$, $\beta$, $s^{2} + \epsilon$, $\widetilde{\sigma}^{2} + \widetilde{\epsilon}$, and $\sigma^{2}$.}
\end{lemma}


\begin{proof}
Similarly as in \cite[Lemma~5]{DueckKorner79}:
\begin{align}
\Pr \big\{
g({\bf Y}) = J \big\} \; = \;
\frac{1}{M}
\sum_{\substack{
{P\mathstrut}_{\!X}\,\in \, {\cal P}_{n}({\cal X}_{n})}}
&
\sum_{\substack{1\,\leq\,j\,\leq\,M \, : \\
{\bf x}^{q}(j) \, \in \, T({P\mathstrut}_{\!X})
}} \Pr \big\{{\bf Y} \in {\cal D}{\mathstrut}_{J} \; | \; J = j\big\}
\nonumber \\
\overset{a}{\leq} \;
\big|\,  {\cal P}_{n}\big({\cal X}_{n}, \,s^{2} + \epsilon\big) \,\big| \cdot
\frac{1}{M}
\max_{\substack{
\\{P\mathstrut}_{\!X}\,\in \, {\cal P}_{n}({\cal X}_{n})}}
&
\sum_{\substack{1\,\leq\,j\,\leq\,M \, : \\
{\bf x}^{q}(j) \, \in \, T({P\mathstrut}_{\!X})
}} \Pr \big\{{\bf Y} \in {\cal D}{\mathstrut}_{J} \; | \; J = j\big\}
\label{eqMaximumoverTypes} \\
\overset{b}{\leq} \;
b^{\,n \,\cdot \, o(1)}\,\frac{1}{M}
&
\sum_{\substack{1\,\leq\,j\,\leq\,M \, : \\
{\bf x}^{q}(j) \, \in \, T({P\mathstrut}_{\!X}^{*})
}} \Pr \big\{{\bf Y} \in {\cal D}{\mathstrut}_{J} \; | \; J = j\big\}
\; \overset{c}{=} \;
b^{\,n \,\cdot \, o(1)}
\Pr \big\{
\widetilde{g}(\widetilde{\bf Y}) = J \big\},
\label{eqCorDecCC}
\end{align}
where:

($a$) holds for $n > n_{0}(\alpha, s^{2}, \epsilon)$ of Lemma~\ref{LemTypeConstraint};

($b$) follows by Lemma~\ref{LemNumofTypes} 
with $c_{X} = s^{2} + \epsilon$,
while ${P\mathstrut}_{\!X}^{*} \in {\cal P}_{n}({\cal X}_{n})$
is a maximizer of (\ref{eqMaximumoverTypes});

($c$) holds for
the channel input/output random vectors $\widetilde{\bf x}(J) \rightarrow \widetilde{\bf Y}$ and
a code $(\widetilde{\cal C}, \widetilde{g})$, such that
$\widetilde{\bf x}(j) = {\bf x}(m_{j})$ with $\widetilde{\cal D}{\mathstrut}_{j} = {\cal D}{\mathstrut}_{m_{j}}$ for
$1 \leq j \leq M({P\mathstrut}_{\!X}^{*})$,
and $\widetilde{\bf x}^{q}(j) \in T({P\mathstrut}_{\!X}^{*})$ with $\widetilde{\cal D}{\mathstrut}_{j} = \varnothing$
for $M({P\mathstrut}_{\!X}^{*}) < j \leq M$, where
$m_{1} < .\,.\,. < m_{M({P\mathstrut}_{\!X}^{*})}$ are the indices of the codewords
in the original codebook ${\cal C}$
with their quantized versions in $T({P\mathstrut}_{\!X}^{*})$.
Since all the codewords of $\widetilde{\cal C}$ have their quantized versions in $T({P\mathstrut}_{\!X}^{*})$,
we can apply Lemma~\ref{LemConstComp2} with $c_{X} = s^{2} + \epsilon$ for the RHS of (\ref{eqCorDecCC}) to obtain (\ref{eqMinofTwoExponents2}), (\ref{eqPartialExp2}).
\end{proof}


\section{PDF to type}\label{PDFtypePDF}


Lemmas~\ref{LemPDFtoT} and~\ref{LemTtoPDF} of this section relate between minimums over types and over PDF's.
The next Lemma~\ref{LemQuant}, which has a laborious proof, is required only in the proof of 
Lemma~\ref{LemPDFtoT}, used for Theorem~\ref{thmErrorExp}.

\bigskip

\begin{lemma}[Quantization of PDF] \label{LemQuant}
{\em Let ${\cal X}_{n}$ be an alphabet defined as in (\ref{eqAlphabets}), (\ref{eqDelta})
with $\alpha \in \big(0, \tfrac{1}{4}\big)$.
Let ${P\mathstrut}_{\!X} \in {\cal P}_{n}({\cal X}_{n})$ be a type and ${p\mathstrut}_{Y|X}(\,\cdot\,|\,x) \in {\cal L}$, $\forall x \in {\cal X}_{n}\,$,
be a collection of functions from (\ref{eqLipschitz}), such that $\mathbb{E}_{{P\mathstrut}_{\!X}}\!\big[X^{2}\big] \leq c_{X}$,
 $\mathbb{E}_{{P\mathstrut}_{\!X}{p\mathstrut}_{Y|X}}\!\big[Y^{2}\big]\leq c_{Y}$,
and $\mathbb{E}_{{P\mathstrut}_{\!X}{p\mathstrut}_{Y|X}}\!\big[(Y-X)^{2}\big]\leq c_{XY}$.
Then
for any alphabet ${\cal Y}_{n}$ defined as in (\ref{eqAlphabets}), (\ref{eqDelta})
with $\tfrac{1}{3} + \tfrac{2}{3}\alpha < \beta < \tfrac{2}{3} - \tfrac{2}{3}\alpha$,
there exists a joint type ${P\mathstrut}_{\!XY} \in {\cal P}_{n}({\cal X}_{n}\times {\cal Y}_{n})$ with the marginal type ${P\mathstrut}_{\!X}$,
such that}
\begin{align}
\sum_{x\,\in\,{\cal X}_{n}}{P\mathstrut}_{\!X}(x)\int_{\mathbb{R}}{p\mathstrut}_{Y|X}(y\,|\,x)
{\log\mathstrut}_{\!b}\,{p\mathstrut}_{Y|X}(y\,|\,x)dy
\;\; & \geq \;\;
\sum_{\substack{x\,\in\,{\cal X}_{n}\\
y\,\in\,{\cal Y}_{n}}}
{P\mathstrut}_{\!XY}(x, y){\log\mathstrut}_{\!b}\,\frac{{P\mathstrut}_{\!Y|X}(y\,|\,x)}{\Delta_{\beta,\,n}}
\, + \, o(1),
\label{eqCondEntropy} \\
\sum_{x\,\in\,{\cal X}_{n}}{P\mathstrut}_{\!X}(x)\int_{\mathbb{R}}{p\mathstrut}_{Y|X}(y\,|\,x)
(y - x)^{2}dy
\;\; & \geq \;\;
\sum_{\substack{x\,\in\,{\cal X}_{n}\\
y\,\in\,{\cal Y}_{n}}}
{P\mathstrut}_{\!XY}(x, y)(y - x)^{2} \, + \, o(1),
\label{eqLogGaussian} \\
\int_{\mathbb{R}}{p\mathstrut}_{Y}(y)
{\log\mathstrut}_{\!b}\,{p\mathstrut}_{Y}(y)dy
\;\; & \leq \;\;
\sum_{y\,\in\,{\cal Y}_{n}}
{P\mathstrut}_{\!Y}(y){\log\mathstrut}_{\!b}\,\frac{{P\mathstrut}_{\!Y}(y)}{\Delta_{\beta,\,n}}
\, + \, o(1),
\label{eqMargEntropy}
\end{align}
{\em where ${p\mathstrut}_{Y}(y) = \sum_{\,x\,\in\,{\cal X}_{n}}{P\mathstrut}_{\!X}(x){p\mathstrut}_{Y|X}(y\,|\,x)$, $\forall y \in \mathbb{R}$,
and $o(1)\rightarrow 0$, as $n\rightarrow \infty$, and depends only on the parameters $\alpha$, $\beta$, $c_{X}$, $c_{Y}$, $c_{XY}$,
and $\sigma^2$ (through (\ref{eqLipschitz})).}
\end{lemma}
The proof is given in the Appendix B.

\bigskip

\begin{lemma}[PDF to type]\label{LemPDFtoT}
{\em Let ${\cal X}_{n}$ and ${\cal Y}_{n}$ be two alphabets defined as in (\ref{eqAlphabets}), (\ref{eqDelta})
with $\alpha \in \big(0, \tfrac{1}{4}\big)$
and $\tfrac{1}{3} + \tfrac{2}{3}\alpha < \beta < \tfrac{2}{3} - \tfrac{2}{3}\alpha$.
Then for any $c_{X}$, $c_{XY}$, and $\epsilon > 0$ there exists $n_{0} = n_{0}(\alpha, \,\beta, \,c_{X}, \,c_{XY}, \,\sigma^{2}, \,\epsilon) \in \mathbb{N}$,
such that for any $n > n_{0}$ and
for any type ${P\mathstrut}_{\!X} \in {\cal P}_{n}({\cal X}_{n})$ with $\mathbb{E}_{{P\mathstrut}_{\!X}}\!
\big[X^{2}\big] \leq c_{X}$:
}
\begin{equation} \label{eqTransitionB}
\inf_{\substack{{p\mathstrut}_{Y|X}:\\
{p\mathstrut}_{Y|X}(\, \cdot \, \,|\, x) \, \in \, {\cal L}, \; \forall x,
\\
\mathbb{E}_{{P\mathstrut}_{\!X}{p\mathstrut}_{Y|X}}[(Y-X)^{2}] \; \leq \; c_{XY}\! ,
\\ I({P\mathstrut}_{\!X}, \,\, {p\mathstrut}_{Y|X}) \; \leq \; R \, - \, 2\epsilon
}}
\!\!\!\!
\Big\{
D\big(\,{p\mathstrut}_{Y|X}\,\|\, \, w \, \,|\, {P\mathstrut}_{\!X}\big)
\Big\}
\; \;
\geq
\min_{\substack{\\{P\mathstrut}_{\!Y|X}:\\
{P\mathstrut}_{\!XY}\,\in \, {\cal P}_{n}({\cal X}_{n}\,\times\, {\cal Y}_{n}),
\\
\mathbb{E}_{{P\mathstrut}_{\!XY}}[(Y-X)^{2}] \; \leq \; c_{XY} \, + \, \epsilon,    
\\ I({P\mathstrut}_{\!X}, \, {P\mathstrut}_{\!Y|X}) \; \leq \; R \, - \, \epsilon
}}
\!\!\!\!
\Big\{
D\big({P\mathstrut}_{\!Y|X}\,\|\, {W\mathstrut}_{\!n} \,|\,  {P\mathstrut}_{\!X}\big)
\Big\} \; + \; o(1),
\end{equation}
{\em where $o(1)\rightarrow 0$, as $n\rightarrow \infty$,
and depends only on the parameters $\alpha$, $\beta$, $c_{X}$, $c_{XY}$, and $\sigma^{2}$.}
\end{lemma}

\bigskip

\begin{proof}
For a type ${P\mathstrut}_{\!X}$ with a collection of ${p\mathstrut}_{Y|X}$ such that $\mathbb{E}_{{P\mathstrut}_{\!X}}\!
\big[X^{2}\big] \leq c_{X}$
and $\mathbb{E}_{{P\mathstrut}_{\!X}{p\mathstrut}_{Y|X}}\!\big[(Y-X)^{2}\big]\leq c_{XY}$,
we can find also an upper bound $c_{Y}$ on $\mathbb{E}_{{P\mathstrut}_{\!X}{p\mathstrut}_{Y|X}}\!\big[Y^{2}\big]$.
For example, using the Cauchy-Schwarz inequality:
\begin{displaymath}
\mathbb{E}_{{P\mathstrut}_{\!X}{p\mathstrut}_{Y|X}}\!\big[Y^{2}\big]
\;\; 
\leq
\;\;
\Big(
\mathbb{E}_{{P\mathstrut}_{\!X}}^{1/2}\big[X^{2}\big]\,+\,
\mathbb{E}_{{P\mathstrut}_{\!X}{p\mathstrut}_{Y|X}}^{1/2}\!\big[(Y-X)^{2}\big]
\Big)^{2}
\;\; \leq \;\;
\big(
\sqrt{c_{X}} \, + \, \sqrt{c_{XY}}
\big)^{2}
\;\; \triangleq \;\;
c_{Y}.
\end{displaymath}
Then by Lemma~\ref{LemQuant} there exists a joint type ${P\mathstrut}_{\!XY}$
with the marginal type ${P\mathstrut}_{\!X}$, such that simultaneously
the three inequalities (\ref{eqCondEntropy})-(\ref{eqMargEntropy}) are satisfied
and it also follows by (\ref{eqLogGaussian}) and (\ref{eqChanApprox}) that
\begin{equation} \label{eqConseq4}
-\mathbb{E}_{{P\mathstrut}_{\!X}{p\mathstrut}_{Y|X}}\!\big[{\log\mathstrut}_{\!b}\,w(Y\,|\,X)\big]
\;\; \geq \;\;
-\mathbb{E}_{{P\mathstrut}_{\!XY}}\!\big[{\log\mathstrut}_{\!b}\,{W\mathstrut}_{\!n}(Y\,|\,X)\big]
 \, + \,
{\log\mathstrut}_{\!b}\,\Delta_{\beta,\,n}
\, + \, o(1).
\end{equation}
Then the sum of (\ref{eqCondEntropy}) and (\ref{eqConseq4}) gives
\begin{equation} \label{eqObjective}
D\big(\,{p\mathstrut}_{Y|X}\,\|\, \, w \, \,|\, {P\mathstrut}_{\!X}\big)
\;\; \geq \;\;
D\big({P\mathstrut}_{\!Y|X}\,\|\, {W\mathstrut}_{\!n} \,|\,  {P\mathstrut}_{\!X}\big)
\, + \, o(1),
\end{equation}
while the difference of
(\ref{eqCondEntropy})
and
(\ref{eqMargEntropy}) gives
\begin{equation} \label{eqMutualIneq}
I\big({P\mathstrut}_{\!X}, {P\mathstrut}_{\!Y|X}\big)
\;\; \leq \;\;
I\big({P\mathstrut}_{\!X}, \, {p\mathstrut}_{Y|X}\big)
\, + \, o(1).
\end{equation}
Note that all $o(1)$ in the above relations are independent 
of the joint type ${P\mathstrut}_{\!XY}$
and the functions
${p\mathstrut}_{Y|X}$. Therefore
by (\ref{eqObjective}), (\ref{eqMutualIneq}), and (\ref{eqLogGaussian})
we conclude, that
given any $\epsilon > 0$ for $n$ sufficiently large
for every type ${P\mathstrut}_{\!X}$ with the prerequisites of this lemma and every collection of ${p\mathstrut}_{Y|X}$
which satisfy the conditions under the infimum on the LHS of (\ref{eqTransitionB})
there exists a joint type ${P\mathstrut}_{\!XY}$
such that simultaneously
\begin{align}
I\big({P\mathstrut}_{\!X}, {P\mathstrut}_{\!Y|X}\big)
\;\; & \leq \;\;
I\big({P\mathstrut}_{\!X}, \, {p\mathstrut}_{Y|X}\big)
\, + \, \epsilon,
\nonumber \\
\mathbb{E}_{{P\mathstrut}_{\!XY}}\!\big[(Y - X)^{2}\big]
\;\; & \leq \;\;
\mathbb{E}_{{P\mathstrut}_{\!X}{p\mathstrut}_{Y|X}}\!\big[(Y - X)^{2}\big]
\, + \, \epsilon,
\nonumber
\end{align}
and (\ref{eqObjective}) 
holds with a uniform $o(1)$, i.e., independent of ${P\mathstrut}_{\!XY}$ and ${p\mathstrut}_{Y|X}$.
It follows that such
${P\mathstrut}_{\!XY}$
satisfies also the conditions under the {\em minimum} on the RHS of (\ref{eqTransitionB})
and results in the objective function of (\ref{eqTransitionB}) 
satisfying (\ref{eqObjective}) with the uniform $o(1)$.
Then the minimum itself, which can only possibly be taken over 
a greater variety of ${P\mathstrut}_{\!XY}$,
satisfies the inequality (\ref{eqTransitionB}).
\end{proof}

\bigskip

\begin{lemma}[Type to PDF]\label{LemTtoPDF}
{\em For any $c_{XY}$ and $\epsilon > 0$ there exists $n_{0} = n_{0}(\beta, \,c_{XY}, \,\epsilon) \in \mathbb{N}$,
such that for any $n > n_{0}$ and
for any type ${P\mathstrut}_{\!X} \in {\cal P}_{n}({\cal X}_{n})$:
}
\begin{align}
\min_{\substack{\\{P\mathstrut}_{\!Y|X}:\\
{P\mathstrut}_{\!XY}\,\in \, {\cal P}_{n}({\cal X}_{n}\,\times\, {\cal Y}_{n}),
\\
\mathbb{E}_{{P\mathstrut}_{\!XY}}[(Y-X)^{2}] \; \leq \; c_{XY}
}}
\;\;\;\;
&
\Big\{
D\big({P\mathstrut}_{\!Y|X}\,\|\, {W\mathstrut}_{\!n} \,|\,  {P\mathstrut}_{\!X}\big)
\, + \,
\big|\,
R \, - \, I\big({P\mathstrut}_{\!X}, {P\mathstrut}_{\!Y|X}\big)
\,\big|^{+}
\Big\}
\nonumber \\
\geq \;\;
\inf_{\substack{{p\mathstrut}_{Y|X}:\\
{p\mathstrut}_{Y|X}(\, \cdot \, \,|\, x) \, \in \, {\cal F}_{n}, \; \forall x,
\\
\mathbb{E}_{{P\mathstrut}_{\!X}{p\mathstrut}_{Y|X}}[(Y-X)^{2}] \; \leq \; c_{XY} \, + \, \epsilon
}}
&
\Big\{
D\big(\,{p\mathstrut}_{Y|X}\,\|\, \, w \, \,|\, {P\mathstrut}_{\!X}\big)
\, + \,
\big|\,
R \, - \, I\big({P\mathstrut}_{\!X}, \, {p\mathstrut}_{Y|X}\big)
\,\big|^{+}
\Big\}
\; + \; o(1),
\label{eqTypePDF}
\end{align}
{\em where
$o(1)\rightarrow 0$, as $n\rightarrow \infty$,
and depends only on the parameters $\beta$ and $c_{XY}$.}
\end{lemma}

\bigskip

\begin{proof}
Observe first that any collection of conditional PDF's ${p\mathstrut}_{Y|X}$, which satisfies the conditions under the infimum of (\ref{eqTypePDF}),
has finite differential entropies and well-defined quantities $D\big(\,{p\mathstrut}_{Y|X}\,\|\, \, w \, \,|\, {P\mathstrut}_{\!X}\big)$
and $I\big({P\mathstrut}_{\!X}, \, {p\mathstrut}_{Y|X}\big)$.
For any
conditional type ${P\mathstrut}_{\!Y|X}$ from the LHS of (\ref{eqTypePDF}) 
we can define a
set of histogram-like conditional
PDF's:
\begin{align}
{p\mathstrut}_{Y|X}(y\,|\,x) \; & \triangleq \;
\frac{{P\mathstrut}_{\!Y|X}(j\Delta_{\beta,\,n} \, | \, x)}{\Delta_{\beta,\,n}},
\;\;\;\;\;
\forall y \in \big[(j - 1/2)\Delta_{\beta,\,n}, \;  (j + 1/2)\Delta_{\beta,\,n}\big),
\;\forall j \in \mathbb{Z}, \;\forall x \in {\cal S}({P\mathstrut}_{\!X}),
\label{eqStepFunction}
\end{align}
which are step functions of $y \in \mathbb{R}$ for each $x \in {\cal S}({P\mathstrut}_{\!X}) $. 
Then
$I\big({P\mathstrut}_{\!X}, \, {p\mathstrut}_{Y|X}\big) = I\big({P\mathstrut}_{\!X}, {P\mathstrut}_{\!Y|X}\big)$,
and
${p\mathstrut}_{Y|X}(\,\cdot \, \, | \, x) \in {\cal F}_{n}$,
as defined in (\ref{eqBoundedContinuous}).
Analogously to (\ref{eqRealAverage})-(\ref{eqRealBounded2}), it can be obtained that
\begin{displaymath}
\mathbb{E}_{{P\mathstrut}_{\!X}{p\mathstrut}_{Y|X}}\big[(Y-X)^{2}\big] \;\; \leq \;\;
\mathbb{E}_{{P\mathstrut}_{\!XY}}\big[(Y-X)^{2}\big]  \, + \,\Delta_{\beta,\,n}
\sqrt{c_{XY}}
\, + \, \Delta_{\beta,\,n}^{2}/4.
\end{displaymath}
Then also
$D\big(\,{p\mathstrut}_{Y|X}\,\|\, \, w \, \,|\, {P\mathstrut}_{\!X}\big)  \leq
D\big({P\mathstrut}_{\!Y|X}\,\|\, {W\mathstrut}_{\!n} \,|\,  {P\mathstrut}_{\!X}\big)  +  o(1)$.
Then the lemma follows.
\end{proof}


We use Lemmas~\ref{LemPDFtoT} and~\ref{LemTtoPDF} in Section~\ref{Main}
in the derivation of Theorems~\ref{thmErrorExp} and~\ref{thmCorDecExp},
respectively.



\section*{Appendix A}\label{AppendixA}


\subsection*{Proof of Lemma~\ref{LemSupportXY}:}
Adding the two power constraints together, we successively obtain the following inequalities:
\begin{align}
\sum_{(x, \, y) \, \in \, {\cal S}({P\mathstrut}_{\!XY})}{P\mathstrut}_{\!XY}(x, y) (x^{2} + y^{2})
\;\; & \leq \;\;
c_{X} + c_{Y},
\nonumber \\
\sum_{(x, \, y) \, \in \, {\cal S}({P\mathstrut}_{\!XY})}\frac{1}{n} (x^{2} + y^{2})
\;\; & \leq \;\;
c_{X} + c_{Y},
\nonumber \\
\sum_{(x, \, y) \, \in \, {\cal S}({P\mathstrut}_{\!XY})}\frac{1}{|\, {\cal S}({P\mathstrut}_{\!XY}) \,|} (x^{2} + y^{2})
\;\; & \leq \;\;
\frac{n(c_{X} + c_{Y})}{|\, {\cal S}({P\mathstrut}_{\!XY}) \,|},
\nonumber \\
\mathbb{E} \big[\|{\bf U}\|^{2}\big]
\;\; & \leq \;\;
\frac{n(c_{X} + c_{Y})}{|\, {\cal S}({P\mathstrut}_{\!XY}) \,|},
\nonumber
\end{align}
where ${\bf U} \sim \text{Discrete-Uniform}\big({\cal S}({P\mathstrut}_{\!XY})\big)$. To this let us add a {\em continuously} distributed random vector
\begin{displaymath}
{\bf D} \; \sim \; \text{Continuous-Uniform}\big(\,[-\Delta_{\alpha,\,n}/2, \, \Delta_{\alpha,\,n}/2) \, \times \,
[-\Delta_{\beta,\,n}/2, \, \Delta_{\beta,\,n}/2)
\,\big),
\end{displaymath}
independent of ${\bf U}$.
Then $\widetilde{\bf U} \, = \, {\bf U} + {\bf D} \, \sim \, \text{Continuous-Uniform}(A)$,
where
\begin{displaymath}
A \;\; \triangleq \;\; \bigcup_{(x, \, y) \, \in \, {\cal S}({P\mathstrut}_{\!XY})}
\big\{
x \, + \, [-\Delta_{\alpha,\,n}/2, \, \Delta_{\alpha,\,n}/2)
\big\}
\, \times \,
\big\{
y \, + \, [-\Delta_{\beta,\,n}/2, \, \Delta_{\beta,\,n}/2)
\big\}
.
\end{displaymath}
So we have

\begin{align}
\frac{n(c_{X} + c_{Y})}{|\, {\cal S}({P\mathstrut}_{\!XY}) \,|}
\, + \,
\frac{\Delta_{\alpha,\,n}^{2} + \Delta_{\beta,\,n}^{2}}{12}
\;\; & \geq \;\;
\mathbb{E} \big[\|{\bf U}\|^{2}\big] \, + \, \mathbb{E} \big[\|{\bf D}\|^{2}\big]
\;\; = \;\;
\mathbb{E} \big[\|\widetilde{\bf U}\|^{2}\big]
\;\; = \;\;
\frac{1}{\text{vol}(A)}
\int_{A}\|\widetilde{\bf u}\|^{2}d^{2}\widetilde{\bf u}
\nonumber \\
& = \;\;
\frac{1}{\text{vol}(A)}
\bigg[
\int_{A\,\cap\,B}\|\widetilde{\bf u}\|^{2}d^{2}\widetilde{\bf u}
\, +
\int_{A\,\cap\,B{\mathstrut}^{c}}\|\widetilde{\bf u}\|^{2}d^{2}\widetilde{\bf u}
\bigg]
\nonumber \\
& \overset{a}{\geq} \;\;
\frac{1}{\text{vol}(A)}
\bigg[
\int_{A\,\cap\,B}\|\widetilde{\bf u}\|^{2}d^{2}\widetilde{\bf u}
\, +
\int_{A{\mathstrut}^{c}\,\cap\,B}\|{\bf t}\|^{2}d^{2}{\bf t}
\bigg]
\nonumber \\
& = \;\;
\frac{1}{\text{vol}(A)}
\int_{B}\|{\bf t}\|^{2}d^{2}{\bf t},
\label{eqCenteredDisk}
\end{align}
where for ($a$) we use the disk set $B \, \triangleq \, \big\{{\bf t}: \; \pi\|{\bf t}\|^{2} \, \leq \, \text{vol}(A)\big\}$,
centered around zero and of the same total area as $A$,
and the resulting property
\begin{align}
\text{vol}(A\cap B) \, + \, \text{vol}(A\cap B{\mathstrut}^{c})
\;\; = \;\;
\text{vol}(A)
\;\; & = \;\;
\text{vol}(B)
\;\; = \;\;
\text{vol}(A\cap B) \, + \, \text{vol}(A{\mathstrut}^{c}\cap B),
\nonumber \\
\text{vol}(A\cap B{\mathstrut}^{c})
\;\; & = \;\;
\text{vol}(A{\mathstrut}^{c}\cap B),
\nonumber \\
\int_{A\,\cap\,B{\mathstrut}^{c}}\|\widetilde{\bf u}\|^{2}d^{2}\widetilde{\bf u}
\;\; \geq \;\;
\frac{\text{vol}(A)}{\pi}\int_{A\,\cap\,B{\mathstrut}^{c}}d^{2}\widetilde{\bf u}
\;\; & = \;\;
\frac{\text{vol}(A)}{\pi}\int_{A{\mathstrut}^{c}\,\cap\,B}d^{2}{\bf t}
\;\; \geq \;\;
\int_{A{\mathstrut}^{c}\,\cap\,B}\|{\bf t}\|^{2}d^{2}{\bf t}.
\nonumber
\end{align}
Integrating on the RHS of (\ref{eqCenteredDisk}), we obtain
\begin{align}
\frac{\text{vol}(A)}{2\pi}
\;\; = \;\;
\frac{|\, {\cal S}({P\mathstrut}_{\!XY}) \,| \cdot \Delta_{\alpha,\,n}\Delta_{\beta,\,n}}{2\pi}
\;\; & \leq \;\;
\frac{n(c_{X} + c_{Y})}{|\, {\cal S}({P\mathstrut}_{\!XY}) \,|}
\, + \,
\frac{\Delta_{\alpha,\,n}^{2} + \Delta_{\beta,\,n}^{2}}{12},
\nonumber \\
\frac{|\, {\cal S}({P\mathstrut}_{\!XY}) \,|^{\,2} \cdot \Delta_{\alpha,\,n}\Delta_{\beta,\,n}}{2\pi n}
\;\; & \leq \;\;
c_{X} + c_{Y} \, + \,
\frac{\Delta_{\alpha,\,n}^{2} + \Delta_{\beta,\,n}^{2}}{12}\cdot
\underbrace{\frac{|\, {\cal S}({P\mathstrut}_{\!X}) \,|}{n}}_{\leq \; 1}
\;\; \leq \;\;
c_{X} + c_{Y} + 1/6,
\nonumber \\
|\, {\cal S}({P\mathstrut}_{\!XY}) \,|^{\,2}
\;\; & \leq \;\;
2\pi(c_{X} + c_{Y} + 1/6)\cdot n\Delta_{\alpha,\,n}^{-1}\Delta_{\beta,\,n}^{-1}
\nonumber \\
& = \;\;
2\pi(c_{X} + c_{Y} + 1/6)\cdot n^{1 \, + \, \alpha \, + \, \beta}.
\nonumber
\end{align}
$\square$


\bigskip

\subsection*{Proof of Lemma~\ref{LemSupportX}:}
Similarly as in Lemma~\ref{LemSupportXY}, the power constraint gives the following succession of inequalities:
\begin{align}
\sum_{x \, \in \, {\cal S}({P\mathstrut}_{\!X})}{P\mathstrut}_{\!X}(x) x^{2}
\;\; & \leq \;\;
c_{X},
\nonumber \\
\sum_{x \, \in \, {\cal S}({P\mathstrut}_{\!X})} \frac{1}{n} x^{2}
\;\; & \leq \;\;
c_{X},
\nonumber \\
\sum_{x \, \in \, {\cal S}({P\mathstrut}_{\!X})} \frac{1}{|\, {\cal S}({P\mathstrut}_{\!X}) \,|} x^{2}
\;\; & \leq \;\;
\frac{n}{|\, {\cal S}({P\mathstrut}_{\!X}) \,|}c_{X},
\nonumber \\
\mathbb{E} \big[U^{2}\big]
\;\; & \leq \;\;
\frac{n}{|\, {\cal S}({P\mathstrut}_{\!X}) \,|}c_{X},
\nonumber
\end{align}
where $U \sim \text{Discrete-Uniform}\big({\cal S}({P\mathstrut}_{\!X})\big)$. To this let us add a {\em continuously} distributed random variable
\begin{displaymath}
D \; \sim \; \text{Continuous-Uniform}\big(\,[-\Delta_{\alpha,\,n}/2, \, \Delta_{\alpha,\,n}/2)\,\big),
\end{displaymath}
independent of $U$.
Then $\widetilde{U} \, = \, U + D \, \sim \, \text{Continuous-Uniform}(A)$,
where
\begin{displaymath} 
A \;\; \triangleq \;\; \bigcup_{x \, \in \, {\cal S}({P\mathstrut}_{\!X})}
\big\{
x + [-\Delta_{\alpha,\,n}/2, \, \Delta_{\alpha,\,n}/2)
\big\}.
\end{displaymath}
So we have
\begin{align}
\frac{n}{|\, {\cal S}({P\mathstrut}_{\!X}) \,|}c_{X}
\, + \,
\frac{\Delta_{\alpha,\,n}^{2}}{12}
\;\; 
\geq \;\;
\mathbb{E} \big[U^{2}\big] \, + \, \mathbb{E} \big[D^{2}\big]
\;\; = \;\;
\mathbb{E} \big[\widetilde{U}^{2}\big]
\;\; & = \;\;
\frac{1}{\text{vol}(A)}
\int_{A}\widetilde{u}^{\,2}d\widetilde{u}
\nonumber \\
& = \;\;
\frac{1}{\text{vol}(A)}
\bigg[
\int_{A\,\cap\,B}\widetilde{u}^{\,2}d\widetilde{u}
\, +
\int_{A\,\cap\,B{\mathstrut}^{c}}\widetilde{u}^{\,2}d\widetilde{u}
\bigg]
\nonumber \\
& \overset{a}{\geq} \;\;
\frac{1}{\text{vol}(A)}
\bigg[
\int_{A\,\cap\,B}\widetilde{u}^{\,2}d\widetilde{u}
\, +
\int_{A{\mathstrut}^{c}\,\cap\,B}t^{\,2}dt
\bigg]
\nonumber \\
& = \;\;
\frac{1}{\text{vol}(A)}
\int_{B}t^{\,2}d t,
\label{eqCentered}
\end{align}
where for ($a$) we use the interval set $B \, \triangleq \, \big[-\!\text{vol}(A)/2, \; \text{vol}(A)/2 \,\big]$,
centered around zero and of the same total length as $A$
with the resulting property that $\text{vol}(A\cap B{\mathstrut}^{c}) = \text{vol}(A{\mathstrut}^{c}\cap B)$,
so that
\begin{displaymath}
\int_{A\,\cap\,B{\mathstrut}^{c}}\widetilde{u}^{\,2}d\widetilde{u}
\;\; \geq \;\;
\frac{(\text{vol}(A))^{2}}{4}\int_{A\,\cap\,B{\mathstrut}^{c}}d\widetilde{u}
\;\; = \;\;
\frac{(\text{vol}(A))^{2}}{4}\int_{A{\mathstrut}^{c}\,\cap\,B}dt
\;\; \geq \;\;
\int_{A{\mathstrut}^{c}\,\cap\,B}t^{\,2}dt.
\end{displaymath}
Integrating on the RHS of (\ref{eqCentered}), we obtain
\begin{align}
\frac{(\text{vol}(A))^{2}}{12}
\;\; = \;\;
\frac{(|\, {\cal S}({P\mathstrut}_{\!X}) \,|\cdot \Delta_{\alpha,\,n})^{2}}{12}
\;\; & \leq \;\;
\frac{n}{|\, {\cal S}({P\mathstrut}_{\!X}) \,|}c_{X}
\, + \,
\frac{\Delta_{\alpha,\,n}^{2}}{12},
\nonumber \\
|\, {\cal S}({P\mathstrut}_{\!X}) \,|^{\,3}
\;\; & \leq \;\;
12c_{X}n\Delta_{\alpha,\,n}^{-2} \, + \, |\, {\cal S}({P\mathstrut}_{\!X}) \,|
\nonumber \\
& = \;\;
12c_{X}n^{1\,+\,2\alpha} \, + \, \underbrace{|\, {\cal S}({P\mathstrut}_{\!X}) \,|}_{\leq \; n}
\;\; \leq \;\;
(12c_{X} + 1)n^{1\,+\,2\alpha}.
\nonumber
\end{align}
$\square$



\section*{Appendix B}\label{AppendixB}


\subsection*{Proof of Lemma~\ref{LemQuant}:}


Using ${P\mathstrut}_{\!X}$ and ${p\mathstrut}_{Y|X}$, it is convenient to define a joint probability density function over $\mathbb{R}^{2}$ as
\begin{align}
{p\mathstrut}_{XY}(x, y) \;\; & \triangleq \;\;
\frac{{P\mathstrut}_{\!X}(i\Delta_{\alpha,\,n})}{\Delta_{\alpha,\,n}}
{p\mathstrut}_{Y|X}(y\,|\,i\Delta_{\alpha,\,n}),
\nonumber \\
& \;\;\;\;\;\;\;\;\;\;\;\;\;
\forall x \in \big[(i - 1/2)\Delta_{\alpha,\,n}, \;  (i + 1/2)\Delta_{\alpha,\,n}\big),
\;\forall i \in \mathbb{Z}, \;\forall y \in \mathbb{R},
\label{eqJointPDF}
\end{align}
which is changing only stepwise in $x$-direction.
Note that ${p\mathstrut}_{Y}$ 
of (\ref{eqMargEntropy}) is the $y$-marginal of ${p\mathstrut}_{XY}$.
This gives

\begin{equation} \label{eqJPDFVariance}
\mathbb{E}_{{p\mathstrut}_{XY}}\!\big[X^{2}\big]  \; = \; \mathbb{E}_{{P\mathstrut}_{\!X}}\!\big[X^{2}\big]
\, + \,
\frac{\Delta_{\alpha,\,n}^{2}}{12},
\;\;\;\;\;\;
\mathbb{E}_{{p\mathstrut}_{XY}}\!\big[Y^{2}\big] \; = \;
\mathbb{E}_{{P\mathstrut}_{\!X}{p\mathstrut}_{Y|X}}\!\big[Y^{2}\big].
\end{equation}
We proceed in two stages. First we quantize ${p\mathstrut}_{XY}(x, y)$ by rounding it {\em down} and check
the effect of this on the LHS of (\ref{eqCondEntropy})-(\ref{eqMargEntropy}).
Then we complement the total probability back to $1$, so that the type ${P\mathstrut}_{\!X}$ is conserved,
and check the effect of this on the RHS of (\ref{eqCondEntropy})-(\ref{eqMargEntropy}).

The quantization of ${p\mathstrut}_{XY}(x, y)$ is done by first replacing it with its infimum in each
rectangle $$\big[(i-1/2)\Delta_{\alpha,\,n}, \; (i+1/2)\Delta_{\alpha,\,n}\big) \; \times \;
\big[(j-1/2)\Delta_{\beta,\,n}, \; (j+1/2)\Delta_{\beta,\,n}\big):$$
\begin{align}
{p\mathstrut}_{XY}^{\inf}(x, y)
\;\; & \triangleq \;\;
\inf_{(j \, - \, 1/2)\Delta_{\beta,\,n} \; \leq \; \tilde{y} \; < \; (j \, + \, 1/2)\Delta_{\beta,\,n}}
{p\mathstrut}_{XY}(x, \tilde{y})
,
\nonumber \\
&
\;\;\;\;\;\;\;\;\;\;\;\;\;\;\;\;\;\;\;\;\;\;\;\;
\forall y \in \big[(j - 1/2)\Delta_{\beta,\,n}, \;  (j + 1/2)\Delta_{\beta,\,n}\big),
\;\forall j \in \mathbb{Z}, \;\forall x \in \mathbb{R},
\label{eqInf}
\end{align}
and then quantizing this infimum
down
to the nearest value $k\Delta_{\gamma,\,n}$, $k = 0, 1, 2, .\,.\,.\,$:  
\begin{align}
{p\mathstrut}_{XY}^{q}(x, y)
\;\; & \triangleq \;\;
\left\lfloor
{p\mathstrut}_{XY}^{\inf}(x, y)/\Delta_{\gamma,\,n}\right\rfloor\cdot\Delta_{\gamma,\,n},
\;\;\;\;\;\;
\forall (x, y) \in \mathbb{R}^{2}.
\label{eqQuantized}
\end{align}
Due to (\ref{eqCube}), the integral of ${p\mathstrut}_{XY}^{q}$ over $\mathbb{R}^{2}$ can be smaller than $1$ only by an integer 
multiple
of $1/n$.
The resulting difference from ${p\mathstrut}_{XY}(x, y)$ at each point $(x, y) \in \mathbb{R}^{2}$ can be bounded 
by a sum of two terms as
\begin{align}
0 \;\; \leq \;\;
{p\mathstrut}_{XY}(x, y) \, - \, {p\mathstrut}_{XY}^{q}(x, y)
\;\; & \leq \;\;
\underbrace{(K/\Delta_{\alpha,\,n})\cdot\Delta_{\beta,\,n}}_{\text{minimization}} \; + \!\!
\underbrace{\Delta_{\gamma,\,n}}_{\text{quantization}}
\nonumber \\
& = \;\;
K n^{\alpha \, - \, \beta}
\, + \, n^{\alpha \, + \, \beta \, - \, 1}
\;\; \leq \;\;
(K + 1)n^{-\delta}
\;\; \triangleq \;\; h,
\label{eqDifference}
\end{align}
where $\delta \,\triangleq\, \min\,\{\beta, \, 1-\beta\} \, - \, \alpha$ and $K$ is the parameter from (\ref{eqLipschitz}).

For (\ref{eqMargEntropy}) we will require the $y$-marginal of ${p\mathstrut}_{XY}^{q}$ from (\ref{eqQuantized}), defined in the usual manner:
\begin{align}
{p\mathstrut}_{Y}^{q}(y) \;\; \triangleq \;\;
\int_{\mathbb{R}}{p\mathstrut}_{XY}^{q}(x, y)dx
\;\; & \overset{a}{=} \;\;
\sum_{x\,\in\, {\cal S}({P\mathstrut}_{\!X})}{p\mathstrut}_{XY}^{q}(x, y)\Delta_{\alpha,\,n}
\;\; \overset{b}{\geq} \;\;
\sum_{x\,\in\, {\cal S}({P\mathstrut}_{\!X})}({p\mathstrut}_{XY}(x, y) - h)\Delta_{\alpha,\,n}
\label{eqMargUsual} \\
& = \;\;
\sum_{x\,\in\, {\cal S}({P\mathstrut}_{\!X})}{p\mathstrut}_{XY}(x, y)\Delta_{\alpha,\,n} \; - \; h\sum_{x\,\in\, {\cal S}({P\mathstrut}_{\!X})}\Delta_{\alpha,\,n}
\nonumber \\
& = \;\; {p\mathstrut}_{Y}(y) \; - \; h \,\,|\, {\cal S}({P\mathstrut}_{\!X}) \,|\,\,\Delta_{\alpha,\,n},
\;\;\;\;\;\; \forall y \in \mathbb{R},
\nonumber
\end{align}
where
the equality
($a$) follows from (\ref{eqJointPDF}), (\ref{eqInf}), (\ref{eqQuantized}),
and the inequality ($b$) follows by (\ref{eqDifference}).
Then
\begin{align}
0 \;\; \leq \;\;
{p\mathstrut}_{Y}(y) \, - \, {p\mathstrut}_{Y}^{q}(y)
\;\; & \leq \;\;
h \,\,|\, {\cal S}({P\mathstrut}_{\!X}) \,|\,\,\Delta_{\alpha,\,n}
\;\; \leq \;\;
(K + 1)(12c_{X} + 1)^{1/3}n^{-\delta_{1}} \;\; \triangleq \;\; \widetilde{h},
\label{eqYDifference}
\end{align}
where the last inequality follows by Lemma~\ref{LemSupportX}, (\ref{eqDifference}), (\ref{eqDelta}),
with $\delta_{1} \,\triangleq\, \min\,\{\beta, \, 1-\beta\} \, - \, (1 + 2\alpha)/3$.
Note that the previously defined $\delta > \delta_{1}$, 
while $\delta_{1} > 0$ if and only if $\alpha \in \big(0, \tfrac{1}{4}\big)$ and
$\tfrac{1}{3} + \tfrac{2}{3}\alpha < \beta < \tfrac{2}{3} - \tfrac{2}{3}\alpha$.

\bigskip

{\em The LHS of (\ref{eqCondEntropy})}


Now 
consider
the LHS of (\ref{eqCondEntropy}).
Note that each function ${p\mathstrut}_{Y|X}(\,\cdot\,|\,x)$ in (\ref{eqCondEntropy})
is bounded and has
a finite variance. It follows that it has 
a finite
differential entropy.
With (\ref{eqJointPDF}) we can rewrite the LHS of (\ref{eqCondEntropy}) as
\begin{align}
\sum_{x\,\in\,{\cal X}_{n}}{P\mathstrut}_{\!X}(x)\int_{\mathbb{R}}{p\mathstrut}_{Y|X}(y\,|\,x)
{\log\mathstrut}_{\!b}\,{p\mathstrut}_{Y|X}(y\,|\,x)dy
\;\; = \;\; &
\int\!\!\!\!\int_{\mathbb{R}^{2}}{p\mathstrut}_{XY}(x, y)
{\log\mathstrut}_{\!b}\,{p\mathstrut}_{XY}(x, y)dxdy
\nonumber \\
&
- \sum_{x\,\in\,{\cal X}_{n}}
{P\mathstrut}_{\!X}(x){\log\mathstrut}_{\!b}\,\frac{{P\mathstrut}_{\!X}(x)}{\Delta_{\alpha,\,n}}.
\label{eqCondEntIdeal}
\end{align}
Let us examine the possible increase in (\ref{eqCondEntIdeal})
when ${p\mathstrut}_{XY}$ is replaced with ${p\mathstrut}_{XY}^{q}$ defined by (\ref{eqInf})-(\ref{eqQuantized}).
For this, let us define a set in $\mathbb{R}^{2}$ with respect to the parameter $h$ of (\ref{eqDifference}):
\begin{align}
A \;\; & \triangleq \;\;
\big\{
(x, y): \;\; {p\mathstrut}_{XY}(x, y) > h
\big\},
\label{eqCountableUnion}
\end{align}
which is a countable union of disjoint rectangles
by the definition of ${p\mathstrut}_{XY}$ in (\ref{eqJointPDF}). Then
\begin{align}
&
\;\;\;\;
\int\!\!\!\!\int_{\mathbb{R}^{2}}{p\mathstrut}_{XY}^{q}(x, y)
{\log\mathstrut}_{\!b}\,{p\mathstrut}_{XY}^{q}(x, y)dxdy
\, - \,
\int\!\!\!\!\int_{\mathbb{R}^{2}}{p\mathstrut}_{XY}(x, y)
{\log\mathstrut}_{\!b}\,{p\mathstrut}_{XY}(x, y)dxdy
\nonumber \\
=  &
\;\;
\bigg[\int\!\!\!\!\int_{A{\mathstrut}^{c}}{p\mathstrut}_{XY}^{q}(x, y)
{\log\mathstrut}_{\!b}\,{p\mathstrut}_{XY}^{q}(x, y)dxdy
\, - \,
\int\!\!\!\!\int_{A{\mathstrut}^{c}}{p\mathstrut}_{XY}(x, y)
{\log\mathstrut}_{\!b}\,{p\mathstrut}_{XY}(x, y)dxdy
\bigg]
\nonumber \\
&
+
\int\!\!\!\!\int_{A}
\big[\,{p\mathstrut}_{XY}^{q}(x, y)
{\log\mathstrut}_{\!b}\,{p\mathstrut}_{XY}^{q}(x, y)
\, - \,
{p\mathstrut}_{XY}(x, y)
{\log\mathstrut}_{\!b}\,{p\mathstrut}_{XY}(x, y)\big]
\, dxdy.
\label{eqByParts}
\end{align}
Note that the minimum of the function $f(t) = t\,{\log\mathstrut}_{\!b} \,t$ occurs at $t = 1/e$.
Then for $h \leq 1/e$
we have ${p\mathstrut}_{XY}(x, y)
{\log\mathstrut}_{\!b}\,{p\mathstrut}_{XY}(x, y) \, \leq \, {p\mathstrut}_{XY}^{q}(x, y)
{\log\mathstrut}_{\!b}\,{p\mathstrut}_{XY}^{q}(x, y) \, \leq \, 0$
for all $(x, y) \in A{\mathstrut}^{c}$
and the first of the two terms in (\ref{eqByParts}) is upper-bounded as
\begin{align}
&
\;\;\,\,\,
\int\!\!\!\!\int_{A{\mathstrut}^{c}}{p\mathstrut}_{XY}^{q}(x, y)
{\log\mathstrut}_{\!b}\,{p\mathstrut}_{XY}^{q}(x, y)dxdy
\, - \,
\int\!\!\!\!\int_{A{\mathstrut}^{c}}{p\mathstrut}_{XY}(x, y)
{\log\mathstrut}_{\!b}\,{p\mathstrut}_{XY}(x, y)dxdy
\label{eqLocalEntropy} \\
\overset{h\, \leq \, 1/e}{\leq} \;\; &
-\int\!\!\!\!\int_{A{\mathstrut}^{c}}{p\mathstrut}_{XY}(x, y)
{\log\mathstrut}_{\!b}\,{p\mathstrut}_{XY}(x, y)dxdy
\;\; \overset{(*)}{=} \;\;
-p
\int\!\!\!\!\int_{\mathbb{R}^{2}}{p\mathstrut}_{XY}^{c}(x, y)
{\log\mathstrut}_{\!b}\,{p\mathstrut}_{XY}^{c}(x, y)dxdy
\, - \,
p\,{\log\mathstrut}_{\!b}\,p,
\nonumber
\end{align}
where the equality ($*$) is appropriate
for the case when the upper bound is positive, 
with the definitions:
\begin{align}
p \;\; & \triangleq \;\;
\int\!\!\!\!\int_{A{\mathstrut}^{c}}{p\mathstrut}_{XY}(x, y)dxdy,
\label{eqP} \\
{p\mathstrut}_{XY}^{c}(x, y)
\;\; & \triangleq \;\;
\Bigg\{
\begin{array}{r r}
{p\mathstrut}_{XY}(x, y)/p, & (x, y) \in  A{\mathstrut}^{c}, \\
0, & \text{o.w.}
\end{array}
\nonumber
\end{align}
Next we upper-bound the entropy of the probability density function ${p\mathstrut}_{XY}^{c}$ on the RHS of (\ref{eqLocalEntropy})
by that of 
a Gaussian PDF.
By (\ref{eqJPDFVariance}) we have

\begin{align}
&
\;\;\;\;\;\;\;\;\;\;\;\;\;\;\,\,
c_{X} \, + \, 1/12 \, + \, c_{Y}
\, \geq \,
\int\!\!\!\!\int_{A{\mathstrut}^{c}}{p\mathstrut}_{XY}(x, y)(x^{2} + y^{2})dxdy
\, = \,
p\int\!\!\!\!\int_{\mathbb{R}^{2}}{p\mathstrut}_{XY}^{c}(x, y)(x^{2} + y^{2})dxdy,
\label{eqPartialMoment} \\
&
\;\;\;\;\;\;\;\;\;\;\;\;\,\,\,
(c_{X} \, + \, 1/12 \, + \, c_{Y})/p
\, \geq \,
\int_{\mathbb{R}}{p\mathstrut}_{X}^{c}(x)x^{2}dx
\, + \,
\int_{\mathbb{R}}{p\mathstrut}_{Y}^{c}(y)y^{2}dy,
\nonumber \\
&
{\log\mathstrut}_{\!b}\,\big(2\pi e(c_{X} \, + \, 1/12 \, + \, c_{Y})/p\big)
\, \geq \,
-\int\!\!\!\!\int_{\mathbb{R}^{2}}{p\mathstrut}_{XY}^{c}(x, y)
{\log\mathstrut}_{\!b}\,{p\mathstrut}_{XY}^{c}(x, y)dxdy.
\nonumber
\end{align}
So we can rewrite the bound of (\ref{eqLocalEntropy}) in terms of $p$ defined by (\ref{eqP}):
\begin{align}
&
\int\!\!\!\!\int_{A{\mathstrut}^{c}}{p\mathstrut}_{XY}^{q}(x, y)
{\log\mathstrut}_{\!b}\,{p\mathstrut}_{XY}^{q}(x, y)dxdy
\, - \,
\int\!\!\!\!\int_{A{\mathstrut}^{c}}{p\mathstrut}_{XY}(x, y)
{\log\mathstrut}_{\!b}\,{p\mathstrut}_{XY}(x, y)dxdy
\nonumber \\
\overset{h\, \leq \, 1/e}{\leq} \;\; &
-2p\,{\log\mathstrut}_{\!b}\,p
\, + \, p\,{\log\mathstrut}_{\!b}\,\big(2\pi e(c_{X} + c_{Y} + 1/12)\big).
\label{eqLocalp}
\end{align}
From (\ref{eqP}) and (\ref{eqCountableUnion}) it is clear that $p \rightarrow 0$ as $h \rightarrow 0$.
In order to relate between them, let us rewrite the inequality in (\ref{eqPartialMoment})
again as
\begin{align}
& c_{X} \, + \, 1/12 \, + \, c_{Y}
\;\; \geq \;\;
\int\!\!\!\!\int_{A{\mathstrut}^{c}}{p\mathstrut}_{XY}(x, y)(x^{2} + y^{2})dxdy
\nonumber \\
= \;\; &
\int\!\!\!\!\int_{B{\mathstrut}_{1}}h\cdot(x^{2} + y^{2})dxdy
\; + \;
\int\!\!\!\!\int_{A{\mathstrut}^{c}\,\cap\,B{\mathstrut}^{c}_{1}}{p\mathstrut}_{XY}(x, y)(x^{2} + y^{2})dxdy
\; - \;
\int\!\!\!\!\int_{A\,\cap\,B{\mathstrut}_{1}}h\cdot(x^{2} + y^{2})dxdy
\nonumber \\
&
\;\;\;\;\;\;\;\;\;\;\;\;\;\;\;\;\;\;\;\;\;\;\;\;\;\;\;\;\;\;\;\;\;\;\;\;\;\;\;\;
\;\;\;\;\;\;\;\;\;\;\;\;\;\;\;\;\;\;\;\;\;\;\;\;\;\;\;\;\;\;\,\,
-\int\!\!\!\!\int_{A{\mathstrut}^{c}\,\cap\,B{\mathstrut}_{1}}\big(h - {p\mathstrut}_{XY}(x, y)\big)(x^{2} + y^{2})dxdy
\nonumber \\
\geq \;\; &
\int\!\!\!\!\int_{B{\mathstrut}_{1}}h\cdot(x^{2} + y^{2})dxdy
\; + \;
\frac{p}{h\pi}\int\!\!\!\!\int_{A{\mathstrut}^{c}\,\cap\,B{\mathstrut}^{c}_{1}}{p\mathstrut}_{XY}(x, y)dxdy
\; - \;
\frac{p}{h\pi}\int\!\!\!\!\int_{A\,\cap\,B{\mathstrut}_{1}}hdxdy
\nonumber \\
&
\;\;\;\;\;\;\;\;\;\;\;\;\;\;\;\;\;\;\;\;\;\;\;\;\;\;\;\;\;\;\;\;\;\;\;\;\;\;\;\;
\;\;\;\;\;\;\;\;\;\;\;\;\;\;\;\;\;\;\;\,\,
-\frac{p}{h\pi}\int\!\!\!\!\int_{A{\mathstrut}^{c}\,\cap\,B{\mathstrut}_{1}}\big(h - {p\mathstrut}_{XY}(x, y)\big)dxdy
\nonumber \\
= \;\; &
\int\!\!\!\!\int_{B{\mathstrut}_{1}}h\cdot(x^{2} + y^{2})dxdy \;\; = \;\; \frac{p^{2}}{2\pi h},
\nonumber
\end{align}
where 
we use the disk set $B{\mathstrut}_{1} \triangleq \big\{ (x, y): \; h\pi(x^{2} + y^{2}) \, \leq \, p\big\}$,
centered around zero.
This results in the following upper bound on $p$ in terms of $h$:
\begin{equation} \label{eqHP}
c_{1} h^{1/2}
\;\; \geq \;\; p,
\end{equation}
where $c_{1} \triangleq \sqrt{2\pi(c_{X} + c_{Y} + 1/12)}$.
Substituting the LHS of (\ref{eqHP}) in (\ref{eqLocalp})
in place of $p$, we obtain the following upper bound on
the first half 
of (\ref{eqByParts}) in terms of $h$
of (\ref{eqDifference}), (\ref{eqCountableUnion}):
\begin{align}
&
\int\!\!\!\!\int_{A{\mathstrut}^{c}}{p\mathstrut}_{XY}^{q}(x, y)
{\log\mathstrut}_{\!b}\,{p\mathstrut}_{XY}^{q}(x, y)dxdy
\, - \,
\int\!\!\!\!\int_{A{\mathstrut}^{c}}{p\mathstrut}_{XY}(x, y)
{\log\mathstrut}_{\!b}\,{p\mathstrut}_{XY}(x, y)dxdy
\nonumber \\
\leq \;\; &
c_{1}
h^{1/2}
\,{\log\mathstrut}_{\!b}\,(e/h),
\;\;\;\;\;\;
h \,\leq \,
1/(c_{1}e)^{2}.
\label{eqFirstTerm}
\end{align}

In the second term of (\ref{eqByParts})
for $(x, y) \in A$ the integrand can be upper-bounded by Lemma~\ref{Lemxlogx} with its parameters $t$ and $t_{1}$ such that
\begin{align}
t_{1} \; = \; {p\mathstrut}_{XY}^{q}(x, y)
\; & \leq \; {p\mathstrut}_{XY}(x, y) \; = \; t_{1} + t,
\;\;\;\;\;\;\;\;\;
t \, \leq \, h \, \leq \, 1/e.
\nonumber
\end{align}
This gives
\begin{align}
&
\int\!\!\!\!\int_{A}
\big[\,{p\mathstrut}_{XY}^{q}(x, y)
{\log\mathstrut}_{\!b}\,{p\mathstrut}_{XY}^{q}(x, y)
\, - \,
{p\mathstrut}_{XY}(x, y)
{\log\mathstrut}_{\!b}\,{p\mathstrut}_{XY}(x, y)\big]
\,
dxdy
\nonumber \\
\;\; 
\leq
\;\; &
\text{vol}(A)\,
h\,
{\log\mathstrut}_{\!b}\,(1/h),
\;\;\;\;\;\; h \, \leq \, 1/e,
\label{eqVolh}
\end{align}
where
$\text{vol}(A)$ is the total area of $A$.
To find an upper bound on $\text{vol}(A)$, we use (\ref{eqJPDFVariance}):
\begin{align}
c_{X} \, + \, 1/12 \, + \, c_{Y}
\;\; & \geq \;\;
\int\!\!\!\!\int_{A}{p\mathstrut}_{XY}(x, y)(x^{2} + y^{2})dxdy
\;\; \geq \;\;
\int\!\!\!\!\int_{A}h \cdot(x^{2} + y^{2})dxdy
\nonumber \\
& = \;\;
h\left[
\int\!\!\!\!\int_{A\,\cap\,B{\mathstrut}_{2}}(x^{2} + y^{2})dxdy
\, +
\int\!\!\!\!\int_{A\,\cap\,B{\mathstrut}^{c}_{2}}(x^{2} + y^{2})dxdy
\right]
\nonumber \\
& \overset{a}{\geq} \;\;
h\left[
\int\!\!\!\!\int_{A\,\cap\,B{\mathstrut}_{2}}(x^{2} + y^{2})dxdy
\, +
\int\!\!\!\!\int_{A{\mathstrut}^{c}\,\cap\,B{\mathstrut}_{2}}(x^{2} + y^{2})dxdy
\right]
\nonumber \\
& = \;\;
\int\!\!\!\!\int_{B{\mathstrut}_{2}}h \cdot(x^{2} + y^{2})dxdy
\;\; = \;\;
\frac{h}{2\pi}(\text{vol}(A))^{2},
\nonumber
\end{align}
where in ($a$) we use the disk set
$B{\mathstrut}_{2} \, \triangleq \, \big\{(x, y): \; \pi(x^{2} + y^{2}) \, \leq \, \text{vol}(A)\big\}$,
centered around zero,
of the same total area as $A$, and the resulting property that
$\text{vol}(A{\mathstrut}^{c}\cap B{\mathstrut}_{2}) = \text{vol}(A\cap B{\mathstrut}^{c}_{2})$.
So that
\begin{equation} \label{eqVol}
c_{1}h^{-1/2}
\;\; \geq \;\;
\text{vol}(A).
\end{equation}
Continuing (\ref{eqVolh}), therefore we obtain the following
upper bound on the second term in (\ref{eqByParts}):
\begin{align}
\int\!\!\!\!\int_{A}
\big[\,{p\mathstrut}_{XY}^{q}(x, y)
{\log\mathstrut}_{\!b}\,{p\mathstrut}_{XY}^{q}(x, y)
\, - \,
{p\mathstrut}_{XY}(x, y)
{\log\mathstrut}_{\!b}\,{p\mathstrut}_{XY}(x, y)
\big]
dxdy
\;
& \leq
\; 
c_{1}
h^{1/2}\,
{\log\mathstrut}_{\!b}\,(1/h),
\;\;\; h \leq 1/e.
\label{eqSecondTerm}
\end{align}
Putting (\ref{eqByParts}), (\ref{eqFirstTerm}) and (\ref{eqSecondTerm}) together:
\begin{align}
&
\int\!\!\!\!\int_{\mathbb{R}^{2}}{p\mathstrut}_{XY}^{q}(x, y)
{\log\mathstrut}_{\!b}\,{p\mathstrut}_{XY}^{q}(x, y)dxdy
\, - \,
\int\!\!\!\!\int_{\mathbb{R}^{2}}{p\mathstrut}_{XY}(x, y)
{\log\mathstrut}_{\!b}\,{p\mathstrut}_{XY}(x, y)dxdy
\nonumber \\
\leq \;\; &
2c_{1}
h^{1/2}
\,{\log\mathstrut}_{\!b} (\sqrt{e}/h),
\;\;\;\;\;\; h \, \leq \, 1/(c_{1}e)^{2},
\label{eqPutting}
\end{align}
where $c_{1}$ 
and $h$ are such as in (\ref{eqHP}) 
and (\ref{eqDifference}), respectively.
So that if $\delta > 0$ in (\ref{eqDifference}), then the 
possible increase in (\ref{eqCondEntIdeal})
caused by substitution of ${p\mathstrut}_{XY}^{q}$ in place of ${p\mathstrut}_{XY}$ is at most $o(1)$.

Later on, for the RHS of (\ref{eqCondEntropy})-(\ref{eqMargEntropy}) we will require also
the loss in the total probability incurred in the replacement of ${p\mathstrut}_{XY}$ by ${p\mathstrut}_{XY}^{q}$.
This loss is strictly positive and tends to zero with $h$ of (\ref{eqDifference}): 

\begin{align}
0 \;\; < \;\; p_{1} \;\; & \triangleq \;\;
\underbrace{\int\!\!\!\!\int_{\mathbb{R}^{2}}{p\mathstrut}_{XY}(x, y)dxdy}_{= \; 1} \; - \;
\int\!\!\!\!\int_{\mathbb{R}^{2}}{p\mathstrut}_{XY}^{q}(x, y)dxdy
\nonumber \\
& \overset{a}{\leq} \;\;
\underbrace{\int\!\!\!\!\int_{A{\mathstrut}^{c}}{p\mathstrut}_{XY}(x, y)dxdy}_{= \; p} \; + \;
\int\!\!\!\!\int_{A}\underbrace{\big({p\mathstrut}_{XY}(x, y) - {p\mathstrut}_{XY}^{q}(x, y)\big)}_{\leq \; h}dxdy
\nonumber \\
& \overset{b}{\leq} \;\;
p \, + \, h\cdot\text{vol}(A) \;\; \overset{c}{\leq} \;\;
c_{1}h^{1/2} \, + \, c_{1}h^{1/2} \;\; = \;\; 2c_{1}h^{1/2},
\label{eqProbLoss}
\end{align}
where the set $A$ in ($a$) is defined in (\ref{eqCountableUnion}),
($b$) follows by (\ref{eqP}) and (\ref{eqDifference}),
and ($c$) follows by (\ref{eqHP}), 
(\ref{eqVol}).

\bigskip


{\em The LHS of (\ref{eqMargEntropy})}


Consider next the LHS of (\ref{eqMargEntropy}). Since ${p\mathstrut}_{Y} \in {\cal L}$
and has a finite variance, its differential entropy is finite.
Let us examine the possible decrease in the LHS of (\ref{eqMargEntropy})
when ${p\mathstrut}_{Y}$ is replaced with ${p\mathstrut}_{Y}^{q}$ defined in (\ref{eqMargUsual}).
For this, let us define a set in $\mathbb{R}$ with respect to the parameter $\widetilde{h}$ of (\ref{eqYDifference}):
\begin{align}
\widetilde{A} \;\; & \triangleq \;\;
\big\{
y: \;\; {p\mathstrut}_{Y}(y) > \widetilde{h}
\big\},
\label{eqCountableUnion2}
\end{align}
which is a countable union of disjoint open intervals. Then
\begin{align}
&
\int_{\mathbb{R}}{p\mathstrut}_{Y}(y)
{\log\mathstrut}_{\!b}\,{p\mathstrut}_{Y}(y)dy
\, - \,
\int_{\mathbb{R}}{p\mathstrut}_{Y}^{q}(y)
{\log\mathstrut}_{\!b}\,{p\mathstrut}_{Y}^{q}(y)dy
\nonumber \\
= \; &
\int_{\widetilde{A}{\mathstrut}^{c}}
\big[\,{p\mathstrut}_{Y}(y)
{\log\mathstrut}_{\!b}\,{p\mathstrut}_{Y}(y)
\, - \,
{p\mathstrut}_{Y}^{q}(y)
{\log\mathstrut}_{\!b}\,{p\mathstrut}_{Y}^{q}(y)
\big]
\,dy
\; + \, 
\int_{\widetilde{A}}
\big[\,
{p\mathstrut}_{Y}(y)
{\log\mathstrut}_{\!b}\,{p\mathstrut}_{Y}(y)
\, - \,
{p\mathstrut}_{Y}^{q}(y)
{\log\mathstrut}_{\!b}\,{p\mathstrut}_{Y}^{q}(y)
\big]
\,dy.
\label{eqByParts2}
\end{align}
For $\widetilde{h} \leq 1/e$
we have ${p\mathstrut}_{Y}(y)
{\log\mathstrut}_{\!b}\,{p\mathstrut}_{Y}(y) \, \leq \, {p\mathstrut}_{Y}^{q}(y)
{\log\mathstrut}_{\!b}\,{p\mathstrut}_{Y}^{q}(y) \, \leq \, 0$
for all $y \in \widetilde{A}{\mathstrut}^{c}$
and the first of the two terms in (\ref{eqByParts2}) is non-positive:
\begin{align}
&
\int_{\widetilde{A}{\mathstrut}^{c}}
\big[\,
{p\mathstrut}_{Y}(y)
{\log\mathstrut}_{\!b}\,{p\mathstrut}_{Y}(y)
\, - \,
{p\mathstrut}_{Y}^{q}(y)
{\log\mathstrut}_{\!b}\,{p\mathstrut}_{Y}^{q}(y)
\big]
\,dy
\;\; \leq \;\; 0.
\label{eqLocalEntropy2} 
\end{align}

In the second term of (\ref{eqByParts2})
for $y \in \widetilde{A}$ the integrand can be upper-bounded by Lemma~\ref{Lemxlogx} with its parameters $t$ and $t_{1}$ such that
\begin{align}
t_{1} \; & = \; {p\mathstrut}_{Y}^{q}(y)
\; \leq \; t_{1} + t \; 
= \; {p\mathstrut}_{Y}(y)
\; \leq \; \sup_{y\, \in \, \mathbb{R}}{p\mathstrut}_{Y}(y)
\; \leq \; \sqrt{K},
\;\;\;\;\;\;\;\;\;
t \, 
\leq \, \widetilde{h} \, \leq \, 1/e,
\nonumber
\end{align}
where $K$
is the parameter from (\ref{eqLipschitz}).
This gives
\begin{align}
\int_{\widetilde{A}}
\big[\,
{p\mathstrut}_{Y}(y)
{\log\mathstrut}_{\!b}\,{p\mathstrut}_{Y}(y)
\, - \,
{p\mathstrut}_{Y}^{q}(y)
{\log\mathstrut}_{\!b}\,{p\mathstrut}_{Y}^{q}(y)
\big]
\,dy
\;\;
&
\overset{\widetilde{h}\, \leq \, 1/e}{\leq}
\;\;
\text{vol}(\widetilde{A})\,
\widetilde{h}\cdot
\max
\big\{
{\log\mathstrut}_{\!b}\,(1/\widetilde{h}),
\;
{\log\mathstrut}_{\!b} (e\sqrt{K})
\big\},
\label{eqVolh2}
\end{align}
where $\text{vol}(\widetilde{A})$ is the total length of $\widetilde{A}$.
It remains to find an upper bound on $\text{vol}(\widetilde{A})$. 
We use (\ref{eqJPDFVariance}):

\begin{align}
c_{Y}
\;\; \geq \;\;
\int_{\widetilde{A}}{p\mathstrut}_{Y}(y)y^{2}dy
\;\; 
\geq \;\;
\int_{\widetilde{A}}\widetilde{h} \cdot y^{2}dy
\;\; & = \;\;
\widetilde{h}\left[
\int_{\widetilde{A}\,\cap\,\widetilde{B}{\mathstrut}_{2}}y^{2}dy
\, +
\int_{\widetilde{A}\,\cap\,\widetilde{B}{\mathstrut}^{c}_{2}}y^{2}dy
\right]
\nonumber \\
& \overset{a}{\geq} \;\;
\widetilde{h}\left[
\int_{\widetilde{A}\,\cap\,\widetilde{B}{\mathstrut}_{2}}y^{2}dy
\, +
\int_{\widetilde{A}{\mathstrut}^{c}\,\cap\,\widetilde{B}{\mathstrut}_{2}}y^{2}dy
\right]
\nonumber \\
& = \;\;
\int_{\widetilde{B}{\mathstrut}_{2}}\widetilde{h} \cdot y^{2}dy
\;\; = \;\;
\frac{\widetilde{h}}{12}(\text{vol}(\widetilde{A}))^{3},
\nonumber
\end{align}
where in ($a$) we use the interval set
$\widetilde{B}{\mathstrut}_{2} \, \triangleq \, \big[
-\!\text{vol}(\widetilde{A})/2, \;
\text{vol}(\widetilde{A})/2
\,\big]$, centered around zero, and
of the same total length as $\widetilde{A}$ with the resulting property that
$\text{vol}(\widetilde{A}{\mathstrut}^{c}\cap \widetilde{B}{\mathstrut}_{2}) = \text{vol}(\widetilde{A}\cap \widetilde{B}{\mathstrut}^{c}_{2})$.
So that
\begin{equation} \label{eqVol2}
\widetilde{c}_{1}\widetilde{h}^{-1/3}
\;\; \geq \;\;
\text{vol}(\widetilde{A}),
\end{equation}
where $\widetilde{c}_{1} \triangleq (12c_{Y})^{1/3}$.
Continuing (\ref{eqVolh2}), with (\ref{eqVol2}) we obtain the following
upper bound on the second term in (\ref{eqByParts2}), which is by (\ref{eqLocalEntropy2}) also
an upper bound
on both terms of (\ref{eqByParts2}):
\begin{align}
&
\int_{\mathbb{R}}
{p\mathstrut}_{Y}(y)
{\log\mathstrut}_{\!b}\,{p\mathstrut}_{Y}(y)dy
\, - \,
\int_{\mathbb{R}}
{p\mathstrut}_{Y}^{q}(y)
{\log\mathstrut}_{\!b}\,{p\mathstrut}_{Y}^{q}(y)
dy
\;\;
\overset{\widetilde{h}\, \leq \, 1/e}{\leq}
\;\;
\widetilde{c}_{1}
\widetilde{h}^{2/3}
\max
\big\{
{\log\mathstrut}_{\!b}\,(1/\widetilde{h}),
\;
{\log\mathstrut}_{\!b} (e\sqrt{K})
\big\}.
\label{eqSecondTerm2}
\end{align}
So that if $\delta_{1} > 0$ in (\ref{eqYDifference}), then the possible decrease
caused by substitution of ${p\mathstrut}_{Y}^{q}$ in place of ${p\mathstrut}_{Y}$ on the LHS of (\ref{eqMargEntropy}) is at most $o(1)$.

\bigskip

{\em The LHS of (\ref{eqLogGaussian})}


Let us define two functions of $y \in \mathbb{R}$:
\begin{equation} \label{eqRelationship}
Q_{\beta}(y) \; \triangleq \; \Delta_{\beta,\,n}\cdot\lfloor y/\Delta_{\beta,\,n} + 1/2\rfloor,
\;\;\;\;\;\;
r_{\beta}(y) \; \triangleq \; y \, - \, Q_{\beta}(y).
\end{equation}
Then
with ${p\mathstrut}_{XY}^{q}$ defined in (\ref{eqQuantized}) we can obtain a lower bound for the expression on the LHS of (\ref{eqLogGaussian}):
\begin{align}
&
\Delta_{\alpha,\,n}
\Delta_{\beta,\,n}
\sum_{\substack{x\,\in\,{\cal X}_{n}\\
y\,\in\,{\cal Y}_{n}}}
{p\mathstrut}_{XY}^{q}(x, y)
(y - x)^{2}
\;\; = \;\;
\Delta_{\alpha,\,n}
\sum_{x\,\in\,{\cal X}_{n}}
\int_{\mathbb{R}}
{p\mathstrut}_{XY}^{q}(x, y)
\big(y - x - r_{\beta}(y)\big)^{2}dy
\nonumber \\
\overset{a}{\leq} \;\; &
\Delta_{\alpha,\,n}
\sum_{x\,\in\,{\cal X}_{n}}
\int_{\mathbb{R}}
{p\mathstrut}_{XY}(x, y)
(y - x)^{2}dy
\; + \;
\Delta_{\alpha,\,n}
\Delta_{\beta,\,n}^{2}
\sum_{x\,\in\,{\cal X}_{n}}
\int_{\mathbb{R}}
{p\mathstrut}_{XY}(x, y)dy
\nonumber \\
& \;\;\;\;\;\;\;\;\;\;\;\;\;\;\;\;\;\;\;\;\;\;\;\;\;\;\;\;\;\;\;\;\;\;\;\;\;\;\;\;\;\;\;\;\;\;\;\;\;\;\,\,\,
+ \;
\Delta_{\alpha,\,n}
\Delta_{\beta,\,n}
\sum_{x\,\in\,{\cal X}_{n}}
\int_{\mathbb{R}}
{p\mathstrut}_{XY}(x, y)
| y - x | dy
\nonumber \\
\overset{b}{\leq} \;\; &
\sum_{x\,\in\,{\cal X}_{n}}{P\mathstrut}_{\!X}(x)\int_{\mathbb{R}}{p\mathstrut}_{Y|X}(y\,|\,x)
(y - x)^{2}dy
\; + \;
\Delta_{\beta,\,n}^{2}
\; + \;
\Delta_{\beta,\,n}
\bigg[
\sum_{x\,\in\,{\cal X}_{n}}{P\mathstrut}_{\!X}(x)\int_{\mathbb{R}}{p\mathstrut}_{Y|X}(y\,|\,x)
(y - x)^{2} dy
\bigg]^{1/2}
\nonumber \\
\overset{c}{\leq} \;\; &
\sum_{x\,\in\,{\cal X}_{n}}{P\mathstrut}_{\!X}(x)\int_{\mathbb{R}}{p\mathstrut}_{Y|X}(y\,|\,x)
(y - x)^{2}dy
\; + \;
\underbrace{
\Delta_{\beta,\,n}^{2}
\; + \;
\Delta_{\beta,\,n}
\sqrt{c_{XY}}
}_{o(1)}
,
\label{eqLHS}
\end{align}
where ($a$) follows because ${p\mathstrut}_{XY}^{q}(x, y)\leq {p\mathstrut}_{XY}(x, y)$
and $|\,r_{\beta}(y)\,|\,\leq\, \Delta_{\beta,\,n}/2$, ($b$) follows by (\ref{eqJointPDF}) and Jensen's inequality
for the concave ($\cap$) function $f(t) = \sqrt{t}$,
and ($c$) follows by the condition of the lemma.

\bigskip

{\em Joint type ${P\mathstrut}_{\!XY}$}


Let us define two mutually-complementary probability masses for each $(x, y) \in {\cal X}_{n} \times {\cal Y}_{n}$:
\begin{align}
{P\mathstrut}_{\!XY}^{q}(x, y) \;\; & \triangleq \;\;
{p\mathstrut}_{XY}^{q}(x, y)\Delta_{\alpha,\,n}\Delta_{\beta,\,n},
\label{eqPDFtoType} \\
{P\mathstrut}_{\!XY}^{a}(x, y) \;\; & \triangleq \;\;
\Bigg\{
\begin{array}{r r}
{P\mathstrut}_{\!X}(x) -
\sum_{\tilde{y}\,\in\,{\cal Y}_{n}}{P\mathstrut}_{\!XY}^{q}(x, \tilde{y}), & \;\;\;
y =  Q_{\beta}(x), \\
0, & \;\;\; \text{o.w.},
\end{array}
\label{eqDiagonal} 
\end{align}
where $Q_{\beta}(\cdot)$ is defined in (\ref{eqRelationship}).
It follows from 
(\ref{eqQuantized}) and (\ref{eqCube}),
that each number ${P\mathstrut}_{\!XY}^{q}(x, y)$
is an integer multiple of $1/n$
and $\sum_{y\,\in\,{\cal Y}_{n}}{P\mathstrut}_{\!XY}^{q}(x, y) \leq {P\mathstrut}_{\!X}(x)$
for each $x \in {\cal X}_{n}$.
Then a joint type can be formed with the two definitions above:
\begin{equation} \label{eqJointType}
{P\mathstrut}_{\!XY}(x, y) \;\; \triangleq \;\; {P\mathstrut}_{\!XY}^{q}(x, y) \, + \, {P\mathstrut}_{\!XY}^{a}(x, y),
\;\;\;\;\;\; \forall (x, y) \in {\cal X}_{n}\times {\cal Y}_{n},
\end{equation}
such that ${P\mathstrut}_{\!XY} \in {\cal P}_{n}({\cal X}_{n}\times {\cal Y}_{n})$
and $\sum_{y\,\in\,{\cal Y}_{n}}{P\mathstrut}_{\!XY}(x, y) = {P\mathstrut}_{\!X}(x)$ for each $x \in {\cal X}_{n}$.

\bigskip

{\em The RHS of (\ref{eqLogGaussian})}


Having defined ${P\mathstrut}_{\!XY}$ and ${P\mathstrut}_{\!XY}^{q}$,
let us examine the possible decrease in the expression found on the RHS of (\ref{eqLogGaussian}) when 
${P\mathstrut}_{\!XY}$ inside that expression is replaced with ${P\mathstrut}_{\!XY}^{q}$:
\begin{align}
&
\sum_{\substack{x\,\in\,{\cal X}_{n}\\
y\,\in\,{\cal Y}_{n}}}
{P\mathstrut}_{\!XY}(x, y)
(y - x)^{2}
\; - \;
\sum_{\substack{x\,\in\,{\cal X}_{n}\\
y\,\in\,{\cal Y}_{n}}}
{P\mathstrut}_{\!XY}^{q}(x, y)
(y - x)^{2}
\nonumber \\
\overset{a}{=} \;\; &
\sum_{\substack{x\,\in\,{\cal X}_{n}\\
y\,\in\,{\cal Y}_{n}}}
{P\mathstrut}_{\!XY}^{a}(x, y)
(y - x)^{2}
\;\; \overset{b}{=} \;\;
\sum_{\substack{x\,\in\,{\cal X}_{n}\\
y\,\in\,{\cal Y}_{n}}}
{P\mathstrut}_{\!XY}^{a}(x, y)
r_{\beta}^{2}(x)
\;\;
\overset{c}{\leq} \;\; 
p_{1}\Delta_{\beta,\,n}^{2}/4
\;\; \overset{d}{\leq} \;\;
\underbrace{c_{1}h^{1/2}\Delta_{\beta,\,n}^{2}/2}_{o(1)},
\label{eqRHS}
\end{align}
where ($a$) follows by (\ref{eqJointType}),
($b$) follows according to the definitions (\ref{eqDiagonal}) and (\ref{eqRelationship}),
($c$) follows because $|\,r_{\beta}(x)\,| \, \leq \, \Delta_{\beta,\,n}/2$ and because
\begin{align}
\sum_{\substack{x\,\in\,{\cal X}_{n}\\
y\,\in\,{\cal Y}_{n}}}
{P\mathstrut}_{\!XY}^{a}(x, y)
\;\; & \overset{(\ref{eqJointType})}{=} \;\;
1 \, - \,
\sum_{\substack{x\,\in\,{\cal X}_{n}\\
y\,\in\,{\cal Y}_{n}}}
{P\mathstrut}_{\!XY}^{q}(x, y)
\;\; \overset{(\ref{eqPDFtoType})}{=} \;\;
1 \, - \,
\Delta_{\alpha,\,n}
\Delta_{\beta,\,n}
\sum_{\substack{x\,\in\,{\cal X}_{n}\\
y\,\in\,{\cal Y}_{n}}}
{p\mathstrut}_{XY}^{q}(x, y)
\nonumber \\
& \overset{(\ref{eqQuantized})}{=} \;\;
1 \, - \,
\int\!\!\!\!\int_{\mathbb{R}^{2}}{p\mathstrut}_{XY}^{q}(x, y)dxdy
\;\; \overset{(\ref{eqProbLoss})}{=} \;\; p_{1},
\label{eqLostProb}
\end{align}
then ($d$) follows by the upper bound on $p_{1}$ of (\ref{eqProbLoss}).
Since by definition (\ref{eqPDFtoType})
we also have
\begin{displaymath}
\sum_{\substack{x\,\in\,{\cal X}_{n}\\
y\,\in\,{\cal Y}_{n}}}
{P\mathstrut}_{\!XY}^{q}(x, y)
(y - x)^{2} \;\; 
=
\;\;
\Delta_{\alpha,\,n}
\Delta_{\beta,\,n}
\sum_{\substack{x\,\in\,{\cal X}_{n}\\
y\,\in\,{\cal Y}_{n}}}
{p\mathstrut}_{XY}^{q}(x, y)
(y - x)^{2},
\end{displaymath}
which is exactly the beginning of (\ref{eqLHS}),
then combining (\ref{eqLHS}) and (\ref{eqRHS}) we obtain (\ref{eqLogGaussian}).
The remainder of the proof for (\ref{eqCondEntropy}) and (\ref{eqMargEntropy}) will easily follow by Lemma~\ref{Lemxlogx}
applied to corresponding discrete entropy expressions with probability masses.

\bigskip

{\em The RHS of (\ref{eqCondEntropy})}


In order to upper-bound the expression on the RHS of (\ref{eqCondEntropy}), it is convenient to write:
\begin{align}
&
\sum_{\substack{x\,\in\,{\cal X}_{n}\\
y\,\in\,{\cal Y}_{n}}}
{P\mathstrut}_{\!XY}(x, y)\,{\log\mathstrut}_{\!b}\,\frac{{P\mathstrut}_{\!XY}(x, y)}
{\Delta_{\alpha,\,n}
\Delta_{\beta,\,n}}
\, - \,
\sum_{\substack{x\,\in\,{\cal X}_{n}\\
y\,\in\,{\cal Y}_{n}}}
{P\mathstrut}_{\!XY}^{q}(x, y)\,{\log\mathstrut}_{\!b}\,\frac{{P\mathstrut}_{\!XY}^{q}(x, y)}
{\Delta_{\alpha,\,n}
\Delta_{\beta,\,n}}
\nonumber \\
\overset{a}{=} \;\; &
\sum_{\substack{x\,\in\,{\cal X}_{n}\\
y\,\in\,{\cal Y}_{n}}}
\Big[\,
{P\mathstrut}_{\!XY}(x, y)\,{\log\mathstrut}_{\!b}\,{P\mathstrut}_{\!XY}(x, y)
\, - \,
{P\mathstrut}_{\!XY}^{q}(x, y)\,{\log\mathstrut}_{\!b}\,{P\mathstrut}_{\!XY}^{q}(x, y)
\,\Big]
\; + \;
\sum_{\substack{x\,\in\,{\cal X}_{n}\\
y\,\in\,{\cal Y}_{n}}}
{P\mathstrut}_{\!XY}^{a}(x, y)\,{\log\mathstrut}_{\!b}\,\frac{1}{\Delta_{\alpha,\,n}
\Delta_{\beta,\,n}}
\nonumber \\
\overset{b}{\leq} \;\; &
\sum_{\substack{x\,\in\,{\cal X}_{n}\\
y\,\in\,{\cal Y}_{n}}}
\max\Big\{\!
-{P\mathstrut}_{\!XY}^{a}(x, y)\,{\log\mathstrut}_{\!b}\,{P\mathstrut}_{\!XY}^{a}(x, y), \;\;
{P\mathstrut}_{\!XY}^{a}(x, y)\,{\log\mathstrut}_{\!b}\,e
\Big\}
\; + \;
(1-\gamma)p_{1}\,{\log\mathstrut}_{\!b}\,n
\nonumber \\
\overset{c}{\leq} \;\; &
\sum_{\substack{x\,\in\,{\cal X}_{n}\\
y\,\in\,{\cal Y}_{n}}}
{P\mathstrut}_{\!XY}^{a}(x, y)\,{\log\mathstrut}_{\!b}\,n
\; + \;
(1-\gamma)p_{1}\,{\log\mathstrut}_{\!b}\,n
\;\; \overset{d}{=} \;\; (2-\gamma)p_{1}
\,{\log\mathstrut}_{\!b}\,n
\;\; \overset{e}{\leq} \;\;
\underbrace{4c_{1}h^{1/2}\,{\log\mathstrut}_{\!b}\,n}_{o(1)}
, \;\;\; n > 2,
\label{eqDiscreteEntropy}
\end{align}
where ($a$) follows by (\ref{eqJointType});
in ($b$) we use (\ref{eqLostProb}), (\ref{eqDelta}),
and apply Lemma~\ref{Lemxlogx} with its parameters $t_{1} = {P\mathstrut}_{\!XY}^{q}(x, y)$
and $t_{1} + t = {P\mathstrut}_{\!XY}(x, y) \leq 1$
with $t = {P\mathstrut}_{\!XY}^{a}(x, y)$ by (\ref{eqJointType});
($c$) follows for $n > 2$ since ${P\mathstrut}_{\!XY}^{a}(x, y) \geq 1/n$ when positive;
($d$) and ($e$) follow respectively by (\ref{eqLostProb}) and (\ref{eqProbLoss}).
Now since
\begin{displaymath}
\sum_{\substack{x\,\in\,{\cal X}_{n}\\
y\,\in\,{\cal Y}_{n}}}
{P\mathstrut}_{\!XY}^{q}(x, y)\,{\log\mathstrut}_{\!b}\,\frac{{P\mathstrut}_{\!XY}^{q}(x, y)}
{\Delta_{\alpha,\,n}
\Delta_{\beta,\,n}}
\;\; = \;\;
\int\!\!\!\!\int_{\mathbb{R}^{2}}{p\mathstrut}_{XY}^{q}(x, y)
{\log\mathstrut}_{\!b}\,{p\mathstrut}_{XY}^{q}(x, y)dxdy,
\end{displaymath}
the inequality in (\ref{eqCondEntropy}) follows by comparing (\ref{eqCondEntIdeal}), (\ref{eqPutting}), and (\ref{eqDiscreteEntropy}).

\bigskip

{\em The RHS of (\ref{eqMargEntropy})}


With ${P\mathstrut}_{\!Y}^{q}(y) \triangleq \sum_{x\,\in \,{\cal X}_{n}}{P\mathstrut}_{\!XY}^{q}(x, y)$
and ${P\mathstrut}_{\!Y}^{a}(y) \triangleq \sum_{x\,\in \,{\cal X}_{n}}{P\mathstrut}_{\!XY}^{a}(x, y)$
we have
\begin{align}
&
\sum_{y\,\in\,{\cal Y}_{n}}
{P\mathstrut}_{\!Y}(y)\,{\log\mathstrut}_{\!b}\,\frac{{P\mathstrut}_{\!Y}(y)}
{\Delta_{\beta,\,n}}
\, - \,
\sum_{y\,\in\,{\cal Y}_{n}}
{P\mathstrut}_{\!Y}^{q}(y)\,{\log\mathstrut}_{\!b}\,\frac{{P\mathstrut}_{\!Y}^{q}(y)}
{\Delta_{\beta,\,n}}
\nonumber \\
\overset{a}{=} \;\; &
\sum_{y\,\in\,{\cal Y}_{n}}
\Big[\,
{P\mathstrut}_{\!Y}(y)\,{\log\mathstrut}_{\!b}\,{P\mathstrut}_{\!Y}(y)
\, - \,
{P\mathstrut}_{\!Y}^{q}(y)\,{\log\mathstrut}_{\!b}\,{P\mathstrut}_{\!Y}^{q}(y)
\,\Big]
\; + \;
\sum_{y\,\in\,{\cal Y}_{n}}
{P\mathstrut}_{\!Y}^{a}(y)\,{\log\mathstrut}_{\!b}\,\frac{1}{\Delta_{\beta,\,n}}
\nonumber \\
\overset{b}{\geq} \;\; &
\sum_{y\,\in\,{\cal Y}_{n}}
{P\mathstrut}_{\!Y}^{a}(y)\,{\log\mathstrut}_{\!b}\,{P\mathstrut}_{\!Y}^{a}(y)
\; + \;
\beta p_{1}\,{\log\mathstrut}_{\!b}\,n
\nonumber \\
\overset{c}{\geq} \;\; &
\sum_{y\,\in\,{\cal Y}_{n}}
{P\mathstrut}_{\!Y}^{a}(y)\,{\log\mathstrut}_{\!b}\,(1/n)
\; + \;
\beta p_{1}\,{\log\mathstrut}_{\!b}\,n
\;\; \overset{d}{=} \;\; -(1 - \beta)p_{1}
\,{\log\mathstrut}_{\!b}\,n
\;\; \overset{e}{\geq} \;\;
\underbrace{-2c_{1}h^{1/2}\,{\log\mathstrut}_{\!b}\,n}_{o(1)},
\label{eqDiscreteMargEntropy}
\end{align}
where ($a$) follows by (\ref{eqJointType});
in ($b$) we use (\ref{eqLostProb}), (\ref{eqDelta}),
and apply Lemma~\ref{Lemxlogx} with its parameters $t_{1} = {P\mathstrut}_{\!Y}^{q}(y)$
and $t_{1} + t = {P\mathstrut}_{\!Y}(y)$
with $t = {P\mathstrut}_{\!Y}^{a}(y)$ by (\ref{eqJointType});
($c$) follows because ${P\mathstrut}_{\!Y}^{a}(y) \geq 1/n$ whenever positive;
($d$) and ($e$) follow respectively by (\ref{eqLostProb}) and (\ref{eqProbLoss}).
From (\ref{eqMargUsual}) and (\ref{eqPDFtoType}) we observe that ${P\mathstrut}_{\!Y}^{q}(y) = {p\mathstrut}_{Y}^{q}(y)\Delta_{\beta,\,n}$.
Since the function ${p\mathstrut}_{Y}^{q}(y)$
is piecewise constant in $\mathbb{R}$ by the definition of ${p\mathstrut}_{XY}^{q}$,
it follows that
\begin{displaymath}
\sum_{y\,\in\,{\cal Y}_{n}}
{P\mathstrut}_{\!Y}^{q}(y)\,{\log\mathstrut}_{\!b}\,\frac{{P\mathstrut}_{\!Y}^{q}(y)}
{\Delta_{\beta,\,n}}
\;\; = \;\;
\int_{\mathbb{R}}
{p\mathstrut}_{Y}^{q}(y)
{\log\mathstrut}_{\!b}\,{p\mathstrut}_{Y}^{q}(y)dy.
\end{displaymath}
Then the inequality (\ref{eqMargEntropy}) follows by comparing (\ref{eqSecondTerm2}),
(\ref{eqDiscreteMargEntropy}).
This concludes the proof of Lemma~\ref{LemQuant}.
$\square$

\bigskip

\begin{lemma}\label{Lemxlogx}
{\em Let $f(x) = x\ln x\,$, then for $0 < t \leq 1/e$ and $t_{1} > 0\,$,}
\begin{displaymath}
t \ln t
\;\; \leq \;\;
f(t_{1} + t) - f(t_{1})
\;\; \leq \;\; t\,\ln \max
\big\{
1/t, \; (t_{1} + t)e
\big\}.
\end{displaymath}
\end{lemma}


\begin{proof}
For $t \leq 1/e$,
the magnitude of the derivative 
of $f(x)$ in the interval $(0, t)$ is monotonically decreasing and its average there
is $-\ln t$,
while in the adjacent interval $[\,t , \,1/(et)\,]$ the magnitude of the derivative is upper-bounded by $-\ln t$:
$\;\;\; |\, f'(x)\,| \; \leq \; -\ln t,
\;\;\; t \leq x \leq 1/(et)$.

Then there are three possible cases:

1) $\; t_{1} \, < \, t \, < \, t_{1} + t \, < \, 1/(et)\,$:
\begin{align}
|\,f(t_{1} + t) - f(t_{1})\,|
\;\; & \leq \;\;
|\,f(t_{1} + t) - f(t)\,| \, + \,
|\,f(t_{1}) - f(t)\,|
\nonumber \\
& = \;\;
|\,f(t_{1} + t) - f(t)\,| \, + \,
t_{1} \ln t_{1} \, - \,
t \ln t
\nonumber \\
&
\leq \;\;
-t_{1} \ln t
\, + \, t_{1} \ln t_{1}
\, - \, t \ln t
\nonumber \\
& \leq \;\;
-t \ln t.
\nonumber
\end{align}

2) $\; t \, \leq \, t_{1} \, < \, t_{1} + t \, \leq \, 1/(et)\,$:
$\;\;\;\;\;\;\;\;
 |\,f(t_{1} + t) - f(t_{1})\,| \;\; \leq \;\;
-t \ln t$.

3) $\; 1/(et)\, < \, t_{1} + t\,$:
$\;\;\;\;\;\;\;\;\;\;\;\;\;\;\;
0 \;\; \leq \;\;
f(t_{1} + t) - f(t_{1})
\;\; \leq \;\;
f'(t_{1} + t)
\, t
\;\; = \;\;
\big(\ln(t_{1} + t) + 1\big) \, t$.
\end{proof}


\bibliographystyle{IEEEtran}

\end{document}